\documentclass[11pt]{article}
\usepackage{amsmath}
\usepackage{amsthm}
\usepackage{amssymb}
\usepackage{bbm}
\usepackage{mathtools}
\usepackage{dsfont}
\usepackage{color}
\usepackage{tikz}
\usepackage{tabularx}
\usepackage{algorithm}
\usepackage{algpseudocode}
\usepackage[artemisia]{textgreek}

\usepackage{mathtools}

\mathtoolsset{showonlyrefs}

\usepackage{constants}
\usepackage{wrapfig}

\usepackage{hyperref} 
\hypersetup{colorlinks=true,citecolor=red,linkcolor=blue}  
\usepackage[margin=0.99in]{geometry}

\title{Learning Mixtures of Linear Regressions in Subexponential Time via Fourier Moments}
\newcommand{\calD}{\mathcal{D}}
\newcommand{\calE}{\mathcal{E}}
\newcommand{\calF}{\mathcal{F}}
\newcommand{\calG}{\mathcal{G}}
\newcommand{\svd}{\text{svd}}
\newcommand{\M}{\mathcal{M}}
\newcommand{\delsamp}{\delta_{\mathrm{samp}}}
\newcommand{\poly}{\mathrm{poly}}
\newcommand{\noise}{\varsigma}
\newcommand{\Span}{\mathrm{span}}
\newcommand{\deriv}[1]{\frac{\partial}{\partial{#1}}}
\newcommand{\minvar}[1]{\sigma_{\min}(#1)}
\newcommand{\maxvar}[1]{\sigma_{\max}(#1)}
\newcommand{\E}{\mathop{\mathbb{E}}}
\renewcommand{\d}{{\mathrm{d}}} 
\renewcommand{\i}{{\mathbf{i}}} 
\newcommand{\N}{{\mathcal{N}}}  
\newcommand{\argmin}{\text{argmin}}

\newcommand{\Id}{\vec{I}}
\newcommand{\R}{{\mathbb{R}}} 

\renewcommand{\S}{\mathbb{S}}

\newcommand{\tvd}{d_{\mathrm{TV}}}
\newcommand{\bone}[1]{\mathds{1}\left[#1\right]}
\renewcommand{\vec}[1]{\mathbf{#1}}

\renewcommand{\tilde}{\widetilde}
\renewcommand{\hat}{\widehat}

\newcommand{\norm}[1]{\left\lVert#1\right\rVert}
\author{Sitan Chen\thanks{This work was supported in part by a Paul and Daisy Soros Fellowship, NSF CAREER Award CCF-1453261, and NSF Large CCF-1565235. This work was done in part while S.C. was an intern at Microsoft Research AI.} \\
\texttt{sitanc@mit.edu}\\
MIT 
\and Jerry Li \\
\texttt{jerrl@microsoft.com} \\
Microsoft Research AI
\and Zhao Song \\
\texttt{zhaos@ias.edu} \\
Princeton \& IAS
}
\newtheorem{Alg}{Algorithm}

\newcommand{\pmin}{p_{\mathrm{min}}}

\newcommand{\smax}{\sigma_{\mathrm{max}}}
\newcommand{\gap}{\mathrm{gap}}
\newcommand{\calO}{\mathcal{O}}
\newcommand{\sign}{\text{sign}}

\newcommand\encircle[1]{%
  \tikz[baseline=(X.base)] 
    \node (X) [draw, shape=circle, inner sep=0] {\strut #1};}
\renewcommand{\Re}{\text{Re}}
\newconstantfamily{c}{symbol=c}
\newconstantfamily{C}{symbol=C}

\newtheorem{theorem}{Theorem}[section]
\newtheorem{lemma}[theorem]{Lemma}
\newtheorem{definition}[theorem]{Definition}

\newtheorem{corollary}[theorem]{Corollary}

\newtheorem{observation}[theorem]{Observation}
\newtheorem{fact}[theorem]{Fact}
\newtheorem{remark}[theorem]{Remark}
\newtheorem{claim}[theorem]{Claim}

\newtheorem{case}{Case}

\makeatletter
\newcommand{\multiline}[1]{%
  \begin{tabularx}{\dimexpr\linewidth-\ALG@thistlm}[t]{@{}X@{}}
    #1
  \end{tabularx}
}
\makeatother

\newcommand{\algorithmicbreak}{\textbf{break}}
\newcommand{\Break}{\State \algorithmicbreak}

\begin{document}

\begin{titlepage}
\maketitle
\begin{abstract}

We consider the problem of learning a mixture of linear regressions (MLRs).
An MLR is specified by $k$ nonnegative mixing weights $p_1, \ldots, p_k$ summing to $1$, and $k$ unknown regressors $w_1,...,w_k\in\R^d$.
A sample from the MLR is drawn by sampling $i$ with probability $p_i$, then outputting $(x, y)$ where $y = \langle x, w_i \rangle + \eta$, where $\eta\sim\N(0,\noise^2)$ for noise rate $\noise$. 
Mixtures of linear regressions are a popular generative model and have been studied extensively in machine learning and theoretical computer science.
However, all previous algorithms for learning the parameters of an MLR require running time and sample complexity scaling exponentially with $k$.

In this paper, we give the first algorithm for learning an MLR that runs in time which is sub-exponential in $k$.
Specifically, we give an algorithm which runs in time $\widetilde{O}(d)\cdot\exp(\widetilde{O}(\sqrt{k}))$ and outputs the parameters of the MLR to high accuracy, even in the presence of nontrivial regression noise.
We demonstrate a new method that we call \emph{Fourier moment descent} which uses univariate density estimation and low-degree moments of the Fourier transform of suitable univariate projections of the MLR to iteratively refine our estimate of the parameters.
To the best of our knowledge, these techniques have never been used in the context of high dimensional distribution learning, and may be of independent interest.
We also show that our techniques can be used to give a sub-exponential time algorithm for a natural hard instance of the \emph{subspace clustering} problem, which we call \emph{learning mixtures of hyperplanes}.

%
%
%

\end{abstract}
 \thispagestyle{empty}
\end{titlepage}


\section{Introduction}

Mixtures of linear regressions (or MLRs for short) are a popular generative model, and have been studied extensively in machine learning and theoretical computer science.
A standard formulation of the problem is as follows: we have $k$ unknown \emph{mixing weights} $p_1, \ldots, p_k$ which are non-negative and sum to $1$, $k$ unknown \emph{regressors} $w_1, \ldots, w_k \in \R^d$, and a \emph{noise rate} $\noise \ge 0$.
A sample from the MLR is drawn as follows: we first select $i \in [k]$ with probability $p_i$, then we receive $(x, y)$ where $x \in \R^d$ is distributed as $\N (0, \Id)$ and \begin{equation}y = \langle w_i,x\rangle + \eta,\end{equation} where $\eta \sim \N (0, \noise^2)$.
This model has applications to problems ranging from trajectory clustering~\cite{gaffney1999trajectory} to phase retrieval~\cite{balan2006signal,candes2013phaselift,netrapalli2013phase} and is also widely studied as a natural non-linear generative model for supervised data~\cite{faria2010fitting,chen2013convex,chaganty2013spectral,yi2014alternating,yi2016solving,zhong2016mixed,sedghi2016provable,klusowski2017estimating,balakrishnan2017statistical,kwon2018global,li2018learning,kwon2019converges}.

The basic learning question for MLRs is as follows: given i.i.d. samples $(x_1, y_1), \ldots, (x_n ,y_n) \in \R^d \times \R$ from an unknown MLR, can we learn the parameters of the underlying MLR?
To ensure that the parameters are identifiable, it is also typically assumed that the regressors are separated in some way, e.g. there is some $\Delta > 0$ so that $\left\| w_i - w_j \right\|_2 \geq \Delta$ for all $i \neq j$.

Despite the apparent simplicity of the problem, efficiently learning MLRs given samples has proven to be a surprisingly challenging task.
Even in the special case where $\noise = 0$, that is, we assume that there is no noise on the samples, the fastest algorithms for this problem run in time depending on $k^{\Omega(k)}$~\cite{li2018learning,zhong2016mixed}.
It turns out that there are good reasons for this barrier.

Previous algorithms for this problem with end-to-end provable guarantees---and indeed, the vast majority of statistical learning algorithms in general---build in some form or another on the \emph{method of moments} paradigm.
At a high level, these methods require that there exists some statistic which depends only on low degree moments of the unknown distribution, so that a sufficiently good estimate of this statistic will uniquely identify the parameters of the distribution.
This includes widely-used techniques based on tensor decomposition \cite{chaganty2013spectral,yi2016solving,zhong2016mixed,sedghi2016provable}, and SDP hierarchies such as the Sum-of-Squares meta-algorithm \cite{karmalkar2019list,raghavendra2019list}.
If degree $t$ moments are necessary to devise such a statistic, then these methods require $\exp (\Omega (t))$ sample and computational complexity.

Unfortunately, for MLRs, it is not hard to demonstrate pairs of mixtures some of whose parameters are far apart from each other, where all moments of degree at most $2k-1$ of the two mixtures agree exactly (see Appendix~\ref{subsec:failure} for more details).
As a result, any moment-based estimator would need to use moments of degree at least $\Omega (k)$, and hence require a runtime of $\exp (\Omega (k))$.
This imposes a natural bottleneck: any algorithm that hopes to achieve sub-exponential time must somehow incorporate additional information about the geometry of the underlying learning problem.

A related problem, which shares a similar bottleneck, is the problem of learning mixtures of Gaussians under the assumption of \emph{angular separation}.
A concrete instantiation of this problem is a model we call \emph{learning mixtures of hyperplanes}.
A mixture of hyperplanes is parameterized by mixing weights $p_1, \ldots, p_k$, a separation parameter $\Delta > 0$, and $k$ unit vectors $v_1, \ldots, v_k$ satisfying $\| v_i \pm v_j \|_2 \geq \Delta$ for all $i \neq j$ (note that the reason for the $\pm $ is that the directions of a mixture of hyperplanes are only identifiable up to sign).
To draw a sample, we first draw $i \in [k]$ with probability $p_i$, and then draw a sample from $\N (0, \Id - v_i v_i^\top)$.

As before, the corresponding learning question is the following: given samples from an unknown mixture of hyperplanes, can one recover the underlying parameters?
This problem can be thought of as a particularly hard case of the well-studied problem of \emph{subspace recovery}, where current techniques would require time which is exponential in $k$.

In this paper, we give algorithms which are able to achieve strong recovery guarantees for the problems of learning MLRs and learning mixtures of hyperplanes, and which run in time which is sub-exponential in $k$.
To the best of our knowledge, this is the first algorithm for the basic problem of learning MLRs which achieves sub-exponential runtime without placing strong additional assumptions on the model.
At a high level, our key insight is that while low degree moments of the MLR are unable to robustly identify the instance, low degree moments of suitable projections of the \emph{Fourier transform} of the MLR can be utilized to extract non-trivial information about the regressors.
We then give efficient algorithms for computing such ``Fourier moments'' by leveraging algorithms for univariate density estimation~\cite{chan2014efficient,acharya2017sample}.
This allows us to dramatically improve the runtime and sample complexity of the \emph{moment descent} algorithm of~\cite{li2018learning}, and allows us to obtain our desired sub-exponential runtime.
We believe that this sort of algorithmic application of the continuous Fourier transform and of univariate density estimation to a high dimensional learning problem is novel, and may be of independent interest.

\subsection{Our Contributions}
Here, we describe our contributions in more detail.
For simplicity of exposition, in this section we will assume that the mixing weights are uniform, i.e. $p_i = 1/k$ for all $i \in [k]$, although as we show, our algorithms can handle non-uniform mixing weights.

Our main results for learning MLRs are twofold.
Throughout the paper we let $\widetilde{O}(f) = O(f \log^c (f))$ for some universal constant $c$.
First, in the well-studied case where there is no regression noise, we show:
\begin{theorem}[Informal, see Theorem~\ref{thm:learnwithoutnoise_main}]
\label{thm:informal-noiseless}
Assume that the noise rate $\noise = 0$. 
Let $w_1, \ldots, w_k \in \R^d$ be the parameters of an unknown MLR $\mathcal{D}$ with separation $\Delta$.
Then, there is an algorithm which takes $N =  \widetilde{O}( d ) \cdot \exp(\tilde{O}(\sqrt{k}) )$ samples from $\mathcal{D}$, runs in time $\widetilde{O}(N\cdot d)$, and outputs $\tilde{w}_1, \ldots, \tilde{w}_k \in \R^d$ so that with high probability,
there exists some permutation $\pi: [k] \to [k]$ satisfying
\[
\left\| w_i - \tilde{w}_{\pi(i)} \right\|_2 \leq \frac{\Delta}{k^{100}} \; , \forall i \in [k].
\]
\end{theorem}
\noindent
By combining this ``warm start'' with the boosting result of~\cite{li2018learning}, we can also obtain arbitrarily good accuracy with minimal overhead in both the sample complexity and runtime.
See Section~\ref{sec:allcomps_nonoise} for more details.

Secondly, in the case when the noise rate $\noise$ is large, we can also obtain a similar result, though with an additional exponential dependence on $\Delta$:
\begin{theorem}[Informal, see Theorem~\ref{thm:learnwithnoise_main}]
\label{thm:informal-noisy}
Let $w_1, \ldots, w_k \in \R^d$ be the parameters of an unknown MLR $\mathcal{D}$ with separation $\Delta$, and noise rate $\noise > 0$.
Then, there is an algorithm which takes $N =  \widetilde{O}( d ) \cdot \exp (\tilde{O} (\sqrt{k}/\Delta^2 ) )$ samples from $\mathcal{D}$, runs in time $\widetilde{O}(N\cdot d)$, and outputs $\widehat{w}_1, \ldots, \widehat{w}_k \in \R^d$ so that with high probability,
there exists some permutation $\pi: [k] \to [k]$ satisfying
\[
\left\| w_i - \widehat{w}_{\pi(i)} \right\|_2 \leq \frac{\Delta}{k^{100}} + O(\noise) \; , \forall i \in [k].
\]
\end{theorem}
\noindent
In particular, if $\Delta = \Omega (1)$, we again attain runtime which is sub-exponential in $k$.
In the special case when the mixing weights are all known, and assuming that $\noise = O(\tfrac{\Delta}{k^2 \mathrm{polylog} (k)} )$, by combining this result with the local convergence result of~\cite{kwon2019converges}, we can again attain arbitrarily good accuracy by slightly increasing the runtime; see Section~\ref{subsec:boost_noisy} and Theorem~\ref{thm:learnwithnoise_plus} for details.

Finally, for the problem of learning mixtures of hyperplanes, we are able to obtain qualitatively similar results.
Again, for simplicity of exposition, we assume the mixing weights are uniform just in the current section.
We obtain:
\begin{theorem}[Informal, see Theorem~\ref{thm:hyperplanes_main}]
\label{thm:informal-pancakes}
Let $\epsilon > 0$, and let $v_1, \ldots, v_k \in \R^d$ be the parameters of a mixture of hyperplanes $\mathcal{D}$ with separation $\Delta > 0$.
Then, there is an algorithm which takes $N =  \widetilde{O}( d ) \cdot \exp(\tilde{O}(k^{0.6}) )$ samples from $\mathcal{D}$, runs in time $\widetilde{O}(N\cdot d)$, and which outputs $\widehat{v}_1, \ldots, \widehat{v}_k \in \R^d$ so that with high probability, there is a permutation $\pi: [k] \to [k]$ so that
\[
\left\| v_i -\widehat{v}_{\pi(i)} \right\|_2 \leq \frac{\Delta}{k^{100}}  \; , \forall i \in [k].
\]
\end{theorem}

\subsection{Related Work}
\label{subsec:relatedwork}
Mixtures of linear regressions were introduced in~\cite{de1989mixtures}, and later by~\cite{jordan1994hierarchical}, under the name of \emph{hierarchical mixtures of experts}, and have been studied extensively in the theory and ML communities ever since.
Previous work on the problem with provable guarantees can roughly speaking be divided into three groups.
Some of the previous work focuses on special cases of the problem, in particular, when the number of components is small~\cite{chen2013convex,klusowski2017estimating,balakrishnan2017statistical,kwon2018global}.
In contrast, we focus on the setting where $k$ is quite large, which is the setting which is typically true in applications, but is also much more algorithmically complicated.

Another line of work has focused on demonstrating local convergence guarantees for non-convex methods such as expectation maximization or alternating minimization~\cite{faria2010fitting,yi2014alternating,yi2016solving,zhong2016mixed,klusowski2017estimating,balakrishnan2017statistical,kwon2018global,li2018learning,kwon2019converges}.
These papers demonstrate that given a sufficiently good warm start, non-convex methods are able to boost this warm start to arbitrarily good accuracy.
These results should be viewed as largely complementary to our results, as our main result is a method which is able to provably achieve a good warm start. That said, we also demonstrate new algorithms for learning given a warm start that work under a weaker initialization and can tolerate more regression noise than was previously known in the literature.

The final class of results use moment-based methods to learn MLRs. Here, the literature has focused largely on the case of $\noise = 0$ and spherical covariates, that is, covariates all drawn from $\N(0,\Id)$.\footnote{To the best of our knowledge, the primary exception to this is \cite{li2018learning}, which considered noise-less MLRs whose components' covariates are drawn from arbitrary unknown Gaussians satisfying some condition number bounds and obtained a $d\cdot \exp(k^2)$ algorithm in this setting.}
A line of work has studied tensor decomposition-based methods~\cite{chaganty2013spectral,yi2016solving,zhong2016mixed,sedghi2016provable}.
However, these require additional non-degeneracy conditions on the MLR instance beyond separation.
Indeed, as we argued in the Introduction, and more formally in Appendix~\ref{subsec:failure}, moment based methods cannot obtain runtime which is sub-exponential in $k$.
The work that is closest to ours, and that we build off of, is that of~\cite{li2018learning}, which demonstrates an algorithm which runs in $2^{\tilde{O} (k)}$ for learning a MLR under separation conditions.
However, as their warm start algorithm is ultimately moment based, since it interacts through the samples through the moment-based univariate GMM learning algorithm of~\cite{mv}, it cannot achieve runtime sub-exponential in $k$.

\paragraph{List-Decodable Regression} A related problem to---indeed, a generalization of---the problem of learning MLRs is that of \emph{list-decodable regression}~\cite{charikar2017learning,karmalkar2019list,raghavendra2019list}.
Here, we assume that we are given a set of data points $(x_1, y_1), \ldots, (x_n, y_n)$, where an $\alpha$-fraction of them come from an unknown linear regression $y_i = \langle w, x_i \rangle + \eta$, where $x_i$ is Gaussian, $\eta$ is Gaussian noise, and $\alpha < 1/2$.
The goal is then to recover a list of $O(1 / \alpha)$ possible $w_1, \ldots, w_{O(1 / \alpha)} \in \R^d$ so that $\| w_i - w \|_2$ is small for at least one element in the list.

It is not hard to see that given a uniform MLR instance, if we feed it into an algorithm for list-decodable regression, the list must contain something which is close to each of the regressors in the MLR instance, as each mixture component is an equally valid solution to the list-decodable regression problem.
Thus one could hope that these algorithms for list-decodable regression could yield improved algorithms for learning MLRs as well.

Unfortunately, all known techniques, including the state of the art~\cite{karmalkar2019list,raghavendra2019list}, either are too weak to be applied to our setting, or use the Sum-of-Squares SDP hierarchy and again interact through the data via estimating high-degree moments of the distribution.
As a result, these latter algorithms still suffer runtimes which are exponential in $k$.

\paragraph{Subspace Clustering} The mixtures of hyperplanes problem we consider in this paper can be thought of as a special case of the \emph{subspace clustering} or \emph{hyperplane clustering} problem, where data is thought of as being drawn from a union of linear subspaces.
In our problem, we additionally assume that the data is Gaussian within each subspace.
The literature on subspace clusterings is vast and we cannot do it justice here; see~\cite{parsons2004subspace,vidal2011subspace,elhamifar2013sparse} and references therein for a more complete treatment.
On the one hand, the mixture of hyperplanes problem arises naturally in practical contexts of projective motion segmentation \cite{vidal2007three} and hybrid system identification \cite{bako2011identification}.
On the other, it also corresponds to a challenging setting of the problem due to the low codimensionality of the subspaces. Indeed, essentially all algorithms for subspace clustering with provable guarantees either run in time exponential in the dimension of the subspaces (e.g. RANSAC~\cite{fischler1981random}, algebraic subspace clustering~\cite{vidal2005generalized}, spectral curvature clustering~\cite{liu2012robust}) or require the codimension to be at least some small but constant \emph{fraction} of the ambient dimension \cite{elhamifar2013sparse,candes2013phaselift,lu2012robust,tsakiris2015dual}. To our knowledge the only work which addresses the codimension 1 case is \cite{tsakiris2017hyperplane}, though their setting and guarantees are quite different from ours.

\paragraph{The Fourier Transform in Distribution Learning}
One of our main algorithmic tools will be the univariate (continuous) Fourier transform, as a way to estimate Fourier moments of our distribution.
In recent years, the question of learning the Fourier transform of a function has attracted a considerable amount of interest in theoretical computer science~\cite{hikp12a,ik14,m15,ps15,k16,ckps16,k17,nsw19}.
Our application is somewhat different in that we have explicit access to the function we will take the Fourier transform of.

In the context of distribution learning, the discrete Fourier transform has been used to learn families of distributions such as sums of independent integer random variables~\cite{diakonikolas2016optimal}, Poisson Binomial distributions~\cite{diakonikolas2016properly}, and Poisson multinomial distributions~\cite{diakonikolas2016optimal,diakonikolas2016fourier}.
These algorithms typically work by exploiting Fourier sparsity of the underlying distribution.
However, the way we use the Fourier transform is quite different: we only use it to compute different statistics of the data, namely, the Fourier moments of our distribution.

\paragraph{Univariate Density Estimation} 
Another important algorithmic primitive we use is univariate density estimation and specifically, the piecewise polynomial-based estimators given in~\cite{chan2014efficient,acharya2017sample}.
Univariate density estimation has a long history in statistics, ML, and theoretical computer science, and a full literature review of the field is out of the scope of this paper; see e.g.~\cite{diakonikolas2016learning} for a more comprehensive overview of the literature.
However, to the best of our knowledge, there are few previous cases where univariate density estimation has been used as a key tool for a high dimensional learning task.

\section{Preliminaries}

In this section, we give some basic technical preliminaries.

\subsection{Probabilistic Models}
In this section, we formally define the models we consider throughout this paper, namely, mixtures of linear regression and hyperplanes, and some important parameters for these models:
\paragraph{Mixtures of Linear Regressions}
We start with MLRs:
\begin{definition}[Mixtures of Linear Regressions]
	Given mixing weights $p \in [0,1]^k$ with $\sum_{i=1}^k p_i = 1$,  regressors $w_1, \cdots, w_k \in \R^{d}$, and noise rate $\noise \geq 0$, the corresponding \emph{mixture of spherical linear regressions} (or simply mixture of linear regressions) is the distribution over pairs $(x,y) \in \R^d \times \R$ where $x\sim\N (0,\Id)$ and $y = \langle w_i, x \rangle + g$ for $w_i$ sampled with probability $p_i$ and $g \sim \N(0, \noise^2)$.
\end{definition}
\noindent
When $\noise = 0$, we say that the MLR is \emph{noiseless}.

In this paper, we will study the \emph{parameter learning} problem for mixtures of linear regressions, that is, we wish to recover the parameters of the mixture.
To this end, we will need some assumptions to ensure that the regressors are uniquely identifiable.
These assumptions are standard throughout the literature.
Given a mixture of linear regressions $\calD$, let $\Delta = \min_{i \neq j} \| w_i - w_j\|_2$ be the minimum $L_2$ separation among all $w_i$ in the mixture, and we let $\pmin = \min_{i \in [k]} p_i$.
To normalize the instance, we will also assume that $\|w_i \|_2 \leq 1$ for all $i = 1, \ldots, k$.
However, more generally, our algorithms will have a mild polynomial dependence on the maximum $L_2$ norm of any $w_i$.
We omit this case for simplicity of exposition.

\paragraph{Mixtures of Hyperplanes}
We now turn our attention to mixtures of hyperplanes.
Formally:
\begin{definition}
	Given mixing weights $p \in [0,1]^k$ with $\sum_{i=1}^k p_i = 1$ and unit vectors $v_1, \cdots, v_k \in \S^{d-1}$, the corresponding \emph{mixture of hyperplanes} is the distribution with law given by $\sum_{i = 1}^k p_i \N(0, \Pi_i)$ where $\Pi_i \triangleq \Id - v_i v_i^{\top} \in \R^{d \times d}$.\label{def:hyperplanemixture}
\end{definition}
\noindent
As before, we need some assumptions on the parameters to ensure identifiability. Like before, let $\pmin = \min_{i \in [k]} p_i$. Because $v_i$ are now only identifiable up to sign, we define $\Delta$ to be the minimum quantity such that or all $i\neq j$ and $\epsilon_i,\epsilon_j\in\{\pm 1\}$, 
$	\norm{\epsilon_i v_i - \epsilon_j v_j}_2 \ge \Delta.$

\subsection{Miscellaneous Notation}
\label{subsec:notations}

\begin{itemize}
	\setlength\itemsep{0.2em}
	\item For real-valued functions $\calF:\R\to\R$ and $p\in \mathbb{N}$, we will use $\M_p(\calF)$ to denote $\int^{\infty}_{-\infty}x^p\cdot \calF(x)\, \d x$.
	\item Given $f\in L^2(\R)$, the Fourier transform of $f$ is denoted by 
			$\hat{f}[\omega] = \int^{\infty}_{-\infty}f(x)\cdot e^{-2\pi \i \omega x} \, \d x$.
	\item We will denote by $\Delta^n$ the probability simplex of vectors in $\R^n$ with nonnegative entries summing to 1.
	\item We will occasionally use notation like $x = [c_1,c_2]\cdot y$ and $x = 1 \pm \delta$ to mean $c_1 y\le x \le c_2 y$ and $1-\delta\le x \le 1 + \delta$ respectively.
	\item Given $v\in\R^{d+1}$, define $		\Sigma_{\eta}(v) \triangleq 
		\begin{pmatrix}
			\Id_d & v \\
			v^{\top} & \norm{v}^2 + \eta^2
		\end{pmatrix}.
	$ Note that this is the covariance matrix of a single spherical linear regression with noise variance $\eta^2$. When $\eta = 0$, we will denote this matrix by $\Sigma(v)$.
	\item For a vector $u\in\R^d$ and index $\ell\in[d]$, $u_{\ell}$ denotes the $\ell$-th entry of $u$. For indices $a,b\in[d]$, $u_{a:b}\in\R^{b-a+1}$ denotes the $a$-th through $b$-th entries of $u$.
	\item We will sometimes refer to a univariate mixture $\calF$ of zero-mean Gaussians with mixing weights $p\in\Delta^k$ and variances $\sigma^2_1,...,\sigma^2_k$ as a mixture of $k$ univariate zero-mean Gaussians ``with parameters $(\{p_i\}_{i\in[k]}, \{\sigma_i\}_{i\in[k]})$.'' We will define $\minvar{\calF} \triangleq \min_{i \in [k]} \sigma_i$ and $\maxvar{\calF}\triangleq \max_{i \in [k]} \sigma_i$ and refer to $\minvar{\calF}^2$ and $\maxvar{\calF}^2$ as the \emph{minimum} and \emph{maximum variance} of $\calF$, respectively.
\end{itemize}


\section{Overview of Techniques}

In this section, we give a high-level overview of how our algorithms work.
For clarity of exposition, in this subsection we will assume $\pmin = 1/k$ and $\Delta = \Theta(1)$.

\subsection{Fourier Moment Descent}
\label{subsec:fmd_tech}
We first describe our techniques that achieve Theorem~\ref{thm:informal-noiseless}, before describing how to adapt these techniques to achieve Theorems~\ref{thm:informal-noisy} and~\ref{thm:informal-pancakes}.

We begin by briefly recapping the \emph{moment descent} algorithm of~\cite{li2018learning}.
Moment descent is an iterative algorithm which attempts to find the parameters of one component at a time as follows.
Let $w_1, \ldots, w_k \in \R^d$ be the parameters of a MLR $\calD$ with separation $\Delta > 0$ and noise rate $\noise = 0$, and again for simplicity let us assume that the mixing weights are uniform.
To learn a single regressor, the idea is to maintain a guess $a_t \in \R^d$ for one of the regressors at each time step $t$, and iteratively refine it by making random steps and checking progress.
The measure of progress they consider is simply
\[
\sigma_t^2 \triangleq \min_{i \in [k]} \norm{w_i - a_t}^2_2 \; .
\]

Concretely, the algorithm proceeds as follows.
First, by a straightforward PCA step, we can essentially assume that $d \leq k$.
Then, given a guess $a_t$, the moment descent procedure updates by sampling a random unit vector $z\in\S^{d-1}$, defining $a'_{t+1}\triangleq a_t - \eta_t \cdot z \in \R^d$ for some learning rate $\eta_t$, and letting
\[
(\sigma_t')^2 \triangleq \min_{i \in [k]} \norm{w_i - a_t}^2_2 \; .
\]
In general, for $a\in\R^d$, the univariate distribution of the residual $y - \langle a,x\rangle$, where $(x,y) \in \R^d \times \R$ is sampled from $\calD$, is distributed as 
\[
\frac{1}{k} \sum_{i = 1}^k \N (\mu, \norm{w_i - a}^2_2) \; ,
\]
that is, it is distributed as a mixture of univariate Gaussians with mixing weights which are the same as those of $\calD$, and variances which are equal to the squared $L_2$ distances between the regressors and $a$. 
In particular, to estimate $ \min_{i \in [k]} \norm{w_i - a_t}^2_2$, one can learn the univariate mixture sufficiently well via~\cite{mv}, and simply read off the minimum variance.
By doing so, they can check if $\sigma_t' < \sigma_t$.
If so, they set $a_{t + 1} = a'_{t + 1} \in \R^d$, and repeat.

The main bottleneck in this routine is the univariate learning step.
Specifically, the algorithm of~\cite{mv} relies on the method of moments to learn the parameters of the univariate mixture of Gaussians, and as a result, takes $k^{O(k)}$ samples and time.
In fact, this is inherent:~\cite{mv} demonstrates that $k^{\Omega(k)}$ samples are necessary to learn the parameters of a mixture of Gaussians, precisely by leveraging moment matching instances.

However, all we need is an estimate of the minimum variance of the mixture of Gaussians.
One can first observe that it is possible to estimate the \emph{maximum} variance of a component in a univariate mixture of Gaussians based on a sufficiently high degree moment.
This is because the $p$-th moment of a uniform mixture of Gaussians $\calF$ with variances $s_1^2, \ldots, s_k^2$ has the following form, for $p$ even:
\begin{equation}
	\E_{Z\sim\calF}[Z^p] = \sum^k_{i=1} \frac{1}{k} \cdot s_i^p \cdot (p - 1)!! = \left[ c/k ,1 \right]\cdot \maxvar{\calF}^p\cdot p^{p/2} \; ,\label{eq:expectedzp}
\end{equation} 
for $\maxvar{\calF} \triangleq \max_{i \in [k]} s_i$
and some universal constant $c > 0$. Therefore, for any $\kappa > 0$, if we set $p = \Theta (\log k / \log (1 + \kappa)) = \Theta (\kappa^{-1} \log k)$, we have that
\begin{equation}
p^{-1/2} \E_{Z\sim\calF}[Z^p]^{1/p} \in \left[ 1 - \Theta \left( \kappa \right), 1 \right] \cdot \maxvar{\calF}^p
\end{equation}
which yields a $(1 + \kappa)$ approximation to the maximum variance approximation to the largest variance of $\calF$.
Moreover, we can estimate the left-hand side in $p^{O(p)} = 2^{\widetilde{O} (\kappa^{-1} \log k)}$ samples:
\begin{lemma}[Concentration of empirical moments]\label{lem:moment_conc}
	Let $p \in \mathbb{N}$ be even and $t \in \mathbb{N}$, and let $\delta,\beta>0$. Then for 
	\begin{equation*}\label{eq:maxvar_samples}
	N = k^{-2} \cdot \beta^{-2} \cdot p^{\Theta(p)} \cdot \ln(1/\delta)^{\Theta(p)} \cdot \maxvar{\calF}^{\Theta(p)},
	\end{equation*}
	we have that
	\begin{equation*}
		\Pr_{Z_1, \cdots, Z_N \sim \calF} \left[\frac{1}{N}\sum^N_{i=1}Z_i^p = (1\pm \beta)\cdot\E_{Z\sim\calF}[Z^p]\right] \ge 1 - \delta.\label{eq:conc}
	\end{equation*}
\end{lemma}
\noindent
For clarity of exposition, we defer the proof of this lemma to Section~\ref{app:moment_conc}.

Unfortunately, \emph{a priori} this argument says nothing about estimating the minimum variance. It is easy to see that two mixtures of univariate zero-mean Gaussians $D_1,D_2$ can have very similar $p$-th moments but wildly different minimum variances (e.g. take $D_1$ to be a single Gaussian $\N(0,k^{100})$, and take $D_2$ to be a uniform mixture of $\N(0,k^{100})$ and $\N(0,1)$).

The key insight is that while higher-degree moments of a mixture $D$ of zero-mean Gaussians tell us nothing about the minimum variance, those of its \emph{Fourier transform} do. The reason is because of the following observation: 
\begin{observation}
\label{obs:fourier}
\rm{
If the components of $D$ have variances $s^2_1,...,s^2_k$ and mixing weights $p_1,...,p_k$, the \emph{Fourier transform} of the density of $D$ is a new (unnormalized) mixture of Gaussians with variances $\Theta(s^{-2}_1),$ $\cdots,$ $\Theta(s^{-2}_k)$ and mixing weights proportional to $p_1/s_1,...,p_k/s_1$ (see Fact~\ref{fact:fouriergaussian}).}
\end{observation}
In particular, if we have a sufficiently good estimate of the maximum variance of any component in the Fourier transform of $D$, then by inverting this estimate, we can estimate the minimum variance of any component in $D$.
So if we had access to the Fourier transform of $D$, we could then use the moments of this distribution to estimate the maximum variance of any component of the Fourier transform, which would allow us to learn the minimum variance of $D$.

What remains is to estimate moments of the Fourier transform of $D$ using solely samples from $D$. Here we use existing primitives for univariate density estimation~\cite{chan2014efficient,acharya2017sample} to obtain an explicit approximation $\tilde{D}$ to the density of $D$, after which we can explicitly compute moments of the Fourier transform of $\tilde{D}$. We defer the technical details of how to argue that its moments are close to those of the Fourier transform of $D$ to Section~\ref{subsec:minvar}, as they are rather involved.
In short, this allows us to achieve the same sorts of guarantees for estimating min-variance as for estimating max-variance: for any $\kappa > 0$, we can learn the minimum variance of the mixture to multiplicative error $1 + O(\kappa)$ with $2^{\widetilde{O} (\kappa^{-1} \log k)}$ samples and time.

However, there is an important subtlety here.
Namely, the quality of the approximation to the minimum variance we can obtain strongly depends on the degree of the Fourier moment we use.
The degree in turn dictates the sample complexity of the algorithm: the higher the Fourier degree we need, the better the univariate density estimate must be, and therefore, the more samples we need in order to adequately perform the density estimate.
Therefore, if the difference between $\sigma_t$ and $\sigma_t'$ is too small, we cannot reliably check if we've made progress without taking too many samples.
Unfortunately, the difference between these two quantities is typically quite small.
Since $a_{t +1}'$ is a random perturbation of $a_t$, we have that with probability $1 / \poly (k)$, it holds that 
\begin{equation*}
\sigma'_t \le \left( 1 - \Omega\Big( \frac{1}{k} \Big) \right) \sigma_t \; ,
\end{equation*}
and moreover, this is tight.
In particular, this says that we would need to take $\kappa = \Omega (1/k)$ in the discussion above, which would result in a $2^{\Theta (k)}$ runtime, which we wish to avoid.

However, we show that with subexponentially large probability, the difference is sufficiently large so that we can detect this difference using subexponentially many samples.
In particular, observe that for any constant $0 < c < 1/2$, if $z$ is a random unit vector in the span of $\{w_i - a_{t}\}$, then with probability $\exp(-k^{1 - 2c})$ we have that 
\begin{equation}
\left\langle z,\frac{w_i - a_{t}}{\norm{w_i - a_{t}}_2}\right\rangle \ge \Omega(k^{-c}),\label{eq:elem_good_event}
\end{equation} 
in which case if we define $a'_{(t+1)}$ as previously, we get that with probability $\exp (-k^{1 - 2c})$,
\begin{equation}
\sigma'_t \le \left( 1 - \Omega\Big( \frac{1}{k^{2c}} \Big) \right) \sigma_t \; .\label{eq:goodprogress}
\end{equation} So by trying super-polynomially many random directions $z$ at every step, we ensure with high probability that one of those directions will make $1 - \Omega(k^{-2c})$ progress, for some $c < 1$.
By combining this with our certification procedure as described above, we show that we can make non-trivial progress in the algorithm after only sub-exponentially many samples.

By iteratively applying this update, we are able to obtain an $a_T$ so that
\[
\| a_T - w_i \|_2 \leq \frac{\Delta}{k^{100}} \; ,
\]
in subexponential time.
We call this subroutine \emph{Fourier moment descent}, and it allows us to learn a single regressor to good accuracy.
For technical reasons, the complexity of this approach grows as we get closer to $w_i$, however, it allows us to obtain a very good ``warm start''.
In the noiseless case, this can be combined with the boosting procedure from~\cite{li2018learning} to obtain arbitrarily high accuracy.

This technology now allows us to learn a single regressor to very high accuracy.
In the noiseless setting, this allows us to ``peel off'' the samples from this component almost completely, and we can now repeat this process on the sub-mixture with this component removed to learn another component, and iterate to eventually learn all of the regressors.

That said, as we shall see, this is much trickier in the presence of noise.

\subsection{Learning With Regression Noise}
\label{subsec:tech_noise}
What changes when we assume that there is a significant amount of noise $\noise$?
From the perspective of our Fourier moment descent algorithm, it turns out not much does, at least to a certain extent: in fact, essentially the same argument goes through and allows us to learn a single component to error at most
\[
\| a_T - w_i \|_2 \leq \frac{\Delta}{k^{100}} + O(\noise) \; ,
\]
where $\noise$ is the standard deviation of the white noise.

However, learning \emph{all} components becomes substantially more difficult.
In particular, the peeling process no longer works: the fact that there is regression noise does not allow us to perfectly remove the influence of a component that we have learned from the rest of the mixture.
As a result, it is no longer clear how to go from an algorithm that can learn a single component to one that can learn all components.

To avoid this, we circumvent the need for peeling altogether.
By a delicate analysis, we will show that with decent probability, we can control the dynamics of the Fourier moment descent algorithm, so that it will converge to the regressor that it was initially closest to.
This is the key technical ingredient behind getting Fourier moment descent to handle regression noise.

We sketch its proof below.
Let $a_t$ be the current iterate of Fourier moment descent, and suppose that $i^* = \argmin_{i \in [k]} \| w_i - a_t \|_2$.
Then, one can show that if $a'_{t + 1} \in \R^d$ is a random perturbation of $a_t$, then with probability at least $\exp (-\Omega (1 / \Delta^2)) / \poly (k)$, we have that $i^*$ is still the closest component to $a_{t + 1}'$.
Moreover, this bound is tight up to polynomial factors, if all we assume about the $w_i$ is that they are $\Delta$-separated.
One could hope that by using this sort of argument, we could argue that with at least subexponentially large probability, we always stay closest to $w_{i^*}$.
However, this argument runs into a couple of difficulties, which we address one at a time.

The first difficulty is that while this happens with decent probability for any individual $a_{t + 1}' \in \R^d$, our algorithm will typically need to try sub-exponentially many perturbations before we find one that makes progress.
Thus the naive union bound over all sub-exponentially many $a_{t + 1}' \in \R^d$ would be far too loose to say anything here.
To get around this, we instead demonstrate that conditioning on the event that $\sigma_t' < \sigma_t$, we remain closest to $i^*$ with non-trivial probability.

However, even this is insufficient, as we only have a multiplicative estimate of $\sigma_t$ and $\sigma_t'$.
To get around this, we make a stronger assumption: we assume that not only is our iterate the closest to $w_{i^*} \in \R^d$, but we also assume that its distance to the other $w_i \in \R^d$ for $i \neq i^*$ is at least some multiplicative factor larger than its distance to $w_{i^*} \in \R^d$.
Specifically, we assume that
\begin{equation}
\label{eq:noisy-gap}
\| w_i - a_t \|_2 \geq \left( 1 + c \frac{\Delta^2}{\sqrt{k}} \right) \| w_{i^*} - a_t \|_2 \; ,
\end{equation}
for all $i \neq i^*$, and some constant $c > 0$.
By making this stronger assumption, we are able to demonstrate that with probability at least $1/\poly(k)$, this gap is maintained for the $a_{t + 1}'$ which Fourier moment descent chooses as the next iterate.

Our overall algorithm then will be to demonstrate an initialization scheme for $a_0$ so that for any $i^*\in[k]$, \eqref{eq:noisy-gap} holds for $a_0$ and that $i^*$ with probability at least $\exp( -\widetilde{O}( \sqrt{k}/\Delta^2 ) )$.
If we can do so, then if we run Fourier moment descent starting from these random initializations enough times, then eventually with high probability we will output multiple estimates for each $w_i$ for $i = 1, \ldots, k$, and then we can simply run a basic clustering algorithm to recover all of the regressors. 
Furthermore, because each run of moment descent only needs to go on for $\widetilde{O}(\sqrt{k}/\Delta^2)$ iterations, and at each iteration we stay closest to the same component that the previous iterate was closest to with $1/\poly(k)$ probability, it suffices to try $\exp(\widetilde{O}(\sqrt{k}/\Delta^2))$ random initializations for all this to work.

Lastly, it turns out this initialization scheme is also quite delicate.
For instance, if we simply chose random initializations over the unit ball, then these will, with overwhelming probability, favor being close to $w_i$ with small norm. If we have regressors with different norms, we will thus likely never be close to a $w_i \in \R^d$ with large norm in our initializations, and as a result, cannot argue that our Fourier moment descent algorithm will ever recover this regressor.
To get around this, we demonstrate a gridding scheme, where we initialize randomly over spheres of up to radius $r$, for $r$ on a fine grid between $0$ and $O(k^{1/4})$.

We emphasize that this is quite counterintuitive: when we start at radius $O(k^{1/4})$, we are exceedingly far away from every $w_i \in \R^d$, as they are assumed to have norm at most $1$.
However, our analysis of Fourier moment descent works fine in this setting, and by having such large radius, we are able to ensure that for each $i^* \in[k]$, \eqref{eq:noisy-gap} holds for $a_0$ and $i^*$ with at least the desired probability.
The details of this are quite technically involved, and we defer them to Section~\ref{sec:allcomps}.

\paragraph{More Noise-Robust Boosting}

As mentioned previously, in the noiseless setting, the boosting algorithm of~\cite{li2018learning} allows us to bootstrap our warm start, obtained via Fourier moment descent, to arbitrarily high accuracy.
It turns out that in the noisy setting, their boosting algorithm also allows one to go slightly below $O(\noise)$.
Interestingly, motivated again by the connection to Fourier analysis, we demonstrate an improved boosting algorithm that is able to tolerate substantially more noise as well as a much weaker warm start.

The boosting algorithm of~\cite{li2018learning} is based on stochastic gradient descent on a regularized form of gravitational potential, which was notably used in \cite{holden2018gravitational}.
While this objective is \emph{concave}, they demonstrate that in a small neighborhood around the true regressors, SGD updates based on this objective contract in expectation, and hence they make progress.

\begin{wrapfigure}{r}{0.4\textwidth}
	\includegraphics[width=0.4\textwidth]{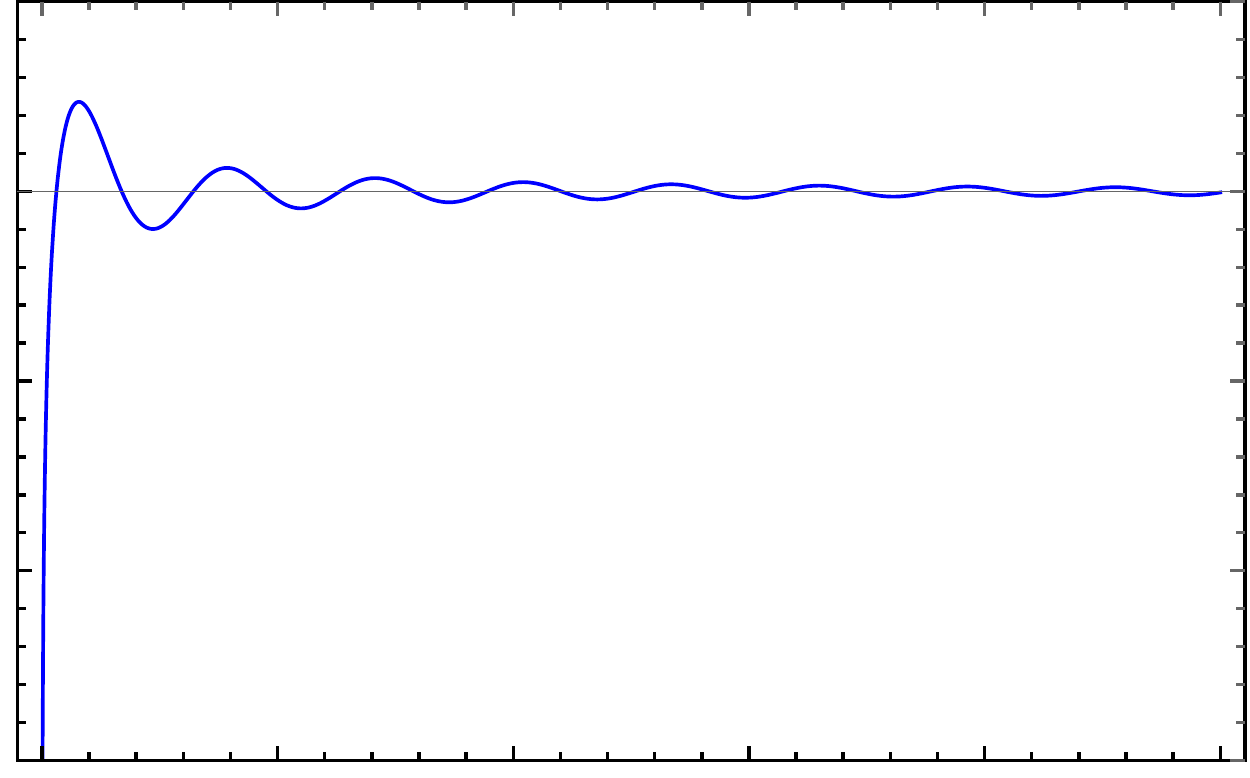}
	\caption{Cosine integral function $\text{Ci}(x) = -\int^{\infty}_{x}\frac{\cos(t)}{t}\, \d t$}
	\label{fig:cosint}
\end{wrapfigure}

In contrast, we propose an update based on a regularized form of the \emph{cosine integral objective} (see Figure~\ref{fig:cosint} for a plot of the cosine integral --- the objective function we use is a regularized version of the objective $g(v)\triangleq \E_{(x,y)\sim\calD}\text{Ci}(|\langle v,x\rangle - y|)$).
This objective looks much worse behaved: it is neither convex nor concave---indeed, it is not even monotone!
However, the key technical fact which makes this objective more noise tolerant is precisely Observation~\ref{obs:fourier}.
This allows us to argue that the contribution to the gradient of the objective from the ``good'' component, which we are close to, dominates the contribution of the gradient from the other ``bad'' components.
At a high level, ignoring many technical issues for the time being, this is because if $v$ is the current iterate, a main part of the contribution to the gradient from component $i$ is $\E_{g \sim \N(0, \beta_i^2)} [\cos (g)]$, where $\beta_i$ is a monotone function of $\norm{w_i - v}_2$.
This is exactly the real part of the Fourier transform of $\N(0, \beta_i^2)$, and the fact that this decreases as $\exp (O(\beta_i^{-2}))$ follows precisely from Observation~\ref{obs:fourier}.
This allows us to have much finer control over the contribution from the ``bad'' components, which allows us to tolerate significantly more noise and a substantially weaker warm start.
For the quantitative details, see Section~\ref{sec:boosting}.

\subsection{Learning Mixtures of Hyperplanes} 

As we will see, mixtures of hyperplanes share enough qualitative features with MLRs that an appropriate instantiation of our techniques also suffices for this problem. 

It is not hard to show that vanilla moment descent can be modified as follows to get a $\widetilde{O}(d\cdot \exp(\widetilde{O}(k)))$-time algorithm for mixtures of hyperplanes. First it is not hard to see that as with spherical MLRs, here we can still effectively reduce the dimension of the problem from $d$ to $k$. At time $t$ we still maintain an estimate $a_t$ for one of the components, and in lieu of the usual progress measure $\min_{i\in[k]}\norm{w_i - a_t}_2$ to which we no longer have access, we can use the modified progress measure $\sigma_t \triangleq\min_{i\in[k]}\norm{\Pi_i a_t}_2$. (See Definition~\ref{def:hyperplanemixture} for definition of $\Pi_i \in \R^{d \times d}$) We can estimate $\sigma_t$ in $\exp(\widetilde{O}(k))$ time by simply projecting in the direction of $a_t$ itself to get a mixture of univariate Gaussians with variances $\{\norm{\Pi_i a_t}^2_2\}_{i\in[k]}$ and learning that mixture via \cite{mv}. This leads to the following straightforward modification of moment descent: repeatedly update $a_t \in \R^d$ by sampling many random steps $\{\eta_t\cdot z_j\}$ for step size $\eta_t$ and $z_j \in \S^{d-1}$ and taking one of them to get to $a_{t+1} \in \R^d$ if it contracts the progress measure.

One immediate issue with this approach, even for achieving $\widetilde{O}(d )\cdot \exp(\widetilde{O}(k))$ runtime, is that if we don't control $\norm{a_t}_2$, then for all we know $\sigma^2_t$ contracts simply because our random steps are taking us closer and closer to 0. The natural workaround is to insist $a_t\in\S^{d-1}$ for all $t$ by projecting onto $\S^{d-1}$ after every step. That is, in every iteration of moment descent, given a random collection of candidate steps $\{\eta\cdot z_j\}_j$, choose one for which, if we define 
\begin{align*}
a_{t+1} = \frac{a_t + \eta\cdot z_j}{\norm{a_t + \eta\cdot z_j}_2},
\end{align*}
then the progress measure contracts. One can show this already suffices to get a $\widetilde{O}( d )\cdot \exp(\widetilde{O}(k))$-time algorithm for learning mixtures of hyperplanes.

The extra projection onto $\S^{d-1}$ at every step introduces a variety of technical challenges for trying to implement the strategy of Section~\ref{subsec:fmd_tech} to achieve sub-exponential runtime. Specifically, in order to carry out the balancing of parameters from Section~\ref{subsec:fmd_tech}, one needs to be careful in choosing the right events, analogous to the event in \eqref{eq:elem_good_event}, such that 1) conditioned on those events $\sigma_t$ contracts by at least a factor of $1 - \Omega(k^{-a})$ for some $a>0$, and 2) these events all occur with probability $\exp(-k^{1 - a'})$ for some $a' > 0$.

If for instance $a_t$ is positively correlated with some $v_{i^*} \in \R^d$, it turns out the right events to choose are that the random step is both $k^{-1/5}$-positively correlated with $\frac{\Pi_{i^*}a_t}{\norm{\Pi_{i^*}a_t}_2}$ and at least $k^{-1/5}$-negatively correlated with $v_{i^*} \in \R^d$, and because these directions are orthogonal, if $z_j$ is a random vector in the span of $v_1,...,v_k$ then we could lower bound the probability of both events occurring by the product of the probabilities they individually occur, which is $\exp(-k^{3/5})$, and get 2) for $a' = 2/5$. The analysis for showing 1) (for $a = 3/5$) is involved, so we defer it to Section~\ref{subsec:hyperplanes_moment}. These together yield a $\widetilde{O}(d)\cdot \exp(k^{3/5})$-time algorithm to learning one component of a mixture of hyperplanes.


To learn all components, we would like to implement some kind of boosting procedure. Our approach here is to regard $\calD$ in a certain way as a non-spherical MLR with well-conditioned covariances, at which point we can invoke, e.g., the boosting algorithm of \cite{li2018learning}. We defer these details to Section~\ref{subsec:hyperplane_boost}.
Once we are able to refine an estimate for a direction of $\calD$ to arbitrary precision, we can carry out the ``peeling'' procedure outlined at the end of Section~\ref{subsec:fmd_tech} to learn all components.



\section{Roadmap}

Here we give a brief overview of the organization of the rest of the paper.
In Section~\ref{subsec:facts} we give some additional technical preliminaries we will need in the paper.
In Section~\ref{sec:fourier} we present our Fourier moment descent algorithm for learning a single component.
In Section~\ref{sec:allcomps_nonoise} we show how to use this to learn all the components, when there is no regression noise.
In Section~\ref{sec:allcomps} we demonstrate a modification of our algorithm to learn all the components in the presence of regression noise.
In Section~\ref{sec:hyperplanes} we demonstrate our subexponential time algorithm for learning a mixture of hyperplanes.
Finally, in Section~\ref{sec:boosting} we demonstrate our improved boosting algorithm based on the cosine integral objective.
Deferred proofs appear in the Appendix, as well as our moment-matching example.



\section{Additional Technical Preliminaries}
\label{subsec:facts}
In this section we give a number of miscellaneous facts that we will require throughout the paper.
For clarity of exposition we defer the proofs of these facts to Appendix~\ref{sec:defer}.
The first is a monotonicity property of moments of Gaussians, restricted to the tails of the Gaussian:
\begin{fact}\label{fact:monotone}
	Let $\sigma^*>0$ and $\tau > \sigma^*$, and let $p \in \mathbb{N}$ be even. Then \begin{equation}
	\int_{[-\tau,\tau]^c} \N(0, \sigma^2; x) \cdot x^p \, \d x< \int_{[-\tau,\tau]^c} \N(0,(\sigma^*)^2; x) \cdot x^p \, \d x \label{eq:monotone}
	\end{equation} 
	for all $0 < \sigma<\sigma^*$.
\end{fact}

The following fact about Fourier transforms of Gaussian pdfs is standard.
\begin{fact}\label{fact:fouriergaussian}
	\begin{equation*}
	\widehat{\N(0,\sigma^2)}[\omega] = e^{-2\pi^2 \omega^2\sigma^2} = \frac{1}{\sqrt{2\pi}\sigma}\N\left(0,\frac{1}{4\pi^2\sigma^2}; \omega\right)
	\end{equation*}
\end{fact}

\noindent
We will require the following standard estimate for Gaussian tails, see e.g. Proposition 2.1.2 in \cite{vershynin2018high}.
\begin{fact}\label{fact:gaussian_tails}
Let $g\sim \N(0,\Id_d)$. Let $w\in\S^{d-1}$ be a fixed vector. Then for all $t > 0$,
	\begin{equation}\label{eq:gaussian_tails}
		\left(\frac{1}{t} - \frac{1}{t^3}\right)\cdot\frac{1}{\sqrt{2\pi}}e^{-t^2/2} \le \Pr_g[\langle g,w\rangle \ge t] \le \frac{1}{t}\cdot\frac{1}{\sqrt{2\pi}}e^{-t^2/2}.
	\end{equation}
\end{fact}

\noindent
We also have the following standard concentration inequality.
\begin{fact}
Let $g\sim\N(0,\Id_d)$. There is a universal constant $c_{\mathrm{shell}} > 0$ such that for all $t > 0$,
	\begin{equation}
		\Pr_g \left[ \left| \|g\|_2 - \sqrt{d} \right| \ge t \right] \le 2e^{-c_{\mathrm{shell}}t^2} \label{eq:thinshell}
	\end{equation}
\label{fact:thinshell}
\end{fact}
\noindent
Facts~\ref{fact:gaussian_tails} and \ref{fact:thinshell} imply the following pair of inequalities about the correlation between a Gaussian vector and a given unit vector.

\begin{corollary}\label{cor:unitcorr}
	Let $w\in\S^{d-1}$, $g\sim\N(0,\Id_d)$, and $v = g/\norm{g}_2$. Then for any constant $0 < \gamma < 1/2$, the following holds.

	There exist increasing functions $\underline{f}_{\gamma}, \overline{f}_{\gamma}, D: \R_{>0}\to\R_{>0}$ such that for any absolute constants $0 < \underline{\alpha}\le \overline{\alpha}$, we have that for $d \ge D(\overline{\alpha})$,
	\begin{equation*}
		\Pr[\langle v,w\rangle \ge \underline{\alpha}\cdot d^{-\gamma}] \ge e^{-\underline{\beta}d^{1-2\gamma}}
	\end{equation*}
	and 
	\begin{equation*}
		\Pr\left[\langle v,w\rangle\le \overline{\alpha}d^{-\gamma}\right]\ge 1 - e^{-\overline{\beta} d^{1-2\gamma}}
	\end{equation*} for $\underline{\beta} = f(\underline{\alpha})$ and $\overline{\beta} = f(\overline{\alpha})$.
\end{corollary}
\noindent
It will also be useful to obtain a similar bounds for the probability that the inner products of a random vector with two orthogonal directions are \emph{simultaneously} in a particular range.
\begin{corollary}
\label{cor:joint}
	Let $w_1,w_2\in\S^{d-1}$ be orthogonal, $g\sim\N(0,\Id_d)$, and $v = g/\norm{g}_2$.

	For any $\alpha_1,\alpha_2 > 0$, we have that for sufficiently large $d$, \begin{equation}
		\Pr\left[\left(\langle v, w_1\rangle \ge \alpha_1\cdot d^{-1/4}\right) \wedge \left(\langle v, w_2\rangle \le \alpha_2\cdot d^{-1/2}\right)\right] \ge \frac{1}{\poly(d)}\Pr[\langle v,w_1\rangle \ge \alpha_1\cdot d^{-1/4}].\label{eq:wantprobrelation}
	\end{equation}
\end{corollary}
\noindent
Finally, we also use the following approximate $k$-SVD algorithm as a black-box:
\begin{fact}[\cite{rokhlin2009randomized}]
\label{thm:power-method}
Let $M \in \R^{d \times d}$, let $k \leq d$ be a non-negative integer, and let $\sigma_1 \geq \sigma_2 \geq \ldots \sigma_d$ be the non-zero singular values of $M$.
For any $k = 1, \ldots, d - 1$, let $\gap_k = \sigma_k / \sigma_{k + 1}$.
Suppose there is an oracle $\calO$ which runs in time $R$, and which, given $v \in \R^d$, outputs $Mv$.
Then, for any $\eta, \delta > 0$, there is an algorithm $\textsc{ApproxBlockSVD} (M, \eta, \delta)$ which runs in time $\widetilde{O} (k \cdot R \cdot \log \tfrac{1}{\eta \cdot \delta \cdot \gap_k})$, and with probability at least $1 - \delta$ outputs a matrix $U \in R^{d \times k}$ with orthonormal columns so that $\| U - U_k \|_2 < \eta$, where $U_k$ is the matrix whose columns are the top $k$ right singular vectors of $M$.
\end{fact}
\noindent

We next collect some elementary matrix perturbation bounds.

\begin{lemma}[Wedin Theorem]
	Let $\epsilon \ge 0$. For psd matrices $\vec{A}, \hat{\vec{A}}\in\R^{d\times d}$ for which $\norm{\vec{A} -\hat{\vec{A}}}_2 \le \epsilon$, and any $k\in[d]$, $\gap>0$, and $\mu\ge 0$ such that the top $k$ eigenvalues of $\vec{A}$ are at least $\mu$ and the bottom $d - k$ eigenvalues of $\hat{\vec{A}}$ are at most $\mu - \gap$, we have $\norm{\vec{U}^{\top}\hat{\vec{U}}}_2 \le \epsilon/\gap$, where $\vec{U}$ is the matrix whose columns consist of the top $k$ eigenvectors of $\vec{A}$, and $\hat{\vec{U}}$ is the matrix whose columns consist of the bottom $d - k$ eigenvectors of $\hat{\vec{B}}$.\label{lem:wedin}
\end{lemma}

\begin{lemma}
	Let $\vec{A},\hat{\vec{A}}\in\R^d$ be rank-$k$ psd matrices for which $\norm{\vec{A} - \hat{\vec{A}}}_2 \le \epsilon$, and suppose $\vec{A}$ has minimum eigenvalue $\lambda$. Let $\vec{U}, \hat{\vec{U}}\in\R^{d\times k}$ be the column matrices consisting of the $k$ nonzero eigenvectors of $\vec{A}, \hat{\vec{A}}$ respectively. Then \begin{equation}\langle \vec{U}^{\top}a,\vec{U}^{\top} b\rangle \ge \langle a,b\rangle - \epsilon/\lambda\end{equation} for any $a,b\in\S^{d-1}$ in the column span of $\hat{\vec{U}}$. In particular, for $a = b \in\S^{d-1}$, we conclude that \begin{equation}\norm{\vec{U}^{\top} a}_2 \ge (1 - \epsilon/\lambda)^2 \ge 1 - 2\epsilon/\lambda.\end{equation}
	\label{lem:wedin_angle}
\end{lemma}

\begin{proof}
	By Lemma~\ref{lem:wedin}, $\norm{(\Id_d - \vec{U}\vec{U}^{\top})\hat{\vec{U}}}_2 \le \epsilon/\lambda$. We can write \begin{equation}\langle \vec{U}^{\top}a,\vec{U}^{\top}b\rangle = \left\langle \left(\Id_d - \vec{U}\vec{U}^{\top}\right)a,\left(\Id_d - \vec{U}\vec{U}^{\top}\right)b\right\rangle \ge \langle a,b\rangle - \epsilon/\lambda,\end{equation} from which the lemma follows.
\end{proof}


\section{Warm Start via Fourier Moment Descent}
\label{sec:fourier}

Here we propose a technique for moment descent based on approximating the minimum variance of a component in a mixture of univariate zero-mean Gaussians. The main result of this section is an algorithm, which we call \textsc{FourierMomentDescent}, for learning a single component of a mixture of $k$ linear regressions in time and sample complexity sub-exponential in $k$:

\begin{theorem}[Fourier moment descent]\label{thm:fmd_main}
	Given $\delta,\epsilon>0$ and a mixture of spherical linear regressions $\calD$ with separation $\Delta$ and noise rate $\noise = O(\epsilon)$, 
	there is an algorithm (\textsc{FourierMomentDesc} \textsc{ent}($\calD,\delta,\epsilon$) inAlgorithm~\ref{alg:fourier_moment_descent}) that outputs a vector $v \in \R^d$ such that with probability $1-\delta$, we have $\norm{w_i - v}_2 \le O(\epsilon)$ for some $i\in[k]$. Furthermore, \textsc{FourierMomentDescent} requires sample complexity 
	\begin{align*}
	N = \widetilde{O}\left(d\epsilon^{-2}\ln(1/\delta)\pmin^{-4}\cdot\poly\left(k,\ln(1/\pmin),\ln(1/\epsilon)\right)^{O(\sqrt{k}\ln(1/\pmin))}\right)
	\end{align*}
	and time complexity $Nd\cdot\poly\log(k,d,1/\Delta,1/\pmin,1/\epsilon)$.
\end{theorem}


In Section~\ref{subsec:notations} we give an algorithm for estimating the minimum variance of a mixture of univariate, zero-mean Gaussians via its Fourier transform. In Section~\ref{subsec:fmd} we show how to leverage this technology to obtain our algorithm \textsc{FourierMomentDescent} and then give a proof of Theorem~\ref{thm:fmd_main}.

\subsection{Estimating Minimum Variance}
\label{subsec:minvar}

Here we give the key primitive underlying all of the algorithmic results of this work: an algorithm for estimating the minimum variance of a mixture of zero-mean Gaussians. 
This requires some setup regarding existing technology for density estimation.

\subsubsection{Density Estimation in \texorpdfstring{$L_2$}{L2}}
Our main density estimation tool will be to use piecewise polynomials.
We favor them because there are clean algorithms for density estimation via piecewise polynomials, and moreover, the form of the estimator will be useful for us later on.
Formally:
\begin{definition}
	An $s$-piecewise degree-$d$ polynomial $p : \R \rightarrow \R$ is specified by a collection of intervals $I_1 = [-\infty,a_1], I_2 = [a_1,a_2], I_3 = [a_2,a_3], ..., I_{s-1} = [a_{s-2},a_{s-1}], I_s = [a_{s-1},+\infty]$ and $s$ degree-$d$ polynomials $p_1,...,p_s$ such that for any $i\in[s]$ and $x\in I_s$, $p(x) = p_i(x)$. We refer to $a_1,...,a_{s-1}$ as the \emph{nodes} of $p$.
\end{definition}
\noindent
We will use the following algorithm as a black box:
\begin{theorem}[Theorem 43 in \cite{acharya2017sample}]\label{thm:piecewise}
For any $\eta > 0$, there is an algorithm that, given sample access to a mixture $\calF$ of $k$ univariate Gaussians, outputs a $O(k)$-piecewise degree-$O(\log 1 / \eta)$ polynomial hypothesis distribution $\calF'$ for which $\tvd(\calF,\calF') \le\eta$, using $N = O\left((k/\eta^2)\ln(1/\eta)\ln(1/\delta)\right)$ samples and running in time $\tilde{O}(N)$.
\end{theorem}

\begin{algorithm}\caption{\textsc{L2Estimate}($\calF,\underline{\sigma},\eta,\delta$), Corollary~\ref{cor:l2_estimate}}\label{alg:l2_estimate}
\begin{algorithmic}[1]
	\State \textbf{Input}: Sample access to univariate $k$-GMM $\calF$, a number $\underline{\sigma}$ for which $\underline{\sigma} \le \minvar{\calF}$, precision parameter $\eta>0$, failure probability $\delta > 0$
	\State \textbf{Output}: Piecewise polynomial function $\calG$ for which $\norm{\calF - \calG}^2_2 \le\eta$
			\State Let $N = \Theta\left(\frac{k}{2\pi \underline{\sigma}^2\eta^2}\log\left(\frac{1}{\sqrt{2\pi}\underline{\sigma}\eta}\right)\right)$
			\State Draw $N$ samples from $\calF$ and run the algorithm from Theorem~\ref{thm:piecewise} to obtain an estimate $\calF'$ for which $\tvd(\calF,\calF') \le \sqrt{2\pi}\underline{\sigma}\cdot \eta$.
			\State For each $x\in\R$, define $\calG(x) = \min\left\{ \max \{ \calF'(x) , 0 \}, \frac{1}{\sqrt{2\pi}\underline{\sigma}} \right \}$.
			\State Output $\calG$.
\end{algorithmic}
\end{algorithm}

\begin{corollary}[Guarantee for \textsc{L2Estimate}]\label{cor:l2_estimate}
	For any $0<\eta,\delta< 1$, mixture of $k$ univariate Gaussians $\calF$, and $\underline{\sigma}>0$ for which $\underline{\sigma} \le \minvar{\calF}$, with probability at least $1 - \delta$ \textsc{L2Estimate}($\calF, \underline{\sigma}, \eta,\delta$) (Algorithm~\ref{alg:l2_estimate}) outputs a $O(k \log (1 / (\eta \underline{\sigma} ) ) )$-piece degree-$O(\log (1/ ( \eta \underline{\sigma} ) )$ polynomial hypothesis distribution $\calG: \R\to\R_{\ge 0}$ for which $\norm{\calF - \calG}^2_2 \le\eta$, using $N = O\left((k/\eta^2) \log (1/(\eta \underline{\sigma})) \log(1/\delta) \right)$ samples and running in time $\tilde{O}(N)$.
\end{corollary}

\begin{proof}
	Because $\calF$ has range in $\left[0,\frac{1}{\sqrt{2\pi}\minvar{\calF}}\right]$, by construction we have that $\tvd(\calF,\calG) \le \tvd(\calF,\calF')$. By Holder's, it suffices to show $\norm{\calF - \calG}_{\infty} \le \frac{1}{\sqrt{2\pi}\cdot\underline{\sigma}}$. But because $\calG(x) \le \calF(x)$ for all $x\in\R$, it is enough to show $\norm{\calF}_{\infty} \le \frac{1}{\sqrt{2\pi}\cdot\underline{\sigma}}$, which just follows from the fact that $\calF$ is a convex combination of Gaussians of variance at least $\minvar{\calF} \ge \underline{\sigma}$. Note that $\calG$ is a piecewise polynomial because we can refine the intervals defining $\calF'$ to incorporate the intersections of $\calF'$ with the lines $y = \frac{1}{\sqrt{2\pi}\underline{\sigma}}$ and $y = 0$.
	Since each individual component can intersect with these lines at most $O(\log ( 1 / ( \eta \underline{\sigma} ) ) )$ times since they have degree at most $O(\log ( 1 / ( \eta \underline{\sigma}) ) )$, this yields the desired bound on the number of pieces of the resulting piecewise polynomial estimate.
\end{proof}

\subsubsection{Minimum Variance Via Fourier Transform Moments}

We now show how to use an $L_2$-close estimator for the density of a mixture $\calF$ of zero-mean univariate Gaussians to approximate $\minvar{\calF}$.
As a first step, we show how to use an $L_2$-close estimator to estimate high moments of $\calF$:

\begin{lemma}\label{lem:moment_diff}
	For any even integer $p\in \mathbb{N}$ and $\xi > 0$ the following holds. Let $\calF: \R\to\R$ be an mixture of $k$ Gaussians given by 
	\begin{equation*}
		\calF(x) = \sum^k_{i=1} p_i \cdot \N(0 , \sigma^2_i ; x),
	\end{equation*} 
	and define $L \triangleq \sum_{i=1}^k p_i$. Let $\overline{\sigma} > 0$ be any number for which $\overline{\sigma} \ge \max_{i \in [k]} \sigma_i$.

	Let $\tau = 8\overline{\sigma}^2 \cdot \max( p , \ln( 4 L / \xi ) )$ and $\eta= \frac{\xi^2p}{8\tau^{2p+1}}$. Then if function $\calG:\R\to\R_{\ge 0}$ for satisfies $\norm{\calF - \calG}^2_2 \le \eta$, then $\calG'$ defined by $\calG'(x) = \bone{x\in[-\tau,\tau]}\cdot\calG(x)$ for all $x\in\R$ satisfies 
	\begin{equation}\label{eq:moment_diff}
	|\M_p(\calF) - \M_p(\calG')| \le \xi.
	\end{equation}
\end{lemma}

\begin{proof}
	For simplicity, we define
	\begin{equation*}
	\smax = \smax(\calF).
	\end{equation*}
	We would like to pick the truncation threshold $\tau$ so that 
	\begin{equation}\label{eq:tail}
		\int_{[-\tau,\tau]^c} x^p \cdot \calF(x) \, \d x \le \xi/2.
	\end{equation} 
	For this, it suffices to take $\tau$ for which 
	\begin{equation}\label{eq:wanttailbound}
	x^p \le e^{x^2/(4 \smax^2)}\cdot\frac{\xi}{4 L }, ~~~ \forall x \not\in [-\tau,\tau],
	\end{equation} 
	 in which case 
	\begin{align*}
		\int_{[-\tau,\tau]^c} x^p \cdot \calF(x)\, \d x 
		= & ~ \sum^k_{i=1}p_i \int_{[-\tau,\tau]^c}x^p \cdot \frac{1}{\sqrt{2\pi}\sigma_i}\cdot e^{-x^2/(2\sigma_i^2)}\, \d x \\
		\le & ~ L \cdot \int_{[-\tau,\tau]^c}x^p\cdot\frac{1}{\sqrt{2\pi}\smax}\cdot e^{-x^2/(2\smax^2)}\, \d x \\
		\le & ~ L \cdot \int_{[-\tau,\tau]^c}\frac{1}{\sqrt{2\pi}\smax }\cdot e^{-x^2/(4\smax^2)} \cdot \frac{\xi}{4 L  }\, \d x \\
		\leq & ~ \xi/2
	\end{align*}
	where the first step follows from definition of $\calF(x)$, the second step follows from Fact~\ref{fact:monotone}, and the third step follows from Eq.~\eqref{eq:wanttailbound}.

	To reach \eqref{eq:wanttailbound}, we want 
	\begin{equation}\label{eq:wanttailbound2}
	p\ln x \le\frac{x^2}{4 \smax^2} - \ln(4 L / \xi), ~~~ \forall x\not\in[-\tau,\tau],
	\end{equation} 
 
	For $x\ge 8\smax^2\ln(4 L / \xi)$, we get that 
	\begin{align*}
	\frac{x^2}{4 \smax^2} -\ln(4 L /\xi)\ge \frac{x^2}{8\smax^2},
	\end{align*}
	and for $x\ge 8p\smax^2$, we have that 
	\begin{align*}
	p\ln x\le\frac{x^2}{8\smax^2}.
	\end{align*}
	We conclude that for $\tau = 8\overline{\sigma}^2\cdot\max\{p,\ln(4L /\xi)\}$, Eq.~\eqref{eq:wanttailbound2} holds.

	We can now complete the proof of \eqref{eq:moment_diff}. We may write $|\M_p(\calF) - \M_p(\calG')|$ as 
	\begin{align*}
		|\M_p(\calF) - \M_p(\calG')|
		= & ~ \left|\int^{\infty}_{-\infty} x^p \cdot \calF(x)\, \d x - \int^{\tau}_{-\tau}x^p\cdot\calG(x)\, \d x \right| \\
		\le & ~\left| \int_{[-\tau,\tau]^c}x^p\cdot \calF(x)\, \d x\right| + \left|\int^{\tau}_{-\tau} x^p\cdot(\calF - \calG)(x)\, \d x\right| \\
		\le & ~ \xi/2 + \left(\int^{\tau}_{-\tau}x^{2p}\, \d x\right)^{1/2}\cdot\norm{\calF - \calG}_2 \\
		= & ~ \xi/2 + \left(\frac{2\tau^{2p+1}}{2p+1}\eta\right)^{1/2} \le \xi,
	\end{align*} 
	where the second step follows from triangle inequality, the third step follows by \eqref{eq:tail} and the last step follows if we take $\eta= \frac{\xi^2p}{8\tau^{2p+1}}$.
\end{proof}
\noindent
This is useful as good estimates of high moments of the mixture allow us to approximate the maximum variance of any component well, as components with large variance contribute significantly more to the high moments than do the components with small variance.
However, we wish to estimate the minimum variance of our mixture.
We now show that we can do so by taking high moments of the \emph{Fourier transform} of our $L_2$-close estimator.
As an important subroutine, we show that it is efficient to compute the Fourier moments of our density estimate, by using the fact that it is piecewise polynomial.
Specifically:
\begin{lemma}
\label{lem:piecewise-fourier-moment}
Given the description of a $s$-piece degree-$d$ polynomial $p: \R \to \R$, and any $\tau > 0$ and nonnegative integer $\ell > 0$, there is an algorithm which runs in time $O(s \ell d^3)$ and which outputs
\[
\int_{-\tau}^\tau \widehat{p} [\omega] \omega^\ell \d \omega \; .
\]
\end{lemma}
\noindent
We defer the description of this algorithm as well as the proof of correctness to Appendix~\ref{sec:hermite}.
With this primitive, we can now show:

\begin{algorithm}[t]\caption{\textsc{EstimateMinVariance}($\calF,\overline{\sigma},\underline{\sigma},p,\delta$) Lemma~\ref{lem:estimate_min_variance}}\label{alg:estimate_min_variance}
\begin{algorithmic}[1]
	\State \textbf{Input}: Sample access to mixture of $k$ univariate Gaussians $\calF$, numbers $\overline{\sigma},\underline{\sigma}$ for which $\overline{\sigma}\ge \maxvar{\calF}$ and $\underline{\sigma}\le \minvar{\calF}$, degree $p\in\N$, failure probability $\delta > 0$
	\State \textbf{Output}: Estimate $\sigma^*$ for which \eqref{eq:min_pestimate} holds
		\State Let $\xi = (2\pi)^{-p-1/2}p^{p/2}\pmin\overline{\sigma}^{-p-1}$
		\State Let $L = \sqrt{2\pi}\cdot\underline{\sigma}^{-1}$
		\State Let $\tau = 8\overline{\sigma}^{2}\cdot \max(p,\ln(2L\sqrt{2}/\xi))$  
		\State Let $\eta = \frac{\xi^2p}{8\tau^{2p+1}}$
		\State Let $\calG = \text{\textsc{L2Estimate}}(\calF,\underline{\sigma},\eta,\delta)$ \Comment{Algorithm~\ref{alg:l2_estimate}, Corollary~\ref{cor:l2_estimate}}
		\State Explicitly compute $\M_p(\widehat{\calG})$ using Lemma~\ref{lem:piecewise-fourier-moment}
		\State \textbf{Return} $\sigma^*\triangleq \left(\frac{\M_p(\widehat{\calG})}{(2\pi)^{-p-1/2}p^{p/2}}\right)^{-1/(p-1)}$
\end{algorithmic}
\end{algorithm}

\begin{algorithm}[t]\caption{\textsc{CompareMinVariances}($\calF_1,\calF_2,\overline{\sigma},\underline{\sigma},\kappa_1,\kappa_2,\delta$), Corollary~\ref{cor:compare_min_variances}}\label{alg:compare_min_variances}
\begin{algorithmic}[1]
	\State \textbf{Input}: Sample access to two mixtures of $k$ univariate Gaussians $\calF_1,\calF_2$, numbers $\overline{\sigma},\underline{\sigma}$ for which $\overline{\sigma}\ge \max(\maxvar{\calF_1},\maxvar{\calF_2})$ and $\underline{\sigma}\le \min(\minvar{\calF_1},\minvar{\calF_2})$, tolerance parameters $\kappa_1 < \kappa_2$, failure probability $\delta > 0$
	\State \textbf{Output}: If the output is $\mathsf{true}$, then $\minvar{\calF_1}\ge(1 + \kappa_1)\minvar{\calF_2}$. 
	\State \hspace{16mm} If the output is $\mathsf{false}$, then $\minvar{\calF_1}\le (1 + \kappa_2)\minvar{\calF_2}$.
			\State Let $p = \Omega\left(\frac{\ln(1/\pmin)}{\kappa_2 - \kappa_1}\right)$.
			\State Let $\sigma^*_j = \text{\textsc{EstimateMinVariance}}(\calF_j,\overline{\sigma},\underline{\sigma},p,\delta)$ for $j = 1,2$. \Comment{Algorithm~\ref{alg:estimate_min_variance}, Lemma~\ref{lem:estimate_min_variance}}
			\If{$\frac{\M_p(\widehat{\calG}_1)}{\M_p(\widehat{\calG}_2)} > \frac{1}{2}\pmin(1 + \kappa_2)^{p-1}$}
				\State Output $\mathsf{true}$
			\Else
				\State Output $\mathsf{false}$
			\EndIf
\end{algorithmic}
\end{algorithm}

\begin{lemma}[Guarantee for \textsc{EstimateMinVariance}]\label{lem:estimate_min_variance}
	Let $\calF$ be a mixture of $k$ univariate zero-mean Gaussians with parameters $( \{ p_i \}_{ i \in [k] }, \{ \sigma_i \}_{ i \in [k] } )$. Let $\overline{\sigma}\ge\maxvar{\calF}$, $\underline{\sigma} \le \minvar{\calF}$. Then with probability at least $1 - \delta$, \textsc{EstimateMinVariance}($p,\calF,\overline{\sigma},\underline{\sigma},\delta$) (Algorithm~\ref{alg:estimate_min_variance}) takes 
\[
N = \pmin^{-4} k \ln ( 1 / \delta ) \cdot \poly \left( \overline{\sigma}, p, \ln (1 / \pmin ), \ln (1 / \underline{\sigma} ) \right)^{O(p)}
\]	
samples, runs in time $\widetilde{O}(N)$, and outputs a number $\sigma^*$ for which 
	\begin{equation}\label{eq:min_pestimate}
	\left(\frac{3}{4}\right)^{1/(p-1)} \cdot \minvar{\calF} \le \sigma^* \le \left(\frac{3}{2\pmin}\right)^{1/(p-1)} \cdot \minvar{\calF}.
	\end{equation} 
\end{lemma}

\begin{proof}
	Note that in the pseudocode our choice of $\eta$ is given by
	\begin{equation}
	\eta \triangleq \frac{\xi^2 p}{8\left(8\overline{\sigma}^2\max\left(p,\ln\left(\frac{4\sqrt{\pi}}{\underline{\sigma}\xi}\right)\right)\right)^{2p+1}}, \ \ \ \xi \triangleq  (2\pi)^{-p-1/2}p^{p/2}\pmin\overline{\sigma}^{-p-1} \label{eq:etadef} \; .
	\end{equation}
\noindent
	Consequently, the runtime and sample complexity bounds for general $p$ just follow from the fact that these quantities are dominated by the cost of running \textsc{L2Estimate}($\mathcal{F},\underline{\sigma},\eta,\delta$) for $\eta$ as defined in \eqref{eq:etadef}. By Corollary~\ref{cor:l2_estimate}, if we run \textsc{L2Estimate}($\calF,\underline{\sigma},\eta,\delta$) and produce the piecewise polynomial $\calG$, we know that $\norm{\calF - \calG}^2_2 \le \eta$. By Plancherel's theorem \cite{pl10}, this means that $\|\widehat{\calF} - \widehat{\calG} \|^2_2 \le \eta$. To apply Lemma~\ref{lem:moment_diff}, first note that by Fact~\ref{fact:fouriergaussian}, \begin{equation}\widehat{\calF}(\omega) = \sum^k_{i=1}p_i \frac{1}{\sqrt{2\pi}\sigma_i}\N\left(0,\frac{1}{4\pi^2 \sigma_i^2},\omega\right).\end{equation} So $\widehat{\calF}$ is an affine linear combination of Gaussian densities, and its coefficients sum to 
	\begin{equation*}
	\sum^k_{i=1}p_i\cdot\frac{1}{\sqrt{2\pi}\sigma_i} \le\frac{1}{\sqrt{2\pi}\minvar{\calF}} \le L,
	\end{equation*} 
	where the last step follows by our choice of $L$ in \textsc{EstimateMinVariance}.

	So by Lemma~\ref{lem:moment_diff}, if we define $\widehat{\calG}'$ by 
	\begin{equation*}
		\widehat{\calG}'(x) = \bone{x\in[-\tau,\tau]}\cdot\widehat{\calG}(x),
	\end{equation*} 
	then we get that 
	\begin{equation}\label{eq:mpgj_dev}
	|\M_p(\widehat{\calF}) - \M_p(\widehat{\calG}')| \le \xi.
	\end{equation} 
	Furthermore, note that 
	\begin{align}\label{eq:mpfj}
		\M_p(\widehat{F})\le & ~ \sum^k_{i = 1} p_i \cdot \frac{1}{\sqrt{2\pi}\sigma_i}\cdot p^{p/2}\cdot\left(\frac{1}{4\pi^2(\sigma_i)^2}\right)^{p/2} \notag \\
		= & ~ (2\pi)^{-p-1/2}p^{p/2}\sum^k_{i=1}p_i\left(\sigma_i\right)^{-p-1}.
	\end{align} If we had $\xi = (2\pi)^{-p-1/2}p^{p/2}\xi'$ for some $\xi'> 0$, then we get by \eqref{eq:mpgj_dev} and \eqref{eq:mpfj} that 
	\begin{equation*}
		\M_p(\widehat{\calG}) = (2\pi)^{-p-1/2}p^{p/2}\left[\sum^k_{i=1}p_i\left(\sigma_i\right)^{-p-1} \pm \xi'\right].
	\end{equation*} 
	If we take $\xi' \triangleq\frac{1}{3}\pmin\overline{\sigma}^{-p-1}$, then observe that because \begin{equation}
		\pmin\cdot \minvar{\calF}^{-p-1} \le \sum^k_{i=1}p_i(\sigma_i)^{-p-1} \le \minvar{\calF}^{-p-1},
	\end{equation} we have \begin{equation}
		\minvar{\calF}\le \left(\sum^k_{i=1}p_i(\sigma_i)^{-p-1}\right)^{-1/(p-1)} \le \pmin^{-1/(p-1)}\cdot \minvar{\calF},
	\end{equation} so $\sigma^*\triangleq \left(\frac{\M_p(\widehat{\calG})}{(2\pi)^{-p-1/2}p^{p/2}}\right)^{-1/(p-1)}$ satisfies \eqref{eq:min_pestimate}.

	For the last part of the lemma, take $p = 20\ln\left(\frac{3}{2\pmin}\right)+1\ge 4$. It is straightforward to check the sample and time complexity bounds for this choice of $p$, and the bound on $\sigma^*$ follows from the fact that $\left(3/4\right)^{1/(p-1)}\ge 0.9$ for $p\ge 4$ and 
	\begin{align*}
	\left(\frac{3}{2\pmin}\right)^{\left(20\ln\left(\frac{3}{2\pmin}\right)\right)^{-1}} = e^{1/20} \le 1.1.
	\end{align*}
\end{proof}
\noindent
We now identify two specific parameter settings for this algorithm which will be useful later on.
First, if we take the degree $p$ to be relatively small, we are able to get a constant approximation to the minimum variance very efficiently:
\begin{corollary}
\label{cor:constant_min_approx}
Let $p = \Theta ( \ln ( 1 / \pmin ) )$.	
	Then, the algorithm $\textsc{EstimateMinVariance}(p, \calF, \overline{\sigma}, \underline{\sigma}, \delta)$ has sample and time complexity 
	\begin{align*}
	\tilde{O}\left( \pmin^{-4} k \ln (1/\delta) \cdot \poly \left( \overline{\sigma}, p, \ln ( 1 / \pmin ), \ln ( 1 / \underline{\sigma} ) \right)^{O(p)} \right),
	\end{align*}
	and the output $\sigma^*$ satisfies
	\begin{align*}
	0.9 \cdot \minvar{\calF} \le \sigma^* \le 1.1 \cdot \minvar{\calF}.
	\end{align*}
\end{corollary}
\noindent
We also have:
\begin{corollary}[Guarantee for \textsc{CompareMinVariances}]\label{cor:compare_min_variances}
	Let $0 < \kappa_1 < \kappa_2 \le 1$, and let $\calF_1$ and $\calF_2$ be two mixtures of $k$ univariate zero-mean Gaussians with parameters $(\{p_i\},\{\sigma^{(1)}_i\})$ and $(\{p_i\},\{\sigma^{(2)}_i\})$ respectively. Let $\overline{\sigma}\ge\max(\maxvar{\calF_1},\maxvar{\calF_2})$, $\underline{\sigma} \le \min(\minvar{\calF_1},\minvar{\calF_2})$. 
	Then, with probability $1 - \delta$, the algorithm \\
	\textsc{CompareMinVariances}($\calF_1,\calF_2,\overline{\sigma},\underline{\sigma},\kappa_1,\kappa_2,\delta$) (Algorithm~\ref{alg:compare_min_variances}) satisfies:
	\begin{itemize}
	\item
	If $\minvar{\calF_1} \ge \left(1 + \kappa_2\right)\minvar{\calF_2}$,  then it outputs $\mathsf{true}$.  
	\item
	If $\minvar{\calF_1} \le (1 + \kappa_1)\minvar{\calF_2}$, then it outputs $\mathsf{false}$.
	\end{itemize}
	Moreover, this algorithm takes 
	\[
	N =  \pmin^{-4} k \ln( 1/ \delta ) \cdot \poly \left( \overline{\sigma}, (\kappa_2 - \kappa_1)^{-1}, \ln ( 1 / \pmin ), \ln ( 1 / \underline{\sigma} ) \right)^{O((\kappa_2 - \kappa_1)^{-1} \ln (1 / \pmin))} 
	\]
	samples and runs in time $\widetilde{O} (N)$.
\end{corollary}

\begin{proof}
	Let $\sigma^*_j$ be the estimate produced by 
	\begin{align*}
	\textsc{EstimateMinVariance}(p,\calF_j,\overline{\sigma},\underline{\sigma},\delta), ~~~ \forall j = 1,2.
	\end{align*}
	By Lemma~\ref{lem:estimate_min_variance}, we have that
	\begin{equation}\label{eq:ratiobound}
		\frac{1}{2} \cdot \pmin \cdot \left(\frac{\minvar{\calF_1}}{\minvar{\calF_2}}\right)^{p-1} \le \left(\frac{\sigma^*_1}{\sigma^*_2}\right)^{p-1} \le 2 \cdot p^{-1}_{\min} \cdot \left(\frac{\minvar{\calF_1}}{\minvar{\calF_2}}\right)^{p-1}.
	\end{equation} Now for the first part of the lemma, by hypothesis $\frac{\minvar{\calF_1}}{\minvar{\calF_2}} \geq 1 + \kappa_2$, so by the lower bound in \eqref{eq:ratiobound} we conclude that 
	\begin{equation*}
		\left(\frac{\sigma^*_1}{\sigma^*_2}\right)^{p-1} \ge \frac{1}{2}\pmin\left(1 + \kappa_2\right)^{p-1}.\label{eq:kappa2}
	\end{equation*}
	For the second part of the lemma, by hypothesis $\frac{\minvar{\calF_1}}{\minvar{\calF_2}} \le 1 + \kappa_1$, so by the upper bound in \eqref{eq:ratiobound} we conclude that 
	\begin{equation*}
		\left(\frac{\sigma^*_1}{\sigma^*_2}\right)^{p-1} \le 2p^{-1}_{\min}\left(1 + \kappa_1\right)^{p-1}.\label{eq:kappa1}
	\end{equation*} The content of this lemma is that the right-hand side of \eqref{eq:kappa2} is strictly greater than the right-hand side of \eqref{eq:kappa1}. Indeed, we need to check that \begin{equation}\left(\frac{1 + \kappa_2}{1 + \kappa_1}\right)^{p-1}\ge 4\pmin^{-2}.\end{equation} But if we take $p-1 = \frac{2\log_2(4/\pmin^2)}{\kappa_2 - \kappa_1}$, then by the fact that $1 + \kappa_1 \le 2$ and $\frac{\kappa_2 - \kappa_1}{1 + \kappa_1} \le 1$, and by the elementary inequality $(1 + 1/x)^x\ge 2$ for $x\ge 1$, we conclude that \begin{equation*}
		\left(\frac{1 + \kappa_2}{1 + \kappa_1}\right)^{p-1} = \left(1 + \frac{\kappa_2 - \kappa_1}{1 + \kappa_1}\right)^{p-1} \ge 8\pmin^{-1}
	\end{equation*} as desired.
\end{proof}

\subsection{Moment Descent}
\label{subsec:fmd}

In this section we will show how to obtain a warm start using the \textsc{CompareMinVariances} subroutine of the previous section.

The first ingredient we need is a subroutine to estimate $\Span(\{w_i - a\}_{i\in[k]})$, where $a$ is our current guess for a direction. 
For any $x, y, a \in \R^d$, define the matrix 
\begin{equation}
	\vec{M}^{x,y}_a \triangleq \frac{1}{2}\left[(y - \langle a,x\rangle)^2 xx^{\top} - (y -  \langle a,x\rangle)^2 \cdot\Id\right] ~~~ \in \R^{d \times d}\label{eq:mxydef}
\end{equation} and let \begin{equation}\widehat{\vec{M}}^{(N)}_a \triangleq \frac{1}{N}\sum^N_{i=1}\vec{M}^{x_i,y_i}_a\label{eq:empiricalM}\end{equation} for $(x_1,y_1),...,(x_N,y_N)$ i.i.d. samples from $\calD$.
Notice there is a matrix-vector oracle for $\widehat{\vec{M}}^{(N)}_a$ which runs in time $O(N d)$.

We then have:
\begin{lemma}[$\widehat{\vec{M}}^{(N)}_a$ approximates $\Span(\{w_i - a_t\})$]
\label{lem:findspan}
Let $\calD$ be a mixture of $k$ spherical linear regressions. Then for any $a\in\R^d$, we have that 
\begin{equation*}
	\E_{(x,y)\sim\calD}\left[\vec{M}^{x,y}_a\right] = \sum^k_{i=1}p_i (w_i - a)(w_i - a)^{\top}.
\end{equation*} 
Furthermore, for any $\beta,\delta>0$ and 
\begin{equation*}
	N = \tilde{\Omega}\left( \max_{i \in [k]} \norm{w_i - a}^2_2 \cdot p^{-1}_{\min} \cdot \beta^{-2} \cdot d \cdot \ln(k/\delta))\right),
\end{equation*} 
we have that 
\begin{equation}
	\Pr\left[\norm{\widehat{\vec{M}}^{(N)}_a - \E_{(x,y)\sim\calD}\left[\vec{M}^{x,y}_a\right]}_2 \ge \beta\right] \le \delta
\end{equation}
\end{lemma}

We emphasize that $\E[\vec{M}^{x,y}_a]$ is the same regardless of $\eta$, but the value of $\eta$ will slightly affect concentration, though not in the regimes in which we will apply Lemma~\ref{app:findspan}.
We defer the proof of this lemma to Appendix~\ref{app:findspan}.
Combining this with Fact~\ref{thm:power-method} allows us to quantify the effectiveness of approximate $k$-SVD of an empirical estimate for $\E[\vec{M}^{x,y}_a]$ for capturing the span of the $w_i$:
\begin{lemma}[Correlation of the top principal subspace]\label{lem:correlation}
	Let 
	\begin{align*}
	N=\tilde{\Omega}\left(\frac{ \max_{ i \in [k] } \| w_i - a \|^2_2 }{ \min_{i \in [k]} \| w_i - a \|^2_2 }\cdot p^{-2}_{\min} \cdot k^2\cdot d \cdot \ln(k/\delta)\right).
	\end{align*} 
	Then $\textsc{ApproxBlockSVD} (\widehat{\vec{M}}^{(N)}, 1/10, \delta/2)$ runs in time $\tilde{O} \left( k \cdot N \cdot d \right)$
	and outputs a matrix $\vec{U}$ so that with probability at least $1 - \delta$, 
	\begin{equation} 
		\frac{1}{2}\le \frac{\norm{\vec{U}^{\top}(w_i - a)}_2}{\norm{w_i - a}_2} \le 1.\label{eq:correlation_subspace}
	\end{equation}
\end{lemma}

\begin{proof}
	The upper bound is trivial. 
	We now prove the lower bound.
	We first observe that Lemma~\ref{lem:findspan} implies that $\gap_k (\widehat{\vec{M}}^{(N)}_a)^{-1} \geq \poly (1 / \pmin, 1 / \min_{i \in [k]} \| w_i - a \|^2_2)$, which proves the runtime claim.
	The lower bound follows by taking $\beta$ in Lemma~\ref{lem:findspan} to be $\beta = \frac{1}{8k} \cdot p^{1/2}_{\min} \cdot \min_{i \in [k]} \|w_i - a \|_2$ and applying Fact~\ref{thm:power-method}.
	Finally, to demonstrate the runtime, observe that there is a matrix-vector oracle for $\widehat{\vec{M}}^{(N)}_a$ which runs in time $O(N d)$.
\end{proof}

We are now ready to analyze the amount of progress each step of moment descent makes.

\begin{lemma}[Progress of moment descent per step]
	For any $\delta > \exp(-\sqrt{k})$, the following holds. Let $\sigma^2_t \triangleq \min_{i\in[k]} \norm{w_i - a_t}_2$. Denote the minimizing index $i$ by $i^*$. For $M \triangleq e^{\sqrt{k}}\ln(2/\delta)$ and $g_1,...,g_M \sim \N(0,\Id_k)$, let $v_j = \frac{\vec{U}g_i}{\norm{\vec{U}g_i}_2} \in \S^{d-1}$ for $j\in[M]$. Let $\sigma^*$ be a number for which $0.9\sigma_t\le \sigma^*\le 1.1\sigma_t$. Let $\eta = \frac{1}{2}k^{-1/4}\sigma^*$.

	Then we have that with probability at least $1 - \delta$, \begin{enumerate}
	\item There exists at least one $j\in[M]$ for which $\norm{w_{i^*} - a_t - \eta \cdot v_j}^2_2 \le \left(1 - \frac{1}{5\sqrt{k}}\right)\sigma_t^2$.
	\item For all $j\in[M]$ and $i\in[k]$, $\norm{w_i - a_t - \eta \cdot v_j}^2_2 \ge \left(1 - \frac{9}{\sqrt{k}}\right)\sigma_t^2$.
	\end{enumerate}\label{lem:progress}
\end{lemma}

\begin{proof}
	For any $i\in[k]$, we may write 
	\begin{align}
		\norm{w_i - a_t - \eta v_j}^2_2 
		= & ~ \norm{w_i - a_t}^2_2 + \eta^2\norm{v_j}^2_2 - 2\eta\langle w_i - a_t,v_j\rangle \nonumber\\
		= & ~ \norm{w_i - a_t}^2_2 + \eta^2 - 2\eta\langle w_i - a_t,v_j\rangle. \label{eq:onestep}
	\end{align} 
	Define $\tilde{w}_i\triangleq \frac{w_i - a_t}{\norm{w_i - a_t}_2}$. For every $j\in[M]$, let $A_j$ be the event that $\langle v_j,\tilde{w}_{i^*}\rangle \ge\frac{1}{2}k^{-1/4}$. For every $j\in[M]$ and $i\in[k]$, let $B_j[i]$ be the event that $\langle v_j,\tilde{w}_i\rangle \le 3k^{-1/4}$. We would like to condition on the event that $\calE\triangleq\left(\bigvee_{j\in [M]} A_j\right)\wedge\left(\bigwedge_{i \in [k], j\in[M]} B_j[i]\right)$.

	We first verify that conditioned on $\calE$, 1) and 2) of the lemma hold. We get that there is at least one $j\in[M]$ for which 
	\begin{align*}
		\norm{w_{i^*} - a_t - \eta v_j}^2_2 \le & ~ \norm{w_{i^*} - a_t}^2_2 + \eta^2 - \eta\sigma_t\cdot k^{-1/4} \\
		\le & ~ \left(1 + \frac{1}{4}\left(\left(\frac{\sigma^*}{\sigma_t}\right)^2 - 2\left(\frac{\sigma^*}{\sigma_t}\right)\right)k^{-1/2}\right)\sigma^2_t\\
		\le & ~ \left(1 - \frac{1-0.1^2}{4\sqrt{k}}\right)\sigma^2_t \le \left(1 - \frac{1}{5\sqrt{k}}\right)\sigma_t^2.
	\end{align*} 
	where the first step follows from the fact that we have conditioned on $A_j$ and also $\norm{w_{i^*} - a_t}_2 = \sigma_t$ by definition, and the third step follows from $\sigma^*/\sigma_t \in [0.9,1.1]$.

	For every $i\in[k], j\in[M]$ we have that 
	\begin{equation*}
		\norm{w_i - a_t - \eta v_j}^2_2 \ge \norm{w_i - a_t}^2 + \eta^2 - 6\eta\sigma_t\cdot k^{-1/4} \ge \left(1 - \frac{9}{\sqrt{k}}\right)\sigma^2_t,
	\end{equation*} 
	where the first step follows from the fact that we have conditioned on $B_j[i]$ and also $\norm{w_i - a_t}_2 \le \sigma_t$ by definition.

	Finally, we show that $\Pr[\calE]\ge 1 - \delta$. For any $j\in[M]$ and $i\in[k]$, by Corollary~\ref{cor:unitcorr}, with probability at least $e^{-\sqrt{k}}$ we have that 
	\begin{equation}
		\left\langle\frac{g_j}{\norm{g_j}_2}, \frac{\vec{U}^{\top}(w_{i} - a_t)}{\norm{\vec{U}^{\top}(w_{i} - a_t)}_2}\right\rangle \ge k^{-1/4}.\label{eq:innerprod}
	\end{equation} 
	Because 
	\begin{equation*}
	\langle g_j, \vec{U}^{\top}(w_{i} - a_t)\rangle = \langle \vec{U} g_j, w_{i} - a_t\rangle,
	\end{equation*} 
	and $\norm{\vec{U} g_j}_2 =\norm{g_j}_2$ by orthonormality of the columns of $\vec{U}$, we can rewrite the left-hand side of \eqref{eq:innerprod} as 
	\begin{equation*}
		\frac{\langle v_j, w_{i} - a_t\rangle}{\norm{\vec{U}^{\top}(w_{i} - a_t)}_2} \le 2\langle v_j,\tilde{w}_{i}\rangle,
	\end{equation*} 
	where the inequality follows by the lower bound in \eqref{eq:correlation_subspace}. So by taking $i = i^*$, we conclude that $\Pr[A_j] \ge e^{-\sqrt{k}}$. The probability that $\bigvee_{ j \in [M] } A_j$ does not occur is thus 
	\begin{equation*}
	\Pr\left[\bigwedge_{j \in [M]} \overline{A}_j\right] \le \left(1 - e^{-\sqrt{k}}\right)^M. 
	\end{equation*}
	On the other hand, by the same analysis, this time invoking the \emph{second} part of Corollary~\ref{cor:unitcorr} and the \emph{upper} bound in \eqref{eq:correlation_subspace}, we see that $\Pr[B_j[i]] \ge e^{-3\sqrt{k}}$, so the probability that $\bigwedge B_j[i]$ does not occur is 
	\begin{equation*}
		\Pr\left[\bigvee_{i\in[k], j \in [M]} \overline{B}_j[i]\right] \le kM \cdot e^{-3\sqrt{k}}.
	\end{equation*} 
	So by taking $M = e^{\sqrt{k}}\ln(2/\delta)$ and noting that for this choice of $M$, $kMe^{-3\sqrt{k}} < \delta/2$ because $\delta > e^{-\sqrt{k}}$, we get that $\Pr[\calE]\ge 1 - \delta$ as claimed.
\end{proof}

\begin{lemma}
	There is an absolute constant $C > 0$ for which the following holds.
	Let $\calD$ be a mixture of spherical linear regressions with mixing weights $\{p_i\}$, directions $\{w_i\}$, and noise rate $\noise$.
	For any $\epsilon,\delta > 0$ and $\noise^2 \le\epsilon^2/10$, with probability at least $1 - \delta$, \textsc{FourierMomentDescent}($\calD,\delta,\epsilon$) (Algorithm~\ref{alg:fourier_moment_descent}) outputs direction $a_T\in\R^d$ for which $\min_{i\in[k]}\norm{w_i - a_T}_2 \le \epsilon$.\label{lem:fourier_moment_descent}
\end{lemma}

\begin{proof}
	Let $\sigma_t\triangleq \min_{i\in[k]}\norm{w_i - a_t}_2$. Because $a_0 = 0$, we have that $\sigma_0 \le \max_{i\in[k]}\norm{w_i}_2 \le 1$.

	By a simple union bound, we first upper bound the probability that the steps of moment descent in the $t$-th iteration of the outer loop in \textsc{FourierMomentDescent} all succeed.

	\begin{claim}
		Let $i\in[S]$. With probability at least $1 - \delta$, the randomized components of the $t$-th iteration of the outer loop in \textsc{FourierMomentDescent} all succeed.\label{claim:allsuccess} 
	\end{claim}

	\begin{proof}
		Each $t$-th iteration of the outer loop in \textsc{FourierMomentDescent} (Algorithm~\ref{alg:fourier_moment_descent}) has the following randomized components: computing $\widehat{\vec{M}}^{(N_1)}_{a_t}$, running \textsc{EstimateMinVariance} (Algorithm~\ref{alg:estimate_min_variance}), trying the Gaussian vectors $g$ in the inner loop over $j\in[M]$, running \textsc{ApproxBlockSVD}, and running \textsc{CompareMinVariances} (Algorithm~\ref{alg:compare_min_variances}) in this inner loop.

		Because the failure probability $\delta'$ for the first four of these tasks was chosen to be $\frac{\delta}{5T}$, and the failure probability $\delta''$ for the last task was chosen to be $\frac{\delta}{5MT}$, we can bound the overall failure probability by $\delta$.
	\end{proof}

	Call the event in Claim~\ref{claim:allsuccess} $\calE$. Next, we show that provided $\calE$ occurs, $\sigma_t$ can be naively bounded by a constant.

	\begin{claim}
		Let $0\le t < T$ and condition on $\calE$. Then $\sigma_t \le 4$.
	\end{claim}

	\begin{proof}
		At the start of the $t$-th step, our initial estimate $a_{t-1}$ is at distance at most 1 from some $w_{i^*}$ (this is a very loose bound). After the $t$-th step, the new estimate $a_t$ satisfies $\norm{w_{i^{**}} - a_t}_2 \le \norm{w_{i^*}- a_{t-1}}_2\le 1$ for some $i^{**}\in[k]$.
		So we have that $\norm{w_i - a_t}_2 \leq \norm{w_{i^{**}} - w_i}_2 + \norm{w_{i^{**}} - a_t}_2 \leq 3$. Recalling that $\sigma^2_t = \min_{i\in[k]}\norm{w_i - a_t}^2_2 + \noise^2$ and noting that $\noise < 1$, we conclude that $\overline{\sigma} = 4$ is a valid upper bound on the standard deviation of any component of any univariate mixture of Gaussians $\calF_t$ or $\calF'^{(j)}_t$ encountered during the course of \textsc{FourierMomentDescent}.	
	\end{proof}

	Next, we show that provided $\calE$ occurs, then we can bound the extent to which every iteration of the outer loop in \textsc{FourierMomentDescent} contracts $\sigma^2_t$.

	\begin{claim}
		Let $0\le t<T$ and condition on $\calE$. Suppose $\noise^2 \le \frac{1}{5}\norm{w_i - a_t}^2_2$ for any $i\in[k]$. Then \begin{enumerate}
			\item {\bf(Completeness)} Either $\norm{w_i - a_t}_2\le \epsilon$ already, or there exists some $j\in[M]$ for which \textsc{CompareMinVariances}($\calF_t,\calF'^{(j)}_t,\overline{\sigma},\underline{\sigma},\kappa,2\kappa,\delta''$) outputs $\mathsf{true}$ for $\kappa = \frac{1}{24\sqrt{k}}$.
			\item {\bf(Soundness)} For any such $j\in[M]$ for which \textsc{CompareMinVariances} outputs $\mathsf{true}$, \begin{equation}
				\left(1 - \frac{9}{\sqrt{k}}\right)\sigma^2_t \le \sigma^2_{t+1} \le \left(1 - \frac{1}{48\sqrt{k}}\right)\sigma^2_t.\label{eq:bounds_iter}
			\end{equation}
		\end{enumerate}\label{claim:complete_sound}
	\end{claim}

	\begin{proof}
		We first show completeness. Suppose $\norm{w_i - a_t}_2 \ge \epsilon$ for all $i\in[k]$. By the first part of Lemma~\ref{lem:progress}, there exists some $j\in[M]$ for which 
		\begin{align*}
			\min_{i\in[k]}\norm{w_i - a'^{(j)}_t}^2_2 \le \left(1 - \frac{1}{5\sqrt{k}}\right)\min_{i\in[k]}\norm{w_i - a_t}^2_2
		\end{align*} 
		and therefore 
		\begin{align*}
			\noise^2 + \min_{i\in[k]}\norm{w_i - a'^{(j)}_t}^2_2 
			\le & ~\noise^2 + \left(1 - \frac{1}{5\sqrt{k}}\right)\min_{i\in[k]}\norm{w_i - a_t}^2_2 \\
			\le & ~ \left(1 - \frac{1}{6\sqrt{k}}\right)\left(\noise^2 + \min_{i\in[k]}\norm{w_i - a_t}^2_2\right),
		\end{align*} 
		where in the last step we used the assumption that $\noise^2 \le \frac{1}{5}\norm{w_i - a_t}^2_2$ for any $i\in[k]$.

		We conclude that $\sigma^{(j)}_t \triangleq \min_{i \in [k]}\left\{\noise^2 + \norm{w_i - a'^{(j)}_t}_2\right\}$ satisfies $(\sigma^{(j)}_t)^2 \le \left(1 - \frac{1}{6\sqrt{k}}\right)(\sigma_t)^2$, and because $1 - \frac{1}{6\sqrt{k}} \le \left(\frac{1}{1 + 2\kappa}\right)^2$ for $\kappa = \frac{1}{24 \sqrt{k}}$, \textsc{CompareMinVariances}($\calF_t,\calF'^{(j)}_t,\overline{\sigma},\underline{\sigma},\kappa,2\kappa,\delta''$) would return $\mathsf{true}$, completing the proof of completeness.

		For soundness, if \textsc{CompareMinVariances}($\calF_t,\calF'^{(j)}_t,\overline{\sigma},\underline{\sigma},\kappa,2\kappa,\delta''$) returns $\mathsf{true}$ for some $j\in[M]$, by Corollary~\ref{cor:compare_min_variances} this means 
		\begin{align*}
		(\sigma^{(j)}_t)^2 \leq ( 1 + \kappa )^{-2} \cdot \sigma^2_t \leq ( 1 - \kappa/2 ) \cdot \sigma_t^2 \leq \Big( 1 - \frac{1}{48 \sqrt{k}} \Big) \cdot \sigma_t^2,
		\end{align*}
		where the second step follows from $\kappa \in (0,1)$, which gives the upper bound in \eqref{eq:bounds_iter}.

		Finally, for the lower bound in \eqref{eq:bounds_iter}, note that the second part of Lemma~\ref{lem:progress} tells us that \begin{equation}
			\min_{i\in[k]}\norm{w_i - a'^{(j)}_t}^2_2 \ge \left(1 - \frac{9}{\sqrt{k}}\right)\norm{w_i - a_t}^2_2
		\end{equation} and therefore \begin{equation}
			\noise^2 + \min_{i\in[k]}\norm{w_i - a'^{(j)}_t}^2_2 \ge \noise^2 + \left(1 - \frac{9}{\sqrt{k}}\right)\norm{w_i - a_t}^2_2 \ge \left(1 - \frac{9}{\sqrt{k}}\right)\left(\noise^2 + \norm{w_i - a_t}^2_2\right).
		\end{equation}
	\end{proof}

	We are now ready to complete the proof of Lemma~\ref{lem:fourier_moment_descent}. Let $\rho = 1.1/0.9$ and condition on $\calE$. 

	If there does not exist $0\le t < T$ for which we have that 
	\begin{equation}
		\min_{i\in[k]} \norm{w_i - a_t}^2_2 \le \epsilon^2/\rho^2 - \noise^2\label{eq:crossed},
	\end{equation} then because $\noise^2\le \epsilon^2/10$, we get that $\norm{w_i - a_t}^2_2 \ge \epsilon^2/2$. So $\noise^2 \le \frac{1}{5}\norm{w_i - a_t}^2_2$, and by completeness and soundness in Claim~\ref{claim:complete_sound}, $\sigma^2_t$ has contracted by at least a factor of $(1 - 1/48\sqrt{k})$ and by at most a factor of $(1 - 9/\sqrt{k})$ at every step. So if we take $T = \Omega(\sqrt{k}\cdot\ln(1/\epsilon))$, we are guaranteed that 
	\begin{equation}\min_{i\in [k]}
	\norm{w_i - a_T}_2 \le \sigma_T \le \epsilon.
	\label{eq:sigmabound}
	\end{equation}

	On the other hand, if \eqref{eq:crossed} holds for some $0\le t < T$, then \begin{equation}
	(\sigma^*_t)^2 \le 1.21\sigma^2_t \le 1.21 \cdot \Big(\min_{i \in [k]}\norm{w_i - a_t}^2_2 + \noise^2 \Big) \le 0.99^2\epsilon^2,
	\end{equation} 
	so \textsc{FourierMomentDescent} breaks out at Line~\ref{line:break} and correctly outputs $a_t$.

	Conversely, if \textsc{FourierMomentDescent} breaks out at Line~\ref{line:break} because $\sigma^*_t \le 0.99\epsilon$, this implies that 
	\begin{align*}
	\min_{i \in [k]} \norm{w_i - a_t}^2 + \noise^2 \le (\sigma^*_t)^2 / 0.99^2 \le \epsilon^2,
	\end{align*}
	so $\min_{i\in [k]} \norm{w_i - a_t}_2 \le \epsilon$.

	The last thing to check is that $\underline{\sigma} = \epsilon/3$ is always a valid lower bound for any $\sigma_t$. If \eqref{eq:crossed} holds for some $t$, $t$ is necessarily the first (and last) $t$ in \textsc{FourierMomentDescent} for which \eqref{eq:crossed} holds because of Line~\ref{line:break}. So it must be that \begin{equation}\sigma_{t-1}^2 \ge \min_{i \in [k]} \norm{w_i - a_{t-1}}^2_2 > \epsilon^2/\rho^2 - \noise^2 \ge \epsilon^2\cdot\left((0.9/1.1)^2 - 1/5\right) \ge 0.4\epsilon^2\end{equation} and thus, by the fact that $\sigma_t \ge (1 - 9/\sqrt{k})\sigma_{t-1}\ge 0.99\sigma_{t-1}$, we conclude that $\sigma_t > \epsilon/3$ as desired.
\end{proof}

\begin{algorithm}\caption{\textsc{FourierMomentDescent}$(\calD,\delta,\epsilon)$, Lemma~\ref{lem:fourier_moment_descent}}\label{alg:fourier_moment_descent}
\begin{algorithmic}[1]
	\State \textbf{Input}: Sample access to mixture of linear regressions $\calD$ with separation $\Delta$ and noise rate $\noise$, failure probability $\delta$, error $\epsilon$
	\State \textbf{Output}: $a_T\in\R^d$ satisfying $\min_{i\in[k]}\norm{w_i - a_T}_2 \le\epsilon$, with probability at least $1 - \delta$.
		\State Set $a_0 = 0$, $T = \Omega(\sqrt{k}\cdot\ln(1/\epsilon))$.
		\State Set $\delta' = \frac{\delta}{5T}$.
		\State Set $M = e^{\sqrt{k}}\ln(2/\delta')$.
		\State Set $\delta'' = \frac{\delta}{5MT}$.
		\State Set $\overline{\sigma} = 4$ and $\underline{\sigma} = \epsilon/3$.
		\For{$0\le t < T$}
			\State \multiline{Let $\calF_t$ be the univariate mixture of Gaussians which can be sampled from by drawing $(x,y)\sim\calD$ and computing $y - \langle x,a_t\rangle$.}
			\State Let $p = 20\ln\left(\frac{3}{2\pmin}\right)+1$.
			\State Let $\kappa = \frac{1}{24\sqrt{k}}$.
			\State $\sigma^*_t \triangleq$~\textsc{EstimateMinVariance}($\calF_t,\overline{\sigma}, \underline{\sigma},p,\delta'$). \Comment{Algorithm~\ref{alg:estimate_min_variance}}
			\If{$\sigma^*_t < 0.99\epsilon$}
				\State Output $a_t$. \label{line:break}
			\EndIf
			\State Let $N_1 \triangleq \tilde{\Omega}\left(\frac{\overline{\sigma}^2}{(\sigma^*_t)^2}\cdot p^{-2}_{\min}\cdot k^2\cdot d \cdot \ln(k/\delta')\right)$.
			\State \multiline{Draw $N_1 $ i.i.d. samples $\{(x_i,y_i)\}_{i\in[N_1]}$ from $\calD$ and form the matrix $\widehat{\vec{M}}^{(N_1)}_{a_t}$.} \label{line:n1}
			\State Let $\vec{U}_t = \textsc{ApproxBlockSVD} (\widehat{\vec{M}}^{(N_1)}_{a_t}, 1/10, \delta')$. \Comment{Lemma~\ref{lem:correlation}}
			\For{$j\in[M]$}
				\State Sample $g^{(j)}_t\sim\N(0,\Id_k)$ and define $v^{(j)}_t = \frac{\vec{U}_t g^{(j)}_t}{\norm{\vec{U}_t g^{(j)}_t}_2}\in\S^{d-1}$.
				\State Let $a'^{(j)}_t = a_t + \eta_t v_j$ for $\eta_t \triangleq \frac{1}{2}k^{-1/4}\cdot\sigma^*_t$.
				\State \multiline{Let $\calF'^{(j)}_t$ be the univariate mixture of Gaussians which can be sampled from by drawing $(x,y)\sim\calD$ and computing $y - \langle x,a'^{(j)}_t\rangle$.}
				\If{\textsc{CompareMinVariances}($\calF_t,\calF'^{(j)}_t,\overline{\sigma},\underline{\sigma},\kappa,2\kappa,\delta''$) = $\mathsf{true}$} \Comment{Algorithm~\ref{alg:compare_min_variances}, Corollary~\ref{cor:compare_min_variances}}
					\State Set $a_{t+1} = a'^{(j)}_t$
					\Break
				\EndIf
			\EndFor
		\EndFor
		\State Output $a_T$.
\end{algorithmic}
\end{algorithm}
\noindent
Lastly, we calculate the runtime and sample complexity of \textsc{FourierMomentDescent}.

\begin{lemma}[Running time of \textsc{FourierMomentDescent}]
Let
\begin{align} 
N_1 &= \widetilde{O} ( \epsilon^{-2} \pmin^{-2}  d k^2 \ln(1 / \delta) ) \label{eq:N1def}\\
N &=   \pmin^{-4} k \ln ( 1 / \delta ) \cdot \poly \left(\sqrt{k}, \ln ( 1 / \pmin ), \ln (1 / \epsilon ) \right)^{O\left(\sqrt{k} \ln (1 / \pmin)\right)} \; .\label{eq:Ndef}
\end{align}
Then \textsc{FourierMomentDescent} (Algorithm~\ref{alg:fourier_moment_descent}) requires sample complexity $\widetilde{O} (\sqrt{k} e^{\sqrt{k}} (N_1 + N))$ and runs in time $\widetilde{O} (\sqrt{k} e^{\sqrt{k}} (d N_1 + N))$.
\label{lem:runtime}
\end{lemma}
\noindent
We defer the proof of Lemma~\ref{lem:runtime} to Appendix~\ref{app:runtime}.

We can now complete the proof of Theorem~\ref{thm:fmd_main}.

\begin{proof}[Proof of Theorem~\ref{thm:fmd_main}]
	By Lemma~\ref{lem:fourier_moment_descent}, \textsc{FourierMomentDescent} outputs a vector $a_T\in\R^d$ for which $\norm{w_i - a_T}_2 \le \epsilon$ for some $i\in[k]$. The runtime and sample complexity bounds follow from Lemma~\ref{lem:runtime}.
\end{proof} 


\section{Learning All Components Under Zero Noise}
\label{sec:allcomps_nonoise}

In this short section we briefly describe how to use \textsc{FourierMomentDescent} in conjunction with existing techniques for boosting to learn \emph{all} components in a mixture of linear regressions. We remark that the arguments in this section are fairly standard.

We will make use of the following local convergence result of \cite{li2018learning}.

\begin{theorem}
	Let $\calD$ be a mixture of linear regressions in $\R^d$ with regressors $\{w_j\}$, minimum mixing weight $\pmin$, separation $\Delta$, and components whose covariances have eigenvalues all bounded within $[1,\sigma]$. Let $\zeta \triangleq \Delta\cdot \min\left(\frac{1}{2\sigma},\frac{\pmin}{64}\right)$. There is an algorithm \textsc{LL-Boost}$(\calD,v,\epsilon,\delta)$ which, given any $\epsilon > 0$ and $v\in\R^d$ for which there exists $j\in[k]$ with $\norm{w_j - v}_2 \le \zeta/\sigma$, draws $T\cdot M$ samples from $\calD$ for 
	\begin{equation*}
		T = O ( \pmin^{-2} d \ln(\zeta/\epsilon) ) \ \ \ \text{and} \ \ \ M = \poly (1/\Delta, 1 / \pmin ,\sigma,\log T )\cdot\ln(1/\delta),
	\end{equation*} runs in time $T\cdot M\cdot d$, and outputs $\tilde{v}\in\R^d$ for which $\norm{w_j - \tilde{v}}_2 \le \epsilon$ with probability at least $1 - \delta$. \label{thm:liliang}
\end{theorem}

We give a formal specification of our procedure \textsc{LearnWithoutNoise} for learning all components of a noise-less mixture of linear regressions in Algorithm~\ref{alg:learnwithoutnoise} below. The basic approach is to repeatedly invoke \textsc{FourierMomentDescent} to produce an estimate for one of the regressors of $\calD$ to within $O(\Delta\pmin)$ error, run \textsc{LL-Boost} to refine it to an estimate $v$ with error essentially as small as one would like (because of the exponential convergence rate of \textsc{LL-Boost}), and then filter out all samples $(x,y)$ for which the residual $|y - \langle x,v\rangle|$ is sufficiently small.

\begin{algorithm}\caption{\textsc{LearnWithoutNoise}($\calD,\delta,\epsilon$), Theorem~\ref{thm:learnwithoutnoise_main}}\label{alg:learnwithoutnoise}
\begin{algorithmic}[1]
	\State \textbf{Input}: Sample access to mixture of linear regressions $\calD$ with separation $\Delta$ and zero noise and regressors $\{w_i\}$, failure probability $\delta$, error $\epsilon$
	\State \textbf{Output}: List of vectors $\mathcal{L}\triangleq \{\tilde{w}_1,...,\tilde{w}_k\}$ for which there is a permutation $\pi:[k]\to[k]$ for which $\norm{\tilde{w}_i - w_{\pi(i)}}_2 \le \epsilon$ for all $i\in[k]$, with probability at least $1 - \delta$.
			\State Set $\delta' = \delta/2k$
			\State Set $\epsilon_{\text{FMD}} = \Delta\pmin/64$.
			\State Set $\epsilon_{\text{boost}} = \min\{\epsilon,\poly(\pmin,\Delta,1/k,1/d)^{\sqrt{k}\ln(1/\pmin)}\}$.
			\For{$i\in[k]$}
				\State Let $w'_i$ be the output of \textsc{FourierMomentDescent}($\calD, \delta',\epsilon_{\text{FMD}}$)
				\State Let $\tilde{w}_i$ be the output of \textsc{LL-Boost}($\calD, w'_i,\epsilon_{\text{boost}},\delta'$)
				\State Henceforth when sampling from $\calD$, ignore all samples $(x,y)$ for which $|y - \langle x,\tilde{w}_i\rangle| \le \epsilon_{\text{boost}}\cdot\poly(\log d)$.
			\EndFor
\end{algorithmic}
\end{algorithm}

\begin{theorem}
	\sloppy Given $\delta,\epsilon>0$ and a mixture of spherical linear regressions $\calD$ with separation $\Delta$ and zero noise, with probability at least $1 - \delta$, \textsc{LearnWithoutNoise}($\calD,\delta,\epsilon$) (Algorithm~\ref{alg:learnwithoutnoise}) returns a list of vectors $\mathcal{L}\triangleq \{\tilde{w}_1,...,\tilde{w}_k\}$ for which there is a permutation $\pi:[k]\to[k]$ for which $\norm{\tilde{w}_i - w_{\pi(i)}}_2 \le \epsilon$ for all $i\in[k]$. Furthermore, \textsc{LearnWithoutNoise} requires sample complexity 
	\begin{align*}
	N = \widetilde{O}\left(d\ln(1/\epsilon)\ln(1/\delta)\pmin^{-4}\Delta^{-2}\cdot \poly\left(k,\ln(1/\pmin),\ln(1/\Delta)\right)^{O(\sqrt{k}\ln(1/\pmin))}\right)
	\end{align*}
	and time complexity $Nd\cdot\poly\log(k,d,1/\Delta,1/\pmin,1/\epsilon)$.
\label{thm:learnwithoutnoise_main}
\end{theorem}

\begin{proof}
	By Theorem~\ref{thm:fmd_main}, every $w'_i$ in \textsc{LearnWithoutNoise} is $\frac{\Delta\pmin}{64}$-close to a regressor $w_{i'}$ of $\calD$, and by Theorem~\ref{thm:liliang}, \textsc{LL-Boost} improves this to a vector $\tilde{w}_i$ for which $\norm{\tilde{w}_i - w_{i'}}_2 \le\epsilon_{\text{boost}}$, where \begin{equation}\epsilon_{\text{boost}} \min\{\epsilon,\poly(\pmin,\Delta,1/k,1/d)^{\sqrt{k}\ln(1/\pmin)}\}.\end{equation} As a result, only a $\poly(\pmin,\Delta,1/k,1/d)^{\sqrt{k}\ln(1/\pmin)}$ fraction of subsequent samples will be removed, and the resulting error can be absorbed into the sampling error that goes into subsequent calls to \textsc{L2Estimate} and subsequent matrices $\widehat{M}^{(N)}_a$ that we run \textsc{ApproxBlockSVD} on, in the remainder of \textsc{LearnWithoutNoise}.
\end{proof}


\section{Learning All Components Under Noise}
\label{sec:allcomps}

In this section, we describe how to learn all components under the much more challenging setting where there is regression noise. We show that, at the extra cost of running in time exponential in $1/\Delta^2$ in addition to $\sqrt{k}$, there is an algorithm, which we call \textsc{LearnWithNoise}, that can learn mixtures of linear regressions to error $\epsilon$ when $\noise = O(\epsilon)$.

\begin{theorem}\label{thm:learnwithnoise_main}
	Given $\delta,\epsilon>0$ and a mixture of spherical linear regressions $\calD$ with regressors $\{w_1,...,w_k\}$, separation $\Delta$, and noise rate $\noise = O(\epsilon)$, with probability at least $1 - \delta$, \textsc{LearnWithNoise} ($\calD,\delta,\epsilon$) (Algorithm~\ref{alg:learnwithoutnoise}) returns a list of vectors $\mathcal{L}\triangleq \{\tilde{w}_1,...,\tilde{w}_k\}$ for which there is a permutation $\pi:[k]\to[k]$ for which $\norm{\tilde{w}_i - w_{\pi(i)}}_2 \le \epsilon$ for all $i\in[k]$. Furthermore, \textsc{LearnWithNoise} requires sample complexity 
	\begin{align*}
	N = \widetilde{O}\left(d\epsilon^{-2}\ln(1/\epsilon)\ln(1/\delta)\pmin^{-4}\Delta^{-2}\cdot\poly\left(k,1/\epsilon,\ln(1/\pmin)\right)^{O(\sqrt{k}\ln(1/\pmin)/\Delta^2)}\right)
	\end{align*}
	and time complexity $Nd\cdot \poly\log(k,d,1/\Delta,1/\pmin,1/\epsilon)$.
\end{theorem}

In Section~\ref{subsec:stay} we prove the key technical ingredient behind our proof of Theorem~\ref{thm:learnwithnoise_main}, Lemma~\ref{lem:stay_main}, which allows us to carefully control the dynamics of Fourier moment descent. In Section~\ref{subsec:randinit} we describe how to get an initialization which satisfies the hypotheses of Lemma~\ref{lem:stay_main}. In Section~\ref{subsec:learnwithnoise_spec} we give the full specification of \textsc{LearnWithNoise}. In Section~\ref{subsec:learnwithnoise_correct} we prove Theorem~\ref{thm:learnwithnoise_main}. Finally, in Section~\ref{subsec:boost_noisy}, we briefly describe how to leverage the local convergence result of \cite{kwon2019converges} in conjunction with our algorithm to get improved noise tolerance in the setting where the mixing weights are \emph{a priori} known.

\subsection{Staying on the Same Component}
\label{subsec:stay}

The main result of this section and the primary technical component behind Theorem~\ref{thm:learnwithnoise_main} is Lemma~\ref{lem:stay_main} below. This is a substantially more refined version of Lemma~\ref{lem:progress} in which we control not only the probability we make progress in the $t$-th step of moment descent, but also the probability that the the component $a_{t+1}$ is closest to is the same as the one $a_t$ is closest to.

We first introduce some preliminary notation and facts that we will use in the proof of Lemma~\ref{lem:stay_main}.

For $v\in\S^{d-1}$, define $\mathcal{F}$ and $\mathcal{F}'_v$ respectively to be the distribution of $y - \langle a_t, x\rangle$ and of $y - \langle a_t + \eta v, x\rangle$, where $(x,y)\sim\calD$.

Let $\sigma^2_t\triangleq \noise^2 + \min_{i\in[k]}\norm{w_i - a_t}_2$. Denote the minimizing index $i$ by $i^*$.

We record the following application of Lemma~\ref{app:findspan} and Fact~\ref{thm:power-method} which says we have access to $\Span(\{w_i - a_t\})$ up to $1/\poly(k)$ additive error.

\begin{lemma}\label{lem:Unoise}
	Let $\delsamp, \delta' > 0$ and $a_t\in\R^d$. If we draw $N_1 = \widetilde{\Omega}\left(\pmin^{-1}\delsamp^{-1}\cdot d\cdot \ln(k/\delta')\right)$ samples, form $\widehat{\vec{M}}^{(N_1)}_{a_t} \in \R^{d \times d}$ as defined in \eqref{eq:empiricalM}, and run \textsc{ApproxBlockSVD}$(\widehat{\vec{M}}^{(N_1)}_{a_t}, \delsamp, \delta')$ to produce a matrix $\vec{U}\in\R^{k\times d}$, then with probability $1 - \delta'$ we have that for any $a,b\in\S^{d-1}$ in the row span of $\vec{U}$, \begin{enumerate}
		\item $\langle \vec{U}a, \vec{U}b\rangle \le \langle a,b\rangle - \delsamp$
		\item $1 - \delsamp\le \norm{\vec{U}(w_i - a_t)}_2\le 1-\delsamp$.
	\end{enumerate}
\end{lemma}

\begin{proof}
	By Lemma~\ref{lem:findspan} and Fact~\ref{thm:power-method}, with probability $1 - \delta'$ we can ensure that \begin{equation}\norm{\vec{U}^{\top}\vec{\Lambda}\vec{U} - \E_{x,y}[\vec{M}^{x,y}_{a_t}]}_2 \le \delsamp\cdot\pmin/2,\end{equation} where $\vec{\Lambda}\in\R^k$ is some diagonal matrix of eigenvectors and $\vec{M}^{x,y}_{a_t}$ is defined in \eqref{eq:mxydef} and satisfies $\E_{x,y}[\vec{M}^{x,y}_{a_t}] = \sum^k_{i=1}p_i (w_i - a)(w_i - a)^{\top}$ by Lemma~\ref{lem:findspan}. Parts 1 and 2 of the lemma then follow by Lemma~\ref{lem:wedin_angle}.
\end{proof}

Lastly, the following elementary fact will be useful:

\begin{fact}
	If $x\in\R_{\ge 0}$ satisfies $\frac{1}{2}(x + x^{-1}) \ge 1 + \beta^2$ for some $0 < \beta \le 1$, then $1-\beta/2 \le x \le 1+\beta/2$.\label{fact:elem_ineq}
\end{fact}

\begin{proof}
	The solutions to $x + x^{-1} = 2 + 2\beta^2$ are 
	\begin{align*}x = 1 + \beta^2 \pm \beta\sqrt{2 + \beta^2}.
	\end{align*}
	One can check that $\beta^2 + \beta\sqrt{2 + \beta^2} \ge \beta$ for all $\beta\in\R$, while $-\beta^2 + \beta\sqrt{2 + \beta^2} \ge \beta/2$ for $\beta\in[0,1]$.
\end{proof}

We are now in a position to state and prove our main result of this section, Lemma~\ref{lem:stay_main}. This lemma roughly says that if we sample $M = \exp(\Omega(\sqrt{k}/\Delta^2))$ random steps $v_1,...,v_M$ at time $t$ of moment descent, then with high probability, if $j^*\in[M]$ is the first index on which \textsc{CompareMinVariances} outputs $\mathsf{true}$, then not only does walking in direction $v_{j^*}$ contract $\sigma_t$ by a factor of $1 - \Omega(\Delta/\sqrt{k})$ with high probability, but additionally, with at least $1/\poly(k)$ probability, it also keeps us closest to the component we were already closest to, in the following robust sense. Specifically, if we have a $(1 + \Omega(\Delta^2 \sqrt{k}))$ gap between $\norm{w_{i^*} - a_t}_2$ and all other $\norm{w_i - a_t}_2$, then with at least $1/\poly(k)$ probability, after one more iteration of moment descent, the $i^*$-th component will still be the closest to our new guess $a_{t+1} \in \R^d$, and this gap will persist.

\begin{lemma}
	There exist constants $a_{\mathrm{LR}}, \, a_{\mathrm{trials}}, \, a_{\mathrm{scale}}$, constants $ \overline{\beta}> \underline{\beta}$, a constant $0 \le a_{\mathrm{noise}} \le 1/5$, and a constant $\tau_{\mathrm{gap}} > 0$, such that for all $c < \tau_{\mathrm{gap}}$, the following holds for some $0<\kappa_1<\kappa_2\le 1$ satisfying $\kappa_2 - \kappa_1 = c\Delta^2k^{-1/2}$.

	Let $\delta > 0$. Suppose that $\noise^2 \le a_{\mathrm{noise}}\cdot \epsilon^2$. Suppose that 
	\begin{equation}
		\norm{w_i - a_t}_2 \le a_{\mathrm{scale}}\cdot k^{1/4}\label{eq:normbounded}
	\end{equation} for all $i\in[k]$. For $M\triangleq e^{a_{\mathrm{trials}}\sqrt{k}/\Delta^2}\ln(3/\delta)$ and $g_1,..,.g_M\sim\N(0,\Id_k)$, let $v_j = \frac{g_j\vec{U}}{\norm{g_j\vec{U}}_2}\in\S^{d-1}$ for $j\in[M]$. Let $\sigma^*$ be a number for which $0.9\sigma_t\le \sigma^*\le 1.1\sigma_t$, and let $\eta\triangleq a_{\mathrm{LR}}\cdot\Delta\cdot \sigma_*\cdot k^{-1/4}$.

	Then with probability at least $1 - \delta$ over the randomness of $g_1,...,g_M$ as well as over the behavior of all runs of \textsc{CompareMinVariances}, the following events hold: 
	\begin{enumerate}
		\item {\bf(Progress detected)} If $\min_{i \in [k]}\norm{w_i - a_t}^2_2 \ge \epsilon^2/2$, then 
		\begin{align*}
		\textsc{CompareMinVariances}(\calF,\calF'_v,a_{\mathrm{scale}}\cdot k^{-1/4},\sigma^*/1.1,\kappa_1,\kappa_2,\delta/3M)
		\end{align*}
		outputs $\mathsf{true}$ for at least one $j\in[M]$.

		Let $j^*$ be the smallest such $j$, and define 
		\begin{equation}
		a_{t+1} \triangleq a_t + \eta v_{j^*}.
		\end{equation}
		\item {\bf (Make at least some amount of progress)} If $\min_{ i \in [k] }\norm{w_i - a_t}^2_2 \ge \epsilon^2/2$, then
		\begin{equation}\sigma^2_{t+1} \le \left(1 - \underline{\beta}\Delta^2/\sqrt{k}\right)\sigma^2_t.\label{eq:someprogress}\end{equation}
		\item {\bf (Make at most some amount of progress)} Regardless of whether $\min_{ i \in [k] } \norm{w_i - a_{t}}^2_2 \le \epsilon^2/2$, 
		\begin{equation*}
		\sigma^2_{t+1} \ge \left(1 - \overline{\beta}\Delta^2/\sqrt{k}\right)\sigma^2_t.
		\end{equation*}
	\end{enumerate}

	If we assume that for all $i\neq i^*$, \begin{equation}\norm{w_i - a_t}_2 \ge\left(1 + c\Delta^2/\sqrt{k}\right)\cdot\norm{w_{i^*} - a_t}_2,\label{eq:startwithgap}\end{equation} then crucially, we have that with probability $1/\poly(k)$, the events above hold and \emph{additionally}: \begin{enumerate}\setcounter{enumi}{3}
		\item {\bf($i^*$ remains closest by same margin)} If $\min_{i\in[k]}\norm{w_i - a_t}^2_2 \ge \epsilon^2/2$, then for all $i\neq i^*$, \begin{equation}\norm{w_i - a_t - \eta v_{j^*}}_2 \ge \left(1 + c\Delta^2/\sqrt{k}\right)\cdot \norm{w_{i^*} - a_t - \eta v_{j^*}}_2.\end{equation}
	\end{enumerate}
	\label{lem:stay_main}
\end{lemma}

We emphasize that the main content of Lemma~\ref{lem:stay_main} is part 4.

\newcommand{\delvec}{\text{\textdelta}}
\newcommand{\delperp}{\delvec^{\perp}}

\begin{proof}
	Henceforth we will say that ``\textsc{CompareMinVariances} succeeds and outputs $\mathsf{true}$/$\mathsf{false}$ on direction $v$'' to mean that a single run of 
	\begin{align*}
	\textsc{CompareMinVariances}(\calF,\calF'_v,a_{\mathrm{scale}}\cdot k^{-1/4},\sigma^*/1.1,\kappa_1,\kappa_2,\delta/3M)
	\end{align*} 
	is successful (in the language of Corollary~\ref{cor:compare_min_variances}, this happens with probability $1 - \delta/3M$) and outputs $\mathsf{true}$/$\mathsf{false}$.

	Recall from \eqref{eq:onestep} that we have 
	\begin{equation}
		\norm{w_i - a_t - \eta v_j}^2_2 = \norm{w_i - a_t}^2_2 + \eta^2 - 2\eta\langle w_i - a_t, v_j\rangle.\label{eq:onestepagain}
	\end{equation} 
	Define $\delvec_i \triangleq w_i - a_t$ and $\hat{\delvec}_i \triangleq \frac{w_i - a_t}{\norm{w_i - a_t}_2}$. For every $i\neq i^*$, define $\delperp_i \triangleq \hat{\delvec}_i - \langle \hat{\delvec}_{i^*}, \hat{\delvec}_i\rangle\hat{\delvec}_{i^*}$. Finally, let $\gamma^{(j)}_i = \langle \hat{\delvec}_i, v_j\rangle$. Where the context is clear, we will omit the superscript $(j)$. 

	Let $\nu_A, \nu_B, \nu_C>0$ be absolute constants, and suppose $\nu_A < \nu_B$. For $i\in[k]$ and $j\in[M]$, define the following events: 
	\begin{enumerate}
		\item Let $A_j[i]$ be the event that $\gamma^{(j)}_{i^*}\ge \nu_A\Delta k^{-1/4}$.
		\item Let $B_j[i]$ be the event that $\gamma^{(j)}_i \le \nu_B\Delta k^{-1/4}$.
		\item Let $C_j$ be the event that $A_j[i^*]$ occurs and also $\langle v_j, \delperp_{i}\rangle \le \nu_C\Delta^2 k^{-1/2}\norm{\delperp_i}_2$ for all $i\neq i^*$.
	\end{enumerate}

	For $j\in[M]$, also let $B_j$ denote the event that $B_j[i]$ occurs for every $i\in[k]$.

	By our assumption on $\sigma_*$ and the definition of $\eta$, we know that $\eta = a'_{\mathrm{LR}}\cdot k^{-1/4}\cdot\norm{\delvec_i}_2$, where $a'_{\mathrm{LR}} \in[0.9,1.1]\cdot a_{\mathrm{LR}}$. It will be useful later in the proof to assume that $\nu_B < a'_{\mathrm{LR}} < 2\nu_A$.

	First, we compute the exact distance to $v_{i^*}$ after walking along $v_j$ and, provided the events $B_j[i]$ occur, lower bound the distances to all other components $v_i$.

	\begin{claim}
		Let $i\in[k], j\in[M]$, and suppose $B_j[i]$ occurs. Then 
		\begin{equation}
		\norm{\delvec_{i} - \eta v_j}^2_2 \ge \norm{\delvec_{i}}^2_2\cdot \left(1 + a'^2_{\mathrm{LR}}k^{-1/2} - 2a'_{\mathrm{LR}}k^{-1/4}\gamma^{(j)}_{i}\right), \label{eq:itoistar_pre}
		\end{equation} 
		with equality when $i = i^*$. Furthermore, when $i\neq i^*$ we get from \eqref{eq:startwithgap} that 
		\begin{equation}
		\norm{\delvec_i - \eta v_j}^2_2 \ge \norm{\delvec_{i^*}}^2_2\cdot \left(\left(1 + c\Delta^2 k^{-1/2}\right)^2 + a'^2_{\mathrm{LR}}\Delta^2k^{-1/2} - 2a'_{\mathrm{LR}}\Delta k^{-1/4}\gamma^{(j)}_i - 2ca'_{\mathrm{LR}}\Delta^3 k^{-3/4}\gamma^{(j)}_i\right).\label{eq:itoistar}
	\end{equation}\label{claim:itoistar}
	\end{claim}

	\begin{proof}
		We may rewrite \eqref{eq:onestepagain} as \begin{equation}
			\norm{\delvec_i - \eta v_j}^2_2 = \norm{\delvec_i}^2_2 + a'^2_{\mathrm{LR}}\Delta^2 k^{-1/2}\norm{\delvec_{i^*}}^2_2 - 2a'_{\mathrm{LR}}\Delta k^{-1/4}\norm{\delvec_{i^*}}_2\norm{\delvec_i}_2\cdot \gamma^{(j)}_i. \label{eq:onesteprewrite}
		\end{equation} The right-hand side of \eqref{eq:onesteprewrite}, as a function of $\norm{\delvec_{i^*}}_2$, is decreasing as long as 
		\begin{align*}
		\norm{\delvec_{i^*}}_2 \le a'^{-1}_{\mathrm{LR}}\Delta^{-1}k^{1/4}\norm{\delvec_i}_2 \gamma^{(j)}_i.
		\end{align*}
		But this condition holds because event $B_j[i]$ occurs, $\nu_B < a'_{\mathrm{LR}}$, and $\norm{\delvec_{i^*}}_2 \le \norm{\delvec_i}_2$. So \eqref{eq:itoistar_pre} follows, with equality when $i = i^*$.

		When $i\neq i^*$, we additionally know that $\norm{\delvec_i}_2 \ge \left(1 + c\Delta^2 k^{-1/2}\right)\cdot \norm{\delvec_{i^*}}_2$. So by the fact that the right-hand side of \eqref{eq:onesteprewrite} is decreasing as a function of $\norm{\delvec_{i^*}}_2$ for $\norm{\delvec_{i^*}}_2\in(-\infty,\norm{\delvec_i}_2]$, we get \eqref{eq:itoistar}.
	\end{proof}

	Using \eqref{eq:itoistar_pre} of Claim~\ref{claim:itoistar}, which is an equality when $i = i^*$, we can upper bound the distance to $v_{i^*}$ after walking along $v_j$, provided events $A_j[i^*]$ and $B_j[i]$ occur.

	\begin{claim}
		Let $j\in[M]$, and suppose $A_j[i^*]$ and $B_j[i^*]$ occur. Then there is an absolute constant $\underline{\beta}' > 0$ for which \begin{equation}\norm{\delvec_{i^*} - \eta v_j}^2_2 \le \norm{\delvec_{i^*}}^2_2 \cdot (1 - \underline{\beta}'\Delta^2 k^{-1/2}).\label{eq:Aevent}\end{equation}\label{claim:Aevent}
	\end{claim}

	\begin{proof}
		By \eqref{eq:itoistar_pre} which is an equality when $i = i^*$, \begin{equation}\norm{\delvec_{i^*} - \eta v_j}^2_2 \le \norm{\delvec_{i^*}}^2_2 \cdot \left(1 - (2a'_{\mathrm{LR}}\nu_A - a'^2_{\mathrm{LR}})\Delta^2 k^{-1/2}\right).\end{equation} The claim follows by taking $\underline{\beta}' \triangleq 2a'_{\mathrm{LR}}\nu_A - a'^2_{\mathrm{LR}}$, which is positive by the assumption that $a'_{\mathrm{LR}} < 2\nu_A$.
	\end{proof}

	Next, using \eqref{eq:itoistar} of Claim~\ref{claim:itoistar}, we argue that the only way to make progress towards a \emph{different} component $i\neq i^*$ by an amount comparable to that of Claim~\ref{claim:Aevent}, is if $A_j[i]$ has occurred. In particular, the following claim is the contrapositive of this.

	\begin{claim}
		Let $i\neq i^*$ and $j\in[M]$, and suppose $B_j[i]$ occurs and $A_j[i]$ \emph{does not occur}. Then \begin{equation}
			\norm{\delvec_i - \eta v_j}^2_2 \ge \norm{\delvec_{i^*}}^2_2 \cdot \left(1 - (\underline{\beta}'-c)\Delta^2 k^{-1/2} \right)\label{eq:Bevent}
		\end{equation}\label{claim:notAevent}
	\end{claim}

	\begin{proof}
		Because $A_j[i]$ does not occur, $\gamma^{(j)}_i < \nu_A k^{-1/4}$. So by \eqref{eq:itoistar}, \begin{align*}
			\norm{\delvec_i - \eta v_j}^2_2 &\ge \norm{\delvec_{i^*}}^2_2 \cdot \left(\left(1 + c\Delta^2 k^{-1/2}\right)^2 + a'^2_{\mathrm{LR}}\Delta^2 k^{-1/2} - 2\nu_Aa'_{\mathrm{LR}}\Delta^2 k^{-1/2} - 2\nu_A\cdot c\cdot a'_{\mathrm{LR}}\Delta^4 k^{-1}\right) \\
			&= \norm{\delvec_{i^*}}^2_2 \cdot \left(\left(1 + c\Delta^2 k^{-1/2}\right)^2 - \underline{\beta}'^2\Delta^2 k^{-1/2} - 2\nu_A\cdot c\cdot a'_{\mathrm{LR}}\Delta^4 k^{-1}\right) \\
			&= \norm{\delvec_{i^*}}^2_2 \cdot \left(1 - \underline{\beta}'\Delta^2 k^{-1/2} + 2c\Delta^2k^{-1/2}\cdot \left(1 - \nu_Aa'_{\mathrm{LR}}\Delta^2 k^{-1/2}\right) + c^2\Delta^4k^{-1}\right) \\
			&\ge \norm{\delvec_{i^*}}^2_2 \cdot \left(1 - \underline{\beta}'\Delta^2 k^{-1/2} + c\Delta^2k^{-1/2}\right),
		\end{align*} where in the last step we used the fact that $1 - \nu_Aa'_{\mathrm{LR}}\Delta^2 k^{-1/2}\ge 1/2$ for sufficiently large $k$. 
	\end{proof}

	Henceforth, let $\kappa_1 = \left(\underline{\beta}' - \frac{3c}{2}\right)\Delta^2k^{-1/2}$ and $\kappa_2 = \left(\underline{\beta}' - \frac{c}{2}\right)\Delta^2 k^{-1/2}$. In Lemma~\ref{lem:stay_main}, we will take $\underline{\beta}\triangleq \underline{\beta}' - \frac{3c}{2}$.

	Claims~\ref{claim:Aevent} and \ref{claim:notAevent} now imply the following about the behavior of \textsc{CompareMinVariances}. The upshot of the following two corollaries is that for any $j\in[M]$, if $B_j[i]$ occurs for every $i$ and \textsc{CompareMinVariances} succeeds and outputs $\mathsf{true}$ on direction $v_j$, the conditional probability of $A_j[i^*]$ happening is at least the conditional probability of $A_j[i]$ happening for any $i\neq i^*$.

	\begin{corollary}
		Let $j\in[M]$, and suppose $A_j[i^*]$ and $B_j[i^*]$ occur. Then \textsc{CompareMinVariances} succeeds and outputs $\mathsf{true}$ on direction $v_j$.\label{cor:Acompareminvariances}
	\end{corollary}

	\begin{proof}
		By adding $\noise^2$ to both sides of \eqref{eq:Aevent} in Claim~\ref{claim:Aevent}, we see that \begin{equation}
			\sigma_{t+1}^2 \le \norm{\delvec_{i^*} - \eta v_j}^2_2 + \noise^2 \le \norm{\delvec_{i^*}}^2_2 \cdot (1 - \underline{\beta}'\Delta^2 k^{-1/2}) + \noise^2 \le \left(1 - (\underline{\beta}'-c/2)\Delta^2 k^{-1/2}\right)(\norm{\delvec_{i^*}}^2_2 + \noise^2),
		\end{equation} where in the last step we used the assumptions that $\noise^2 \le a_{\mathrm{noise}}\cdot \epsilon^2$ and $\norm{w_{i^*} - a_t}^2_2\ge \epsilon^2/2$ for some sufficiently small constant $a_{\mathrm{noise}}$, which we just need to be at most $\frac{c}{4\underline{\beta}'}$ here.
	\end{proof}

	\begin{corollary}
		Let $i\neq i^*$ and $j\in[M]$, and suppose $B_j[i]$ holds and \textsc{CompareMinVariances} succeeds and outputs $\mathsf{true}$ on direction $v_j$. Then $A_j[i]$ has also occurred.\label{cor:notAcompareminvariances}
	\end{corollary}

	\begin{proof}
		By the contrapositive of Claim~\ref{claim:notAevent}, if \begin{equation}
			\norm{\delvec_i - \eta v_j}^2_2 < \norm{\delvec_{i^*}}^2_2 \cdot \left(1 - (\underline{\beta}' - c)\Delta^2 k^{-1/2}\right),\label{eq:opposite}
		\end{equation} and $B_j[i]$ occurs, then $A_j[i]$ occurs. We would like to show that \eqref{eq:opposite} then implies that \textsc{CompareMinVariances}, if it succeeds, outputs $\mathsf{true}$ on direction $v_j$. Adding $\noise^2$ to both sides of this, we conclude that \begin{equation}
			\sigma^2_{t+1} = \noise^2 + \norm{\delvec_i - \eta v_j}^2_2 < \noise^2 + \norm{\delvec_{i^*}}^2_2 \cdot \left(1 - (\underline{\beta}' - c)\Delta^2 k^{-1/2}\right) \le \left(1 - \left(\underline{\beta}' - \frac{3c}{2}\right)\Delta^2 k^{-1/2}\right),
		\end{equation} where in the last we used the assumptions that $\noise^2 \le a_{\mathrm{noise}}\cdot \epsilon^2$ and $\norm{w_{i^*} - a_t}^2_2\ge \epsilon^2/2$ for some sufficiently small constant $a_{\mathrm{noise}}$, which we just need to be at most $\frac{c}{4\underline{\beta}' - c}$ here.
	\end{proof}

	We also give an upper bound to the amount of progress that any $v_j$ could make in any direction $i$, provided $B_j[i]$ holds.

	\begin{claim}
		Let $i\in[k], j\in[M]$, and suppose $B_j[i]$ occurs. Then there is an absolute constant $\overline{\beta} > \underline{\beta}$ for which $\norm{\delvec_{i} - \eta v_j}^2_2 \ge \norm{\delvec_{i^*}}^2_2 \cdot (1 - \overline{\beta}\Delta^2 k^{-1/2})$.\label{claim:Bevent}
	\end{claim}

	\begin{proof}
		By \eqref{eq:itoistar_pre}, \begin{equation}
			\norm{\delvec_i - \eta v_j}^2_2 \ge \norm{\delvec_i}^2_2 \cdot \left(1 - (2a'_{\mathrm{LR}}\nu_B - a'^2_{\mathrm{LR}})\Delta^2 k^{-1/2}\right).
		\end{equation} The claim follows by taking $\overline{\beta} = 2a'_{\mathrm{LR}}\nu_B - a'^2_{\mathrm{LR}}$. Note that we have that $\overline{\beta} > \underline{\beta}' > \underline{\beta}$ because $\nu_A < \nu_B$.
	\end{proof}

	At this point, we could already use Corollary~\ref{cor:Acompareminvariances} and Claim~\ref{claim:Bevent}, together with straightforward bounds on the probabilities of the events $A_j[i^*]$ and $B_j[i]$ (see Claims~\ref{claim:Aevent_prob_bound} and \ref{claim:Bevent_prob_bound} below) to show that parts 1), 2), and 3) of the lemma hold with the claimed probability. Note that the proofs of these steps do not use \eqref{eq:startwithgap}, so in particular parts 1), 2), and 3) of the lemma hold with the claimed probability \emph{without assuming \eqref{eq:startwithgap}}.

	We next lay the foundation for showing part 4) of the lemma holds with at least $1/\poly(k)$ probability, assuming \eqref{eq:startwithgap}. Thus far we have not talked about the events $C_j$. It is at this point that we arrive at the main claim of the proof, namely that if event $C_j$ happens, then the gap of \eqref{eq:startwithgap} between the $i^*$-th component and all other components persists in the next step.

	\begin{claim}
		Let $j\in[M]$ and suppose $C_j$ and $B_j[i]$ occur for all $i\in[k]$. Then \begin{equation}
			\norm{\delvec_i - \eta v_j}_2 \ge \left(1 + c\Delta^2k^{-1/2}\right)\cdot \norm{\delvec_{i^*} - \eta v_j}_2
		\end{equation} for all $i\neq i^*$.\label{lem:Cevent}
	\end{claim}

	\begin{proof}
		Suppose we could show that \begin{equation}\gamma^{(j)}_i \le \gamma^{(j)}_{i^*}\label{eq:gammarelation}\end{equation} for all $i\neq i^*$. Then by \eqref{eq:itoistar_pre}, we would conclude that \begin{align*}
			\frac{\norm{\delvec_i - \eta v_j}^2_2}{\norm{\delvec_{i^*} - \eta v_j}^2_2} &\ge \frac{\norm{\delvec_i}^2}{\norm{\delvec_{i^*}}^2_2}\cdot \frac{1 + a'^2_{\mathrm{LR}}k^{-1/2} - 2a'_{\mathrm{LR}}k^{-1/4}\gamma^{(j)}_i}{1 + a'^2_{\mathrm{LR}}k^{-1/2} - 2a'_{\mathrm{LR}}k^{-1/4}\gamma^{(j)}_{i^*}} \\
			&\ge \left(1 + c\Delta^2k^{-1/2}\right)^2\cdot \frac{1 + a'^2_{\mathrm{LR}}k^{-1/2} - 2a'_{\mathrm{LR}}k^{-1/4}\gamma^{(j)}_i}{1 + a'^2_{\mathrm{LR}}k^{-1/2} - 2a'_{\mathrm{LR}}k^{-1/4}\gamma^{(j)}_{i^*}} \\
			&\ge \left(1 + c\Delta^2k^{-1/2}\right)^2
		\end{align*} as desired.

		We now describe the intuition for the remaining argument. \eqref{eq:gammarelation} is not hard to show when the unit vector $\hat{\delvec_i}$ is somewhat far from $\hat{\delvec_{i^*}}$, in which case it is reasonable to imagine a sizable cone of directions around $\delvec_{i^*}$ such that if $v_j$ lies in that cone, \eqref{eq:gammarelation} holds. On the other hand, suppose $\hat{\delvec_i}$ is close to $\delvec_{i^*}$. Then \eqref{eq:gammarelation} can actually be false. But because their non-normalized counterparts $\delvec_i$ and $\delvec_{i^*}$ are assumed to be $\Delta$-separated, $\delvec_i$ and $\delvec_{i^*}$ must therefore be nearly collinear, in which case there must exist a gap between $\norm{\delvec_i}_2$ and $\norm{\delvec_{i^*}}_2$ that's even bigger than the one assumed in \eqref{eq:startwithgap}, and furthermore walking in $v_j$ cannot reduce this gap to below that of \eqref{eq:startwithgap} in the next step.

		We now proceed with the formal details. First note that \begin{equation}\gamma^{(j)}_i = \left\langle \hat{\delvec}_i,\hat{\delvec}_{i^*}\right\rangle\cdot\gamma^{(j)}_{i^*} + \left\langle \delperp_i, v_j \right\rangle\label{eq:gammabreakdown}.\end{equation} Now if $\langle \hat{\delvec}_i, \hat{\delvec}_{i^*}\rangle \le 0$, then by event $C_j$, $\gamma^{(j)}_i \le \langle \delperp_i, v_j\rangle \le \nu_C\Delta^2 k^{-1/2}\cdot\norm{\delperp_i}_2 < \gamma^{(j)}_{i^*}$ and we'd be done. On the other hand, if $\langle \hat{\delvec}_i, \hat{\delvec}_{i^*}\rangle > 0$, then we get that \begin{equation}\gamma^{(j)}_i \le \left\langle \hat{\delvec}_i,\hat{\delvec}_{i^*}\right\rangle\cdot\gamma^{(j)}_{i^*} + \nu_C\Delta^2 k^{-1/2}\cdot\norm{\delperp_i}_2.\end{equation} In this case, to show the desired inequality \eqref{eq:gammarelation}, it would suffice to show that \begin{equation}
			\gamma^{(j)}_{i^*}\left(1 - \left\langle \hat{\delvec}_i,\hat{\delvec}_{i^*}\right\rangle\right) \ge \nu_C\Delta^2 k^{-1/2}\cdot\norm{\delperp_i}_2.
		\end{equation} In particular, because event $A_{j}$ holds so that $\gamma^{(j)}_{i^*}\ge \nu_A\Delta k^{-1/4}$, we just need to show that \begin{equation}
			\nu_A\left(1 - \left\langle \hat{\delvec}_i,\hat{\delvec}_{i^*}\right\rangle\right) \ge \nu_C\Delta k^{-1/4}\cdot\norm{\delperp_i}_2.\label{eq:suffice}
		\end{equation} After squaring both sides of \eqref{eq:suffice}, making the substitution $\norm{\delperp_i}^2_2 = 1 - \left\langle \hat{\delvec}_i,\hat{\delvec}_{i^*}\right\rangle^2$, and rearranging, \eqref{eq:suffice} becomes \begin{equation}
			\nu_A^2\left(1 - \left\langle \hat{\delvec}_i,\hat{\delvec}_{i^*}\right\rangle\right)^2 - \nu_C^2 \Delta^2 k^{-1/2}\cdot\left(1 - \left\langle \hat{\delvec}_i,\hat{\delvec}_{i^*}\right\rangle^2\right) \ge 0.\label{eq:quadratic}
		\end{equation} This is merely a univariate inequality for a quadratic polynomial in $\left\langle \hat{\delvec}_i,\hat{\delvec}_{i^*}\right\rangle$. Let $\nu_{CA} \triangleq \nu_C/\nu_A$. One can compute the smaller of the two zeros of the left-hand side of \eqref{eq:quadratic} and see that the inequality is satisfied provided that $\left\langle \hat{\delvec}_i,\hat{\delvec}_{i^*}\right\rangle$ is at most \begin{equation}
			1 - \frac{2\nu^2_{CA}\Delta^2k^{-1/2}}{1 + \nu^2_{CA}\Delta^2k^{-1/2}} \ge 1 - a_{\text{Del}}\cdot\Delta^2k^{-1/2}
		\end{equation} for absolute constant $a_{\text{Del}}\triangleq \frac{2\nu^2_{CA}}{1 + \nu^2_{CA}\Delta^2k^{-1/2}}$.

		It remains to consider the case where $\left\langle \hat{\delvec}_i,\hat{\delvec}_{i^*}\right\rangle \ge 1 - a_{\text{Del}}\cdot \Delta^2k^{-1/2}$. This is where we use the fact that $\norm{\delvec_i - \delvec_{i^*}}_2 = \norm{w_i - w_{i^*}}_2 \ge\Delta$ to argue that, even though \eqref{eq:quadratic} does not hold and we cannot obtain \eqref{eq:gammarelation}, $\norm{\delvec_{i^*}}_2$ is so much smaller than $\norm{\delvec_i}_2$ that, conditioned on the event $B_j[i]$ for all $j\in[M]$, $\norm{\delvec_i - \eta v_j}_2$ is far larger than $\norm{\delvec_{i^*} - \eta v_j}_2$ for any $j$.

		First note that \begin{equation}
			\Delta^2 \le \norm{\delvec_i - \delvec_{i^*}}^2_2 = \norm{\delvec_i}^2_2 + \norm{\delvec_{i^*}}^2_2 - 2\left\langle \hat{\delvec}_i,\hat{\delvec}_{i^*}\right\rangle \norm{\delvec_i}_2 \norm{\delvec_{i^*}}_2,
		\end{equation} so 
		\begin{align*}
			1 - a_{\text{Del}}\frac{\Delta^2}{\sqrt{k}} 
			\le & ~ \left\langle \hat{\delvec}_i,\hat{\delvec}_{i^*}\right\rangle \\
			\le & ~ \frac{1}{2}\left(\frac{\norm{\delvec_i}_2}{\norm{\delvec_{i^*}}_2} + \frac{\norm{\delvec_{i^*}}_2}{\norm{\delvec_i}_2}\right) - \frac{\Delta^2}{2\norm{\delvec_i}_2\norm{\delvec_{i^*}}_2}\\
			\le & ~ \frac{1}{2}\left(\frac{\norm{\delvec_i}_2}{\norm{\delvec_{i^*}}_2} + \frac{\norm{\delvec_{i^*}}_2}{\norm{\delvec_i}_2}\right) - \frac{\Delta^2}{a_{\mathrm{scale}}\sqrt{k}},
		\end{align*} 
		where the second step follows by the original assumption in \eqref{eq:normbounded} that $\norm{\delvec_i}_2,\norm{\delvec_{i^*}}_2 \le a_{\mathrm{scale}}\Delta^2/\sqrt{k}$ for some absolute constant $a_{\mathrm{scale}}>0$. Recalling the relation between $a_{\text{Del}}$ and $\nu^2_{CA}$, we conclude that if we pick $\nu^2_{CA} < 1/a_{\mathrm{scale}}$, then we get that \begin{equation}
			\frac{1}{2}\left(\frac{\norm{\delvec_i}_2}{\norm{\delvec_{i^*}}_2} + \frac{\norm{\delvec_{i^*}}_2}{\norm{\delvec_i}_2}\right) \ge 1 + \alpha\frac{\Delta^2}{\sqrt{k}}
		\end{equation} for absolute constant $\alpha \triangleq \frac{1}{a_{\mathrm{scale}}} - a_{\text{Del}} > 0$, from which we conclude, by taking $\beta = \alpha^{1/2}\Delta k^{-1/4}$ in Fact~\ref{fact:elem_ineq}, that \begin{equation}\norm{\delvec_i}_2 \ge \left(1 + \alpha'\Delta k^{-1/4}\right)\norm{\delvec_{i^*}}_2\label{eq:compete}\end{equation} for $\alpha' = \alpha^{1/2}/2$ which is increasing in $\alpha$ and therefore in $1/a_{\mathrm{scale}}$. But as we showed in Claim~\ref{claim:Bevent}, \begin{equation}\norm{\delvec_i - \eta v_j}_2 \ge \left(1 - \overline{\beta}\Delta^2 k^{-1/2}\right)\cdot\norm{\delvec_i}_2\label{eq:Brepeat}\end{equation} for all $i\in[k], j\in[M]$ if $B_j[i]$ holds. So by taking $a_{\mathrm{scale}}$ sufficiently small relative to $\overline{\beta}$, we get from \eqref{eq:compete} and \eqref{eq:Brepeat} that \begin{equation}
			\norm{\delvec_i - \eta v_j}_2 \ge \left(1 - \overline{\beta}\Delta^2 k^{-1/2}\right)\cdot\left(1 + \alpha'\Delta k^{-1/4}\right)\norm{\delvec_{i^*}}_2 \gg \left(1 + c\Delta^2k^{-1/2}\right)\norm{\delvec_{i^*}}.
		\end{equation} But because event $C_j$ involves $A_j[i^*]$ happening, we certainly have that $\norm{\delvec_{i^*}}_2 \le \norm{\delvec_{i^*} - \eta v_j}^2_2$, so we are done.
	\end{proof}

	We now proceed to bound the probabilities of the events $A_j[i], B_j[i], C_j$. There are some minor technical complications from the fact that we don't have exact access to $\Span(\{w_i - a_t\})$ which we address now.

	Define $\alpha^{(i)}_{\svd} \triangleq \frac{\norm{\vec{U}\delvec_i}_2}{\norm{\delvec_i}_2}$. By the second part of Lemma~\ref{lem:Unoise}, $1 - \delsamp\le \alpha^{(i)}_{\svd} \le 1$. First note that for any $i\in[k]$, \begin{equation}
		\gamma^{(j)}_i = \langle \hat{\delvec}_i,v_j\rangle = \frac{\langle g, \vec{U}\delvec_i\rangle}{\norm{g\vec{U}}_2\cdot \norm{\delvec_i}_2} = \left\langle \frac{g}{\norm{g}_2}, \frac{\vec{U}\delvec_i}{\norm{\delvec_i}_2}\right\rangle = \alpha^{(i)}_{\svd}\left\langle \frac{g}{\norm{g}_2}, \frac{\vec{U}\delvec_i}{\norm{\vec{U}\delvec_i}_2}\right\rangle,\label{eq:gammaX}
	\end{equation} where $g\sim\N(0,\Id_k)$ and the last step follows by the second part of Lemma~\ref{lem:Unoise}. The random variable $\left\langle\frac{g}{\norm{g}_2}, \frac{\vec{U}\delvec_i}{\norm{\vec{U}\delvec_i}_2}\right\rangle$ is merely the correlation of a random unit vector with a fixed unit vector; call this random variable $X$ (clearly it does not depend on the fixed vector).

	We can now lower bound the probabilities of $A_j[i^*]$ and $B_j[i]$.

	\begin{claim}
		For any $j\in[M]$, $\Pr[A_j[i^*]] \ge e^{-a_{\mathrm{trials}}\sqrt{k}/\Delta^2}$ for some absolute constant $a_{\mathrm{trials}} > 0$.\label{claim:Aevent_prob_bound}
	\end{claim}

	\begin{proof}
		By \eqref{eq:gammaX}, $\Pr[A_j[i^*]] \ge \Pr[X\ge \nu_A\Delta k^{-1/4}]$. By Corollary~\ref{cor:unitcorr}, $\Pr[X\ge \nu_Ak^{-1/4}] \ge e^{-a_{\mathrm{trials}}\sqrt{k}/\Delta^2}$ for some $a_{\mathrm{trials}} > 0$.
	\end{proof}

	\begin{claim}
		For any $i\in[k]$ and $j\in[M]$, $\Pr[B_j[i]]\ge 1 - e^{-\overline{a_{\mathrm{trials}}}\sqrt{k}/\Delta^2}$ for some absolute constant $\overline{a_{\mathrm{trials}}} > a_{\mathrm{trials}}$.\label{claim:Bevent_prob_bound}
	\end{claim}

	\begin{proof}
		By \eqref{eq:gammaX}, $\Pr[B_j[i^*]] \ge \Pr[X \le \nu_B\Delta k^{-1/4}]$. If we take $\delsamp = 1/\poly(k)$ sufficiently small, then because $\nu_B > \nu_A$, we conclude by Corollary~\ref{cor:unitcorr} that $\Pr[X\le \nu_B\Delta k^{-1/4}] \ge 1 - e^{-\overline{a_{\mathrm{trials}}}\sqrt{k}/\Delta^2}$ for some $\overline{a_{\mathrm{trials}}} > a_{\mathrm{trials}}$.
	\end{proof}

	We next lower bound the probability of event $C_j$ relative to that of $A_j[i^*]$. Equivalently, provided $A_j[i^*]$ happens, we lower bound the conditional probability that the gap of \eqref{eq:startwithgap} is preserved. In this proof, we would like to use the fact that $\delperp_i$ is orthogonal to $\hat{\delvec}_{i^*}$ for all $i\neq i^*$ to argue that $\langle g\vec{U}, \delperp_i\rangle$ and $\langle g\vec{U}, \delvec_{i^*}\rangle$ are independent. Again, this is only true if $\vec{U}$ is exactly the projector to the span of $\{w_i - a_t\}$, and we need to argue that it suffices to take $\vec{U}$ an approximation to that projector.

	\newcommand{\perperror}{1/k^{100}}
	\begin{claim}
		For any $j\in[M]$, $\Pr[C_j] \ge \frac{1}{\poly(k)}\Pr[A_j[i^*]]$.\label{claim:probCj}
	\end{claim}

	\begin{proof}
		Let $v = g/\norm{g}_2\in\R^k$. Analogous to \eqref{eq:gammaX}, we can define $\beta^{(i)}_{\svd}\triangleq \frac{\norm{\vec{U}\delperp_i}_2}{\norm{\delperp_i}_2} \in[1 - \delsamp,1]$ and write \begin{equation}
			\frac{1}{\norm{\delperp_i}_2}\langle \delperp_i, v_j\rangle = \frac{\langle g, \vec{U}\delperp_i\rangle}{\norm{g\vec{U}}_2\cdot \norm{\delperp_i}_2} = \left\langle \frac{g}{\norm{g}_2}, \frac{\vec{U}\delperp_i}{\norm{\delperp_i}_2}\right\rangle = \beta^{(i)}_{\svd}\left\langle v, \frac{\vec{U}\delperp_i}{\norm{\vec{U}\delperp_i}_2}\right\rangle,\label{eq:delperpX}
		\end{equation} Define $\rho \triangleq \left\langle \vec{U}\delperp_i, \frac{\vec{U}\delvec_{i^*}}{\norm{\vec{U}\delvec_{i^*}}_2}\right\rangle$. By the first part of Lemma~\ref{lem:Unoise}, if we take $\delsamp < \perperror$, then \begin{equation}
			|\langle \vec{U}\delperp_i,\vec{U}\delvec_{i^*}\rangle| \le \perperror\norm{\delvec_{i^*}}_2\norm{\delperp_i}_2,
		\end{equation} so we conclude that \begin{equation}
			|\rho| \le \frac{\perperror\norm{\delvec_{i^*}}_2\norm{\delperp_i}_2}{\norm{\vec{U}\delvec_{i^*}}_2} \le 2(\perperror)\norm{\delperp_i}_2.\label{eq:rhobound}
		\end{equation} So we may write \begin{equation}
			\vec{U}\delperp_i = \rho\cdot \frac{\vec{U}\delvec_{i^*}}{\norm{\vec{U}\delvec_{i^*}}_2} + v'
		\end{equation} for $v'\in\R^k$ lying in the row span of $\vec{U}$ and orthogonal to $\vec{U}\delvec_{i^*}$.

		By Corollary~\ref{cor:joint}, for any absolute constant $a_{\text{perp}}>0$, with probability at least \begin{equation}\frac{1}{\poly(k)}\Pr\left[\langle v,\vec{U}\hat{\delvec}_{i^*}\rangle \ge \frac{\nu_A\Delta k^{-1/4}}{(1 - \delsamp)}\right] \ge \frac{1}{\poly(k)}\Pr\left[A_j[i^*]\right]\label{eq:relatetoAjistar},\end{equation} we have that \begin{equation}\left\langle v,\frac{\vec{U}\hat{\delvec}_{i^*}}{\norm{\vec{U}\hat{\delvec}_{i^*}}_2}\right\rangle \ge \frac{\nu_A\Delta k^{-1/4}}{(1 - \delsamp)} \ \ \ \text{and} \ \ \ \langle v,v'\rangle \le a_{\text{perp}}\Delta^2 k^{-1/2}\norm{v'}_2.\label{eq:jointevent}\end{equation} In particular, if this happens, then by \eqref{eq:gammaX} we get that \begin{equation}
			\langle v,\hat{\delvec}_{i^*}\rangle \ge \alpha^{(i)}_{\svd}\left\langle v,\vec{U}\hat{\delvec}_{i^*}\right\rangle \ge \nu_A\Delta k^{-1/4}
		\end{equation} and by \eqref{eq:delperpX}, \begin{align*}
			\frac{1}{\norm{\delperp_i}_2}\langle \delperp_i, v_j\rangle &= \beta^{(i)}_{\svd}\left\langle v, \frac{\vec{U}\delperp_i}{\norm{\vec{U}\delperp_i}_2}\right\rangle \\
			&\le \frac{\beta^{(i)}_{\svd}}{\norm{\vec{U}\delperp_i}_2}(\rho + \langle v,v'\rangle) \\
			&\le \frac{\beta^{(i)}_{\svd}}{\norm{\vec{U}\delperp_i}_2} \left(2(\perperror)\norm{\delperp_i}_2 + a_{\text{perp}}\Delta k^{-1/4}\norm{\delperp_i}_2\right)\\
			&= 2(\perperror) + a_{\text{perp}}\Delta k^{-1/4} < 2a_{\text{perp}}\Delta k^{-1/4},
		\end{align*} where in the third step we used \eqref{eq:rhobound}, the bound on $\langle v,v'\rangle$ in the event that \eqref{eq:jointevent} holds, and the fact that $\norm{v'}_2 \le \norm{\vec{U}\delperp_i}_2 \le \norm{\delperp_i}_2$. So by taking $a_{\text{perp}}$ sufficiently small, we conclude that event $C_j$ occurs in the event that \eqref{eq:jointevent} holds, which is with probability at least $\frac{1}{\poly(k)}\Pr[A_j[i^*]]$.
	\end{proof}

	Next, we would like to show that for any $j\in[M]$, if we condition on the events $B_j[i]$ holding for all $i\in[k]$, then the conditional probability of $A_j[i]$ is not much more than that of $A_j[i^*]$. Note that by rotational invariance, these conditional probabilities would be identical if $\vec{U}$ were exactly the projector to the span of $\{w_i - a_t\}$, and here it is straightforward to see that it suffices to take $\vec{U}$ a sufficiently good approximation to that projector.

	\begin{claim}
		For any $j\in[M]$ and $i\neq i^*$, if $\delsamp = 1/\poly(k)$ is sufficiently small, $\Pr[A_j[i]\wedge B_j] \le k\cdot \Pr[A_j[i^*]\wedge B_j]$.\label{claim:probratios}
	\end{claim}

	\begin{proof}
		 By \eqref{eq:gammaX}, we conclude that \begin{equation}
		 	\frac{\Pr[A_j[i]\wedge B_j]}{\Pr[A_j[i^*]\wedge B_j]} \le \frac{\Pr[\nu_A\Delta k^{-1/4}/\alpha^{(i)}_{\svd}\le X\le \nu_B\Delta k^{-1/4}/\alpha^{(i)}_{\svd}]}{\Pr[\nu_A\Delta k^{-1/4}/\alpha^{(i^*)}_{\svd} \le X \le \nu_B\Delta k^{-1/4}/\alpha^{(i^*)}_{\svd}]}.
		\end{equation} which can be upper bounded by $k$ by taking a sufficiently small $\delsamp = 1/\poly(k)$.
	\end{proof}

	We can now put all of these probability bounds together to show that if \textsc{CompareMinVariances} succeeds and outputs true on some direction $v$, the conditional probability that the gap of \eqref{eq:startwithgap} has been preserved is at least $1/\poly(k)$.

	\begin{claim}
		Let $B_v$ be the event that $\Pr[\langle \hat{\delvec}_i, v\rangle] \le \nu_B\Delta k^{-1/4}$ for all $i\in[k]$. Let $\mathsf{detect-progress}_v$ be the event that $B_v$ occurs and additionally that \textsc{CompareMinVariances} succeeds and outputs $\mathsf{true}$ on direction $v$. Let $\mathsf{gap-preserved}_v$ be the event that $B_v$ occurs and additionally $\norm{\delvec_{i} - \eta v}_2 \ge \left(1 + c\Delta^2k^{-1/2}\right)\cdot\norm{\delvec_{i^*} - \eta v}_2$ for all $i\neq i^*$. Then \begin{equation}\Pr_{v}\left[\mathsf{gap-preserved}_v \ | \ \mathsf{detect-progress}_v\right]\ge \frac{1}{\poly(k)}\end{equation}
	\end{claim}

	\begin{proof}
		For every $i\in[k]$, let $S_i$ denote the set of all $v$ for which $B_v$ occurs, \textsc{CompareMinVariances} succeeds and outputs $\mathsf{true}$ on direction $v$, and $\norm{\delvec_{i'} - \eta v}_2 \ge \left(1 + c\Delta^2k^{-1/2}\right)\cdot\norm{\delvec_{i} - \eta v}_2$ for all $i'\neq i$.

		To show the claim, it suffices to lower bound the quantity \begin{equation}
			\frac{\Pr_v[S_{i^*}]}{\Pr_v\left[\bigcup^k_{i=1}S_i\right]} \ge \frac{\Pr_v[S_{i^*}]}{\sum^k_{i=1}\Pr_v[S_i]},
		\end{equation} where the probabilities are over $v$ distributed as $\frac{g\vec{U}}{\norm{g\vec{U}}_2}$ for $g\sim\N(0,\Id_k)$, and where the inequality follows by a union bound. Fix any $j\in[M]$. By Corollary~\ref{cor:notAcompareminvariances}, if $v\in S_i$, then it is part of event $(A_j[i]\wedge B_j)$. By Corollary~\ref{cor:Acompareminvariances} and Lemma~\ref{lem:Cevent}, if $v$ is part of event $(C_j\wedge B_j)$, then $v\in S_{i^*}$. We conclude that \begin{equation}
			\frac{\Pr_v[S_{i^*}]}{\Pr_v\left[\bigcup^k_{i=1}S_i\right]} \ge \frac{\Pr[C_j\wedge B_j]}{\sum^k_{i=1}\Pr[A_j[i]\wedge B_j]} \ge \frac{\frac{1}{2}\Pr[A_j[i^*]\wedge B_j]}{\sum^k_{i=1}\Pr[A_j[i]\wedge B_j]} \ge \frac{1/\poly(k)}{1 + k(k-1)} = \frac{1}{\poly(k)},
		\end{equation} where the second step follows from Claim~\ref{claim:probCj} and the third step follows from Claim~\ref{claim:probratios}.
	\end{proof}

	We are now ready to finish the proof of Lemma~\ref{lem:stay_main}. First, condition on the event that all $M$ runs of \textsc{CompareMinVariances} are successful, which happens with probability at least $1 - \delta/3$. The probability that $B_j[i]$ holds for all $i\in[k], j\in[M]$ is at least $1 - kM e^{-\overline{a_{\mathrm{trials}}}\sqrt{k}/\Delta^2}$, by Claim~\ref{claim:Bevent_prob_bound}. Condition on this happening. The probability that $\mathsf{detect-progress}$$_{v_j}$ occurs for some $j\in[M]$ is at least the probability that $A_j[i^*]$ occurs for some $j\in[M]$, and this is at least $1 - \left(1 - e^{-a_{\mathrm{trials}}\sqrt{k}/\Delta^2}\right)^M \ge 1 - e^{-Me^{-a_{\mathrm{trials}}\sqrt{k}/\Delta^2}}$ by Claim~\ref{claim:Aevent_prob_bound}. So by taking $M = e^{-a_{\mathrm{trials}}\sqrt{k}/\Delta^2}\ln(3/\delta)$, by a union bound we conclude that with probability at least $1 - \delta$, every run of \textsc{CompareMinVariances} succeeds, $B_j[i]$ holds for every $i\in[k],j\in[M]$, and furthermore there is some $j$ for which $\mathsf{detect-progress}$$_{v_j}$ occurs.

	If $\mathsf{detect-progress}$$_{v_j}$ occurs for some $j$, then for that particular $j$, $\mathsf{gap-preserved}$$_{v_j}$ holds with probability at least $1/\poly(k)$.
\end{proof}

\subsection{Initializing With a Gap}
\label{subsec:randinit}

A key assumption in Lemma~\ref{lem:stay_main} is that there is a gap between $\norm{w_{i^*} - a_t}_2$ and all other $\norm{w_{i} - a_t}_2$. We next show that this assumption can be made to hold when $t = 0$. The high-level structure of the proof will be very similar to that of Lemma~\ref{lem:Cevent}.

\begin{lemma}
	There is a constant $\tau'_{\mathrm{gap}} > 0$ such that for all $c' < \tau'_{\mathrm{gap}}$, the following holds for any sufficiently small $\upsilon_* = \poly(\Delta,1/k)$.

	Fix any $i^*\in[k]$ and suppose that $\norm{w_{i^*}}_2\ge \underline{\sigma}$ for some $\underline{\sigma} > 0$. Let \begin{equation}\mathcal{S} \triangleq \{\underline{\sigma}\cdot k^{1/4},\underline{\sigma}\cdot(1 + \upsilon_*)\cdot k^{1/4},\underline{\sigma}\cdot(1+\upsilon_*)^2\cdot k^{1/4},...,k^{1/4}\}.\label{eq:mesh}\end{equation}
	Then any $i\neq i^*$ and $v\in\R^{d}$, let $\calE_v[i]$ denote the event that \begin{equation}
		\norm{w_{i} - v}^2_2 \ge \left(1 + c'\cdot \frac{\Delta^2}{\sqrt{k}}\right)\cdot \norm{w_{i^*} - v}^2_2\label{eq:gap_event}
	\end{equation} and define $\calE_v$ to be the event that $\calE_v[i]$ occurs simultaneously for all $i\neq i^*$. There exists $\alpha\in \mathcal{S}$ for which \begin{equation}\Pr_{\norm{v}_2 = \alpha}\left[\calE_v\right] \ge \exp(-O(\sqrt{k}/\Delta^2)),\end{equation} where the probability is over $v$ a Haar-random vector in $\R^{d}$ with norm $\alpha$.\label{lem:rand_init}
\end{lemma}

\newcommand{\what}{\hat{w}}
\newcommand{\wperp}{w^{\perp}}
\begin{proof}
	By design, there must exist an $\alpha\in \mathcal{S}$ for which $\alpha = (1 + \upsilon)\cdot\norm{w_{i^*}}_2\cdot k^{1/4}$ for $\upsilon\in[-\upsilon_*,\upsilon_*]$. Let $v$ be a random vector with norm $\alpha$.

	Define $\what_i = w_i/\norm{w_i}_2$. For every $i\neq i^*$, define $\wperp_i \triangleq \what_i - \langle \what_{i^*},\what_i\rangle\what_{i^*}$. Finally, repurposing notation from the proof of Lemma~\ref{lem:stay_main}, let $\gamma_i = \langle \what_i, v/\norm{v}_2\rangle$. Also, let $\rho_i \triangleq \norm{w_i}_2 /\norm{w_{i^*}}_2$. Under this notation, we see that \begin{align}
		\frac{\norm{w_i - v}^2_2}{\norm{w_{i^*} - v}^2_2} &= \frac{\norm{w_i}^2_2 + \alpha^2 - 2\alpha\gamma_i \norm{w_i}_2}{\norm{w_{i^*}}^2_2 + \alpha^2 - 2\alpha\gamma_{i^*}\norm{w_{i^*}}_2} \nonumber\\
		&= \frac{\rho_i^2 + (1+\upsilon)^2\sqrt{k} - 2(1+\upsilon)\rho_i k^{1/4}\gamma_i}{1 + (1+\upsilon)^2\sqrt{k} - 2(1+\upsilon)k^{1/4}\gamma_{i^*}}.\label{eq:initialize_ratiobound_basic}
	\end{align}

	Using similar terminology as in the proof of Lemma~\ref{lem:stay_main}, define the following two types of events over the random vector $v$: \begin{enumerate}
		\item Let $B[i]$ be the event that $\gamma_i \le k^{-1/4}\cdot (1 + c\Delta^2)$ for some absolute constant $c>0$.
		\item Let $C$ be the event that $\gamma_{i^*}\ge k^{-1/4}$ and $\langle v, \wperp_i\rangle \le \nu_{\text{perp}} k^{-1/2}\norm{\wperp_i}_2$ for all $i\neq i^*$, for some absolute constant $\nu_{\text{perp}}$.
	\end{enumerate}

	The main step will be to show that these events imply $\calE_v$.

	\begin{claim}
		Let $i\neq i^*$. If events $B[i]$ and $C$ occur, then $\calE_v[i]$ occurs. \label{claim:BCimplyE}
	\end{claim}

	\begin{proof}
		Henceforth, condition on $B[i]$, $B[i^*]$, and $C$ occurring.

		There are two cases to consider: either $\norm{w_{i^*}}_2$ and $\norm{w_i}_2$ are quite different, or they are relatively similar. 

		It turns out that our choice $\alpha = (1 + \upsilon)\cdot\norm{w_{i^*}}_2 \cdot k^{1/4}$ will allow us to handle the former case quite easily. Indeed, we first show that if $\norm{w_i}_2$ is not $(1\pm O(\Delta))$-close to $\norm{w_{i^*}}_2$, then $\calE_v[i]$ occurs.

		\begin{claim}
			Define the interval \begin{equation}
				\mathcal{I}\triangleq \left[1-2c^{1/2}\Delta, 1+2c^{1/2}\Delta\right].
			\end{equation} Then \begin{equation}\frac{\norm{w_i - v}^2_2}{\norm{w_{i^*} - v}^2_2} \ge 1 + \frac{c\Delta^2}{\sqrt{k}}.\end{equation}\label{claim:BCIimplyE}
		\end{claim} 

		\begin{proof}
			From \eqref{eq:initialize_ratiobound_basic} and events $B[i]$ and $C$ we get \begin{equation}
				\frac{\norm{w_i - v}^2_2}{\norm{w_{i^*} - v}^2_2} \ge \frac{\rho^2_i + (1 + \upsilon)^2\sqrt{k} - 2(1 + \upsilon)(1 + c\Delta^2)\rho_i}{1 + (1 + \upsilon)^2\sqrt{k} - 2(1 + \upsilon)}.\label{eq:initialize_ratiobound}
			\end{equation}

			If we define $\rho'_i = \rho_i - 1$, we can rewrite \eqref{eq:initialize_ratiobound} as 
			\begin{align*}
				& ~ 1 + \frac{2\rho'_i + \rho'^2_i - 2c\Delta^2(1+\upsilon) - 2(1+\upsilon)(1+c\Delta^2)\rho'_i}{1 + (1 + \upsilon)^2\sqrt{k} - 2(1 + \upsilon)} \\
				= & ~ 1 + \frac{\rho'^2_i - 2c\Delta^2(1+\upsilon) - 2(\upsilon + c\Delta^2 + c\upsilon\Delta^2)\rho'_i}{1 + (1 + \upsilon)^2\sqrt{k} - 2(1 + \upsilon)}.
			\end{align*} 
			For $\upsilon = \poly(\Delta,1/k)$ sufficiently small, note that if $\rho'^2_i \ge 4c\Delta^2$, then the numerator is at least $c\Delta^2$. And for such an $\upsilon$, \begin{equation}
				1 + (1 + \upsilon)^2\sqrt{k} - 2(1 + \upsilon) \le \sqrt{k},\label{eq:denom_bound}
			\end{equation} completing the proof of the claim.
		\end{proof} 

		Next we show that if $\norm{w_i}_2$ \emph{is} $(1\pm O(\Delta))$-close to $\norm{w_{i^*}}_2$ and furthermore $v$ is significantly more correlated with $\what_{i^*}$ than with any other $\what_i$, then $\calE_v[i]$ occurs.

		\begin{claim}
			Let $i\neq i^*$. If $\rho_i \in\mathcal{I}$ and \begin{equation}\gamma_i \le \gamma_{i^*}(1 - \omega)\label{eq:gammamugap}\end{equation} for $\omega\triangleq 2c^2 \Delta^4 + c\Delta^2$, then if events $B[i], B[i^*], C$ all occur, then \begin{equation}\frac{\norm{w_i - v}^2_2}{\norm{w_{i^*} - v}^2_2} \ge 1 + \frac{c\Delta^2}{\sqrt{k}}.\end{equation}\label{claim:BCangleimplyE}
		\end{claim}

		\begin{proof}
			From \eqref{eq:initialize_ratiobound_basic} and \eqref{eq:gammamugap} we get that \begin{align}
				\frac{\norm{w_i - v}^2_2}{\norm{w_{i^*} - v}^2_2} &\ge \frac{\rho^2_i + (1 + \upsilon)^2\sqrt{k} - 2(1+\upsilon)\rho_i k^{1/4}\gamma_{i^*}(1 - \omega)}{1 + (1 + \upsilon)^2\sqrt{k} - 2(1+\upsilon)k^{1/4}\gamma_{i^*}} \nonumber\\
				&= \frac{\rho^2_i + (1 + \upsilon)^2\sqrt{k} - 2(1+\upsilon)\rho_i k^{1/4}\gamma_{i^*}}{1 + (1 + \upsilon)^2\sqrt{k} - 2(1+\upsilon)k^{1/4}\gamma_{i^*}} + \frac{2\omega(1 + \upsilon)\rho_i k^{1/4}\gamma_{i^*}}{1 + (1 + \upsilon)^2\sqrt{k} - 2(1+\upsilon)k^{1/4}\gamma_{i^*}} \label{eq:initialize_ratiobound2}.
			\end{align} Next, note that the quantity $\rho^2_i + (1 + \upsilon)^2\sqrt{k} - 2(1+\upsilon)\rho_i k^{1/4}\gamma_{i^*}$, as a function of $\rho_i$, is minimized by $\rho_i = (1 + \upsilon)\cdot k{1/4}\gamma_{i^*}$, in which case it equals \begin{equation}
				(1 + \upsilon)^2\sqrt{k}\cdot (1 - \gamma^2_{i^*}).
			\end{equation} We conclude that the first of the two terms in \eqref{eq:initialize_ratiobound2} is at least \begin{equation}
				\frac{(1 + \upsilon)^2\sqrt{k}\cdot (1 - \gamma^2_{i^*})}{1 + (1 + \upsilon)^2\sqrt{k} - 2(1+\upsilon)k^{1/4}\gamma_{i^*}} = 1 - \frac{(1 - (1 + \upsilon)k^{1/4}\gamma_{i^*})^2}{1 + (1 + \upsilon)^2\sqrt{k} - 2(1+\upsilon)k^{1/4}\gamma_{i^*}}.
			\end{equation} Recall that because of event $B[i^*]$, we know $\gamma_{i^*}\le k^{-1/4}(1 + c\Delta^2)$, so \begin{equation}
				(1 - (1 + \upsilon)k^{1/4}\gamma_{i^*})^2 \le (|\upsilon| + c\Delta^2 - c|\upsilon|\Delta^2)^2 \le 2c^2\Delta^4\label{eq:num1}
			\end{equation} for $\upsilon$ sufficiently small. On the other hand, the numerator of the second of the two terms in \eqref{eq:initialize_ratiobound2} is \begin{equation}
				2\omega(1 + \upsilon)\rho_i k^{1/4}\gamma_{i^*} \ge \omega,\label{eq:num2}
			\end{equation} because $\gamma_{i^*}\ge k^{-1/4}$ by event $C$, and because $(1 + \upsilon)\rho_i\ge 1/2$ when $\upsilon$ is sufficiently small and $\rho_i \in\mathcal{I}$. We conclude from \eqref{eq:denom_bound}, \eqref{eq:initialize_ratiobound2}, \eqref{eq:num1}, and \eqref{eq:num2} that \begin{equation}
				\frac{\norm{w_i - v}^2_2}{\norm{w_{i^*} - v}^2_2} \ge 1 + \frac{\omega - 2c^2\Delta^4}{\sqrt{k}}.
			\end{equation} In particular, if we took $\omega =2c^2\Delta^4 + c\Delta^2$, then again we would have $\frac{\norm{w_i - v}^2_2}{\norm{w_{i^*} - v}^2_2} \ge 1 + \frac{c\Delta^2}{\sqrt{k}}$.
		\end{proof}

		Finally, we show that if $\rho_i\in\mathcal{I}$, then events $B[i]$ and $C$ imply \eqref{eq:gammamugap}. We proceed in a manner similar to the proof of \eqref{eq:gammarelation} in Lemma~\ref{lem:Cevent}. As with that proof, the intuition is that if the normalized vectors $\what_i$ and $\what_{i^*}$ are somewhat separated on the unit sphere, then the upper bound on $\langle v, \wperp_i\rangle$ will ensure the existence of a sizable cone around $w_{i^*}$ for which any $v$ inside that cone is much closer to $w_{i^*}$ than to $w_i$. And if instead $\what_i$ and $\what_{i^*}$ are not separated, the fact that their non-normalized counterparts $w_i$ and $w_{i^*}$ are separated implies that $\what_i$ and $\what_{i^*}$ are nearly collinear and thus too separated for $\rho_i\in\mathcal{I}$ to hold.

		\begin{claim}
			Let $i\neq i^*$. If $\rho_i \in\mathcal{I}$ and events $B[i]$ and $C$ occur, then \eqref{eq:gammamugap} must hold. \label{claim:BCIimplyangle}
		\end{claim}

		\begin{proof}
			As with \eqref{eq:gammabreakdown}, note that 
			\begin{equation*}
				\gamma_i = \left\langle \what_i, \what_{i^*}\right\rangle + \left\langle \wperp_i, v\right\rangle.
			\end{equation*} 
			Now if $\langle \what_i, \what_{i^*}\rangle \le 0$, then by event $C$, $\gamma_i \le \langle \wperp_i, v\rangle \le \nu_{\text{perp}}k^{-1/2}\norm{\wperp_i}_2 \ll \gamma_{i^*}(1 - \omega)$, and we'd be done. On the other hand, if $\langle \what_i, \what_{i^*}\rangle > 0$, then we get that 
			\begin{equation*}
				\gamma_i \le \left\langle \what_i, \what_{i^*}\right\rangle + \nu_{\text{perp}}k^{-1/2}\norm{\wperp_i}_2.
			\end{equation*} 
			In this case, to show the desired inequality \eqref{eq:gammamugap}, it would suffice to show that 
			\begin{equation*}
				\gamma_{i^*}\left(1 - \omega - \left\langle \what_i, \what_{i^*}\right\rangle\right) \ge \nu_{\text{perp}}k^{-1/2}\norm{\wperp_i}_2.
			\end{equation*} 
			In particular, because $\gamma_{i^*}\ge k^{-1/4}$, we just need to show that 
			\begin{equation}
				1 - \omega - \left\langle \what_i, \what_{i^*}\right\rangle \ge \nu_{\text{perp}}k^{-1/4}\norm{\wperp_i}_2.\label{eq:gammamugap_suffice}
			\end{equation} 
			After squaring both sides of \eqref{eq:gammamugap_suffice}, making the substitution $\norm{\wperp_i}^2_2 = 1 - \left\langle \wperp_i, \wperp_{i^*}\right\rangle^2$, and rearranging, \eqref{eq:gammamugap_suffice} becomes 
			\begin{equation}
				\left(1 - \omega - \left \langle \what_i, \what_{i^*}\right\rangle\right)^2 - \nu_{\text{perp}}^2 k^{-1/2}\cdot\left(1 - \left\langle \wperp_i, \wperp_{i^*}\right\rangle^2\right) \ge 0\label{eq:quadratic_initialize}
			\end{equation} 
			This is merely a univariate inequality for a quadratic polynomial in $\langle \what_i, \what_{i^*}\rangle$. The roots of this polynomial are given by 
			\begin{equation*}
				\langle \what_i, \what_{i^*}\rangle = \frac{1 - \omega \pm \sqrt{\frac{\nu_{\text{perp}}^4}{k} + \frac{2\nu_{\text{perp}}^2\omega}{\sqrt{k}} - \frac{\nu_{\text{perp}}^2\omega^2}{\sqrt{k}}}}{1 + \nu_{\text{perp}}^2/\sqrt{k}}.
			\end{equation*} 
			Observe that 
			\begin{align*}
				\sqrt{\frac{\nu_{\text{perp}}^4}{k} + \frac{2\nu_{\text{perp}}^2\omega}{\sqrt{k}} - \frac{\nu_{\text{perp}}^2\omega^2}{\sqrt{k}}} &= \omega\cdot\sqrt{\left(1 + \frac{\nu_{\text{perp}}^2}{\omega\sqrt{k}}\right)^2 - \left(1 + \frac{\nu_{\text{perp}}^2}{\sqrt{k}}\right)} \\
				&\le \omega\cdot \sqrt{\frac{2\nu_{\text{perp}}^2}{\omega\sqrt{k}} + \frac{\nu_{\text{perp}}^4}{\omega^2 k}} \le \frac{2\nu_{\text{perp}}\cdot\omega^{1/2}}{k^{1/4}} \le a_{\text{arb}}\omega,
			\end{align*} where the last step holds for any absolute constant $a_{\text{arb}}>0$ for sufficiently large $k$. We see that the inequality \eqref{eq:quadratic_initialize} is satisfied provided that $\langle \what_i, \what_{i^*}\rangle$ lies outside the interval \begin{equation*}
				\mathcal{J}\triangleq \left[\frac{1 - (1 + a_{\text{arb}})\omega}{1 + \nu_{\text{perp}}^2/\sqrt{k}},\frac{1 - (1 - a_{\text{arb}})\omega}{1 + \nu_{\text{perp}}^2/\sqrt{k}}\right]\subset \left[1 - (1 + 2a_{\text{arb}})\omega, 1 - (1 - 2a_{\text{arb}})\omega\right].
			\end{equation*} It remains to show that under the hypotheses of the claim, we cannot have $\langle \what_i, \what_{i^*}\rangle\in\mathcal{J}$. This is where we will crucially use the fact that $\norm{w_i - w_{i^*}}_2 \ge \Delta$.

			Suppose to the contrary that $\langle \what_i, \what_{i^*}\rangle\in\mathcal{J}$. In particular, this implies \begin{equation}
				\langle \what_i, \what_{i^*}\rangle \ge 1 - (1 + 2a_{\text{arb}})\omega. \label{eq:lowerboundarbmu}
			\end{equation} Now note that \begin{equation}
				\Delta^2 \le \norm{w_i - w_{i^*}}^2_2 = \norm{w_i}^2_2 + \norm{w_{i^*}}^2_2 - 2\langle \what_i, \what_{i^*}\rangle \norm{w_i}_2 \norm{w_{i^*}}_2,
			\end{equation} so \begin{equation}
				1 - (1 + 2a_{\text{arb}})\omega \le \langle \what_i,\what_{i^*}\rangle \le  \frac{1}{2}\left(\rho_i + 1/\rho_i\right) - \frac{\Delta^2}{2\norm{w_i}_2\norm{w_{i^*}}_2} \le \frac{1}{2}\left(\rho_i + 1/\rho_i\right) - \Delta^2/2.
			\end{equation} For $c$ sufficiently small, $\Delta^2/2 - (1 + 2a_{\text{arb}})\omega \ge \Delta^2/4$, so by taking $\beta = \Delta/2$ in Fact~\ref{fact:elem_ineq} we conclude that $\rho_i\not\in[1 - \Delta/4, 1+ \Delta/4]$. We get a contradiction upon noting that if $2c^{1/2} < 1/4$, then $\rho_i \not\in\mathcal{I}$.		
		\end{proof}

		The proof of Claim~\ref{claim:BCimplyE} now follows. Take $\tau'_{\mathrm{gap}} = c$ in the statement of Lemma~\ref{lem:rand_init}. Then for every $i\neq i^*$, either $\rho_i \in\mathcal{I}$, in which case we are done by Claim~\ref{claim:BCIimplyE}. Otherwise, $\rho_i\not\in\mathcal{I}$, in which case Claim~\ref{claim:BCIimplyangle} implies \eqref{eq:gammamugap} holds, and then we are done by Claim~\ref{claim:BCIimplyangle}.
	\end{proof}

	To complete the proof, we must show that the event that $B[i]$ and $C$ occur simultaneously for all $i\in[k]$ is at least $\exp(-O(\sqrt{k}/\Delta^2))$. The proofs for these facts are essentially identical to those of Claims~\ref{claim:Aevent_prob_bound}, \ref{claim:Bevent_prob_bound}, and \ref{claim:probCj} in the proof of Lemma~\ref{lem:stay_main}, so we omit them.

	\begin{claim}
		For any $i\in[k]$, $\Pr[B[i]] \ge 1 - e^{-\overline{a}\sqrt{k}/\Delta^2}$ for some absolute constant $\overline{a}>0$.
	\end{claim}

	\begin{claim}
		For any $i\in[k]$, $\Pr[C] \ge \frac{1}{\poly(k)}e^{-\underline{a}\sqrt{k}/\Delta^2}$ for some absolute constant $\underline{a} > 0$ such that $\overline{a} - \underline{a}$ is nonnegative and strictly increasing in $\Delta$.
	\end{claim}

	Lemma~\ref{lem:rand_init} now follows by a union bound: the probability that all $B[i]$ occur is at least $1 - k\cdot e^{-\overline{a}\sqrt{k}/\Delta^2}$, and the probability that $C$ occurs is $\frac{1}{\poly(k)}e^{-\underline{a}\sqrt{k}/\Delta^2}$, so the probability all of these events occur is at least $\frac{1}{\poly(k)}e^{-\underline{a}\sqrt{k}/\Delta^2} - k\cdot e^{-\overline{a}\sqrt{k}/\Delta^2} = \exp(-O(\sqrt{k^2}/\Delta))$, and then we are done by Claim~\ref{claim:BCimplyE}.
\end{proof}

Lastly, we remark that Lemma~\ref{lem:rand_init} only applies to $w_{i^*}$ for which $\norm{w_{i^*}}_2 \ge \underline{\sigma}$. We could for instance take $\underline{\sigma} = \epsilon/4$ and this would not affect the asymptotics of our runtime. Now for regressors $w_i$ whose norm is less than $\epsilon/4$, we can simply output an arbitrary vector $a$ of norm $\epsilon/4$ as an $\epsilon/2$-close estimate, by triangle inequality. We can also easily check whether there is indeed such a short regressor $w_i$, e.g. by estimating the minimum variance of the univariate mixture $\mathcal{F}$ given by sampling $(x,y)\sim \calD$ and computing $y - \langle a,x\rangle$ (see \textsc{CheckOutcome}).

\subsection{Algorithm Specification}
\label{subsec:learnwithnoise_spec}

We are now ready to describe our algorithm \textsc{LearnWithNoise} for learning all components of $\calD$. The key subroutines are: 
\begin{itemize}
	\item \textsc{OptimisticDescent} (Algorithm~\ref{alg:optimistic_moment_descent}): the pseudocode for this is nearly identical to that of \textsc{FourierMomentDescent}, except \textsc{OptimisticDescent} additionally takes as input an initialization, has a different output guarantee, and has slightly different parameters which are tuned to fit the regime of Lemma~\ref{lem:stay_main}.
	\item \textsc{CheckOutcome} (Algorithm~\ref{alg:check_outcome_low_noise}): \textsc{CheckOutcome} is used to check whether a given estimate is close to any regressor of $\calD$. This only needs to be used to check whether there exists a short regressor, as discussed at the end of the previous Section~\ref{subsec:randinit}.
\end{itemize}

\begin{algorithm}[!t]\caption{\textsc{OptimisticDescent}$(\calD,a_0,\delta)$}\label{alg:optimistic_moment_descent}
\begin{algorithmic}[1]
	\State \textbf{Input}: Sample access to mixture of linear regressions $\calD$ with noise rate $\noise^2 < a_{\mathrm{noise}}\cdot \epsilon^2$, initial vector $a_0\in\R^d$, $\delta > 0$
	\State \textbf{Output}: $a_T\in\R^d$ such that $\norm{w_{i^*} - a_T}_2 \le\epsilon$ with probability at least $1 - \delta$ and, with probability at least $\exp(-\widetilde{O}(\sqrt{k}/\Delta^2))$, additionally $\calE_{a_T}[i^*]$ holds with probability $\exp(-\widetilde{O}(\sqrt{k}/\Delta^2))$ provided $\calE_{a_0}[i^*]$ holds for some $i^*\in[k]$.
		\State Set $T = \Omega(\sqrt{k}\cdot\ln(1/\epsilon))$.
		\State Set $\delta' = \frac{\delta}{4T}$.
		\State Set $M = e^{-a_{\mathrm{trials}}\sqrt{k}/\Delta^2}\ln(3/\delta)$.
		\State Set $\delta'' = \frac{\delta}{4MT}$.
		\State Let $\overline{\sigma} = a_{\mathrm{scale}}\cdot k^{-1/4}$.
		\State Set $\underline{\sigma} = \epsilon/3$.
		\For{$0\le t < T$}
			\State \multiline{Let $\calF_t$ be the univariate mixture of Gaussians which can be sampled from by drawing $(x,y)\sim\calD$ and computing $y - \langle x,a_t\rangle$.}
			\State Let $p = 20\ln\left(\frac{3}{2\pmin}\right)+1$.
			\State Let $\kappa_1 = (\underline{\beta}' - 3c/2)\Delta^2k^{-1/2}$, $\kappa_2 = (\underline{\beta}' - c/2)\Delta^2k^{_1/2}$.
			\State $\underline{\sigma}^{\mathrm{sharp}}_t \triangleq$~\textsc{EstimateMinVariance}($\calF_t,\overline{\sigma}, \underline{\sigma},p,\delta'$). \Comment{Algorithm~\ref{alg:estimate_min_variance}}
			\If{$\underline{\sigma}^{\mathrm{sharp}}_t < 0.99\epsilon$}
				\State Output $a_t$.\label{line:break_all}
			\EndIf
			\State \multiline{Draw $N_1 \triangleq \tilde{\Omega}\left(\frac{\overline{\sigma}^2}{(\underline{\sigma}^{\mathrm{sharp}}_t)^2}\cdot p^{-2}_{\min}\cdot \poly(k)\cdot d \cdot \ln(k/\delta')\right)$ i.i.d. samples $\{(x_i,y_i)\}_{i\in[N]}$ from $\calD$ and form the matrix $\widehat{\vec{M}}^{(N)}_{a_t}$.} \label{line:n1_all}
			\State Let $U_t = \textsc{ApproxBlockSVD} (\widehat{\vec{M}}^{(N)}_{a_t}, 1/\poly(k))$. \Comment{Lemma~\ref{lem:correlation}}
			\For{$j\in[M]$}
				\State Sample $g^{(j)}_t\sim\N(0,\Id_k)$ and define $v^{(j)}_t = \frac{U_t g^{(j)}_t}{\norm{U_t g^{(j)}_t}_2}\in\S^{d-1}$.
				\State Let $a'^{(j)}_t = a_t + \eta_t v_j$ for $\eta_t \triangleq a_{\mathrm{LR}}\cdot\Delta\cdot \underline{\sigma}^{\mathrm{sharp}}_t\cdot k^{-1/4}$.
				\State \multiline{Let $\calF'^{(j)}_t$ be the univariate mixture of Gaussians which can be sampled from by drawing $(x,y)\sim\calD$ and computing $y - \langle x,a'^{(j)}_t\rangle$.}
				\If{\textsc{CompareMinVariances}($\calF_t,\calF'^{(j)}_t,\overline{\sigma},\underline{\sigma},\kappa_1,\kappa_2,\delta''$) = $\mathsf{true}$} \Comment{Algorithm~\ref{alg:compare_min_variances}, Corollary~\ref{cor:compare_min_variances}}
					\State Set $a_{t+1} = a'^{(j)}_t$
					\Break
				\EndIf
			\EndFor
		\EndFor
		\State Output $a_T$.
\end{algorithmic}
\end{algorithm}

\begin{algorithm}\caption{\textsc{LearnWithNoise}($\calD, \delta, \epsilon$), Lemma~\ref{lem:learnwithnoise_correct}}\label{alg:learnwithnoise}
\begin{algorithmic}[1]
	\State \textbf{Input}: Sample access to mixture of linear regressions $\calD$ with regressors $\{w_1,...,w_k\}$ and separation $\Delta$ and noise rate $\noise$, failure probability $\delta$, error $\epsilon < \Delta/4$
	\State \textbf{Output}: List of vectors $\mathcal{L}\triangleq \{\tilde{w}_1,...,\tilde{w}_k\}$ for which there is a permutation $\pi:[k]\to[k]$ for which $\norm{w_i - w_{\pi(i)}}_2 \le \epsilon$ for all $i\in[k]$, with probability at least $1 - \delta$.
			\State Set $\underline{\sigma} = \epsilon/4$.
			\State Set $\mathcal{L} = \emptyset$.
			\State Set $\upsilon_*$ to be a sufficiently small $\poly(\Delta,1/k)$ and define the mesh $\mathcal{S}$ via \eqref{eq:mesh}.
			\State Set $v_{\text{tiny}}$ to be a random vector of norm $\epsilon/4$.
			\If{\textsc{CheckOutcome}$(\calD, v_{\text{tiny}}, \underline{\sigma}, \delta/5) = \text{$\mathsf{true}$}$}
				\State Add $v_{\text{tiny}}$ to $\mathcal{L}$
			\EndIf
			\State Set $W = \exp(O(\sqrt{k}/\Delta^2))\cdot\ln(2k/\delta)$.
			\State Set $\delta^* = \frac{\delta}{2|\mathcal{S}|W}$
			\For{$\alpha \in \mathcal{S}$}
				\For{$0\le i< W$}
					\State Let $v$ be a Haar-random vector in $\R^{d}$ of norm $\alpha$.
					\State Let $\tilde{v}=$\textsc{OptimisticDescent}($\calD, v,\delta^*$)
					\If{$\norm{\tilde{v} - \tilde{w}}_2 > 2\epsilon$ for all $\tilde{w}\in\mathcal{L}$}
						\State Add $\tilde{v}$ to $\mathcal{L}$
					\EndIf
				\EndFor
			\EndFor
			\State Output $\mathcal{L}$.
\end{algorithmic}
\end{algorithm}

\begin{algorithm}\caption{\textsc{CheckOutcome}$({\cal D},v,\epsilon,\delta)$, Lemma~\ref{lem:check_outcome}}\label{alg:check_outcome_low_noise}
\begin{algorithmic}[1]
	\State \textbf{Input}: Sample access to mixture of linear regressions $\calD$ with noise rate $\noise$, vector $v\in\R^{d}$, threshold $\epsilon > 0$, failure probability $\delta$
	\State \textbf{Output}: $\mathsf{true}$ if $\min_{i\in[k]}\norm{w_i - v}_2 \le \epsilon$, $\mathsf{false}$ if $\min_{i \in [k]} \norm{w_i - v}_2 \ge 2\epsilon$, with probability at least $1 - \delta$
		\State Let $\calF$ be the univariate mixture of Gaussians which can be sampled from by drawing $(x,y)\sim\calD$ and computing $y - \langle v,x\rangle$.
		\State Let $p = 20\ln\left(\frac{3}{2\pmin}\right) + 1$.
		\State Let $\sigma^* = $ \textsc{EstimateMinVariance}($\calF,4,\noise,p,\delta$).
		\If{$(\sigma^*)^2 \le 2\epsilon^2$}
			\State Return $\mathsf{true}$
		\Else
			\State Return $\mathsf{false}$
		\EndIf
\end{algorithmic}
\end{algorithm}

\subsection{Proof of Correctness}
\label{subsec:learnwithnoise_correct}

We first give a proof of correctness for \textsc{CheckOutcome}.

\begin{lemma}\label{lem:check_outcome}
	Let $v\in\S^{d-1}$ and $\calD$ be a mixture of linear regressions with noise rate $\noise > \Delta\pmin^3$, and let $\epsilon > 2\noise$. If there is some component $v_{i^*}$ for which $\norm{v - v_{i^*}}_2 \in [-\epsilon,\epsilon]$, then \textsc{CheckOutcome}($\calD, v,\epsilon,\delta$) (Algorithm~\ref{alg:check_outcome_hyperplanes}) returns $\mathsf{true}$ with probability at least $1 - \delta$. Otherwise, if $\norm{v - v_i}_2 > 2\epsilon$ for all $i\in[k]$, then \textsc{CheckOutcome} returns $\mathsf{false}$ with probability at least $1 - \delta$. Furthermore, \textsc{CheckOutcome} has time and sample complexity \begin{align*}
	\tilde{O}\left( k\pmin^{-4}\ln(1/\delta) \cdot \poly \left( \ln(1/\pmin), \ln(1/\Delta) \right)^{\ln(1/\pmin)} \right).
	\end{align*}
\end{lemma}

\begin{proof}
	As usual, $\calF$ is a mixture of univariate Gaussians with variances $\{\noise^2 + \norm{w_i - v}^2_2\}$. By Corollary~\ref{cor:constant_min_approx}, 
	\begin{equation}\sigma^* \in[0.9^2,1.1^2]\cdot( \min_{i \in [k]} \norm{w_i - v}^2_2 + \noise^2).
	\end{equation} 
	If $\min_{i \in [k]} \norm{w_i - v}^2_2 \le \epsilon^2$, then we have that 
	\begin{equation}
	(\sigma^*)^2 \le 1.1^2\left(\epsilon^2 + \noise^2\right) \le 1.21 \cdot (1 + 1/4) \epsilon^2 \le 2\epsilon^2.
	\end{equation} 
	If $\min_{i\in[k]} \norm{w_i - v}^2_2 \ge 4\epsilon^2$, then we have that \begin{equation}
		(\sigma^*)^2 \ge 0.9^2 \left(4\epsilon^2 + \noise^2\right) \ge 0.9^2 \cdot 4\epsilon^2 \ge 3\epsilon^2,
	\end{equation} 
	which completes the proof.
\end{proof}

We can now prove correctness of \textsc{LearnWithNoise}.

\begin{lemma}
	Let $a_{\mathrm{noise}} > 0$ be the constant defined in Lemma~\ref{lem:stay_main}, and let $\epsilon > 0$. Let $\calD$ be a mixture of spherical linear regressions with mixing weights $\{p_i\}$, directions $\{w_i\}$, and noise rate $\noise$.
	For any $\epsilon,\delta > 0$ and $\noise^2 \le\epsilon^2/10$, with probability at least $1 - \delta$, \textsc{LearnWithNoise}($\calD,\delta,\epsilon$) (Algorithm~\ref{alg:learnwithnoise}) outputs a list of vectors $\mathcal{L} = \{w_1,...,w_k\}$ such that there exists a permutation $\pi: [k]\to[k]$ for which $\norm{w_i - w_{\pi(i)}}_2 \le \epsilon$ for all $i\in[k]$.\label{lem:learnwithnoise_correct}
\end{lemma}

\begin{proof}
	Let $v\in\R^d$, and consider a run of \textsc{OptimisticDescent}($\calD, v, \delta^*$). Let $a_t$ be the iterate at time $t$ in this run. Let $\rho = 1.1/0.9$.

	We first note that if \textsc{OptimisticDescent} breaks out at Line~\ref{line:break_all}, the vector it returns is close to some component of $\calD$.

	\begin{claim}
		If for some $0\le t < T$, \textsc{OptimisticDescent} (Algorithm~\ref{alg:optimistic_moment_descent}) breaks out at Line~\ref{line:break_all} and outputs $a_t$, then $\min_i\norm{w_i - a_t}_2 \le \epsilon$.\label{claim:ifbreakwin}
	\end{claim}

	\begin{proof}
		If \textsc{OptimisticDescent} (Algorithm~\ref{alg:optimistic_moment_descent}) breaks out at Line~\ref{line:break_all}, it is because $\sigma^*_t \le 0.99\epsilon$. This implies that $\min_i\norm{w_i - a_t}^2 + \noise^2 \le (\sigma^*_t)^2 / 0.99^2 \le \epsilon^2$, so $\min_i\norm{w_i - a_t}_2 \le \epsilon$ as claimed.
	\end{proof}

	Next, we show that if $a_t$ is still somewhat far from any component, with high probability over the next iteration either \textsc{OptimisticDescent} will break out at Line~\ref{line:break_all}, or the progress measure will contract.

	\begin{claim}
		Let $i^*$ be the index minimizing $\norm{w_i - a_t}_2$. If $\norm{w_{i^*} - a_t}^2_2 \ge \epsilon^2/2$, then with probability at least $1 - \delta^*/T$ over the next iteration of \textsc{OptimisticDescent} (Algorithm~\ref{alg:optimistic_moment_descent}), either of two things will happen: \begin{enumerate}
			\item $\epsilon^2/16 \le \min_{i\in[k]} \norm{w_i - a_{t+1}}^2_2 \le \epsilon^2/\rho^2 - \noise^2$.
			\item $\min_{i\in[k]} \norm{w_i - a_{t+1}}^2_2 \ge \epsilon^2/2$ and \begin{equation}\sigma^2_{t+1} \le \left(1 - \underline{\beta}\Delta^2/\sqrt{k}\right)\cdot \sigma^2_t.\label{eq:one_step_progress}\end{equation} 
		\end{enumerate}

		Condition on either of these outcomes happening. Then additionally, if event $\calE_{a_t}[i^*]$ (see Lemma~\ref{lem:rand_init}) holds, then with probability at least $1/\poly(k)$, $\calE_{a_{t+1}}[i^*]$ holds.\label{claim:one_step_progress}
	\end{claim}

	\begin{proof}
		Condition on outcomes 1), 2), and 3) of Lemma~\ref{lem:stay_main}, which all happen with probability at least $1 - \delta^*/T$.

		Now if $\min_{i\in[k]} \norm{w_i - a_{t+1}}^2_2 \ge \epsilon^2/2$, then 2) is just a consequence of outcomes 1) and 2) of Lemma~\ref{lem:stay_main}.

		If $\min_{i\in[k]} \norm{w_i - a_{t+1}}^2_2 \le \epsilon^2/2$, then by outcome 3) of Lemma~\ref{lem:stay_main}, $\sigma_{t+1} \ge (1 - \overline{\beta}\Delta^2/\sqrt{k})\sigma_{t}\ge 0.99\sigma_t\ge \epsilon^2/16$.

		The last part of the claim is a consequence of outcome 4) of Lemma~\ref{lem:stay_main}, which happens in addition to outcomes 1), 2), and 3) with probability $1/\poly(k)$, provided $\calE_{a_t}[i^*]$ occurs.
	\end{proof}

	Next we show that if at some time $t$ there is an $i\in[k]$ for which $\norm{w_i - a_t}^2_2 \le \epsilon^2/\rho^2 - \noise^2$, then \textsc{OptimisticDescent} will break out at Line~\ref{line:break_all} and correctly output $a_t$.

	\begin{claim}
		If for some $0\le t < T$ we have \begin{equation}
			\epsilon^2/16\le \min_{i\in[k]} \norm{w_i - a_t}^2_2 \le \epsilon^2/\rho^2 - \noise^2,\label{eq:crossed2}
		\end{equation} then \textsc{OptimisticDescent} (Algorithm~\ref{alg:optimistic_moment_descent}) breaks out at Line~\ref{line:break_all} and returns $a_t$.\label{claim:breakall}
	\end{claim}
	
	\begin{proof}
		The lower bound in \eqref{eq:crossed2} implies that the $\underline{\sigma} = \epsilon/4$ which is passed to \textsc{EstimateMinVariance} is a valid lower bound for $\sigma_t$.

		The upper bound in \eqref{eq:crossed2} implies that \begin{equation}(\sigma^*_t)^2 \le 1.21\sigma^2_t \le 1.21\cdot(\min_{i \in [k]}\norm{w_i - a_t}^2_2 + \noise^2) \le 0.99^2\epsilon^2,\end{equation} so \textsc{OptimisticDescent} breaks out at Line~\ref{line:break_all} and outputs $a_t$. The bound on $\norm{w_i - a_{t+1}}_2$ immediately follows from \eqref{eq:crossed2}.
	\end{proof}

	Claims~\ref{claim:ifbreakwin}, \ref{claim:one_step_progress}, and \ref{claim:breakall} imply that with high probability the output of \textsc{OptimisticDescent} (Algorithm~\ref{alg:optimistic_moment_descent}) is close to some component of $\calD$.

	\begin{claim}
		For any $v\in\R^d$, \textsc{OptimisticDescent}$(\calD, v,\delta^*)$ (Algorithm~\ref{alg:optimistic_moment_descent}) outputs some vector $\tilde{v}$ for which $\min_i\norm{\tilde{v} - w_i} \le \epsilon$ with probability at least $1 - \delta^*$.\label{claim:optimisticwhp}
	\end{claim}

	\begin{proof}
		By Claims~\ref{claim:ifbreakwin} and \ref{claim:breakall}, it suffices to consider the case where there does not exist $0\le t < T$ for which \eqref{eq:crossed2} holds. Then $\langle{w_i - a_t}^2_2 \ge \epsilon^2/2$ for every $t$, so by Claim~\ref{claim:one_step_progress},
		\begin{equation}\min_{i\in[k]}\norm{\tilde{v} - w_i}^2_2 + \noise^2 \le \left(1 - \underline{\beta}\Delta^2/\sqrt{k}\right)^T\left(\min_{i\in[k]}\norm{v - w_{i}}^2_2 + \noise^2\right)\le \epsilon^2\label{eq:Tprogress},\end{equation} where the last inequality follows by taking $T = \frac{2\sqrt{k}}{\underline{\beta}\Delta^2}\ln(1/\epsilon)$.
	\end{proof}

	For $v\in\R^d$, denote by $Z_v$ the event in Claim~\ref{claim:optimisticwhp}.

	We next use Lemma~\ref{lem:rand_init} and part 4) of Claim~\ref{claim:one_step_progress} to lower bound the probability that $v$ chosen in the inner loop of \textsc{LearnWithNoise} ends up being closest to any given component of $\calD$.

	\begin{claim}
		Take any $i^*\in[k]$ for which $\norm{w_{i^*}}_2 \ge \underline{\sigma}$. Then there exists some $\alpha\in\mathcal{S}$ such that if $v$ is a Haar-random vector in $\R^{d}$ of norm $\alpha$, then conditioned on $\calE_v$, the output $\tilde{v}$ of \textsc{OptimisticDescent}($\calD,v$) (Algorithm~\ref{alg:optimistic_moment_descent}) is closest to $w_{i^*}$ with probability at least $\exp(-\widetilde{O}(\sqrt{k}/\Delta^2))$ over the choice of $v$.\label{claim:lockin}
	\end{claim}

	\begin{proof}
		Take any $i^*\in[k]$ for which $\norm{w_{i^*}}_2 \ge \underline{\sigma}$. Recall the definition of $\calE_v[i]$ from Lemma~\ref{lem:rand_init}. By Lemma~\ref{lem:rand_init}, there is some $\alpha\in\mathcal{S}$ such that if $v$ is a Haar-random vector in $\R^{d}$ of norm $\alpha$, then with probability at least $q_1 = \exp(-O(\sqrt{k}/\Delta^2))$ over $v$, $\calE_v[i^*]$ holds. Then by 4) in Lemma~\ref{lem:stay_main}, with probability $q_2 = \exp(-O(\sqrt{k}/\Delta^2))\cdot\poly(k)^{-T}$, $\calE_{v^*}[i^*]$ holds for, where $\tilde{v}$ is the output of \textsc{OptimisticDescent}($\calD,v$). This completes the proof, as $q_1\cdot q_2 = \exp(-\widetilde{O}(\sqrt{k}/\Delta^2) )$.
	\end{proof}

	We can now complete the proof of Lemma~\ref{lem:learnwithnoise_correct}. Take $\delta^* = \frac{\delta}{2W\cdot|\mathcal{S}|}$ so that $Z_v$ holds for all $v$ sampled in \textsc{LearnWithNoise} with probability at least $1 - \delta/2$, by Claim~\ref{claim:optimisticwhp}. In this case, any $\tilde{v}$ produced in the course of \textsc{LearnWithNoise} must be a $\epsilon$-close to some component of $\calD$.

	Then by Claim~\ref{claim:lockin}, for any $i^*\in[k]$ for which $\norm{w_{i^*}}_2 \ge \underline{\sigma} = \epsilon/4$, the probability that some $\tilde{v}$ produced in the course of \textsc{LearnWithNoise} is $\epsilon$-close to $i^*$ is at least $1 - (1 - q)^{W}$, where $q\triangleq \exp (-\widetilde{O}(\sqrt{k}/\Delta^2) )$. By taking $W= \ln(2k/\delta)/q$, we ensure that this happens with probability at least $\frac{\delta}{2k}$. We conclude by a union bound over $[k]$ that for every $i^* \in [k]$ for which $\norm{w_{i^*}}_2\ge \epsilon/4$, there is some $\tilde{v}$ produced in the course of \textsc{LearnWithNoise} which is $\epsilon$-close to $i^*$.

	Furthermore, by triangle inequality note that we never add vectors $\tilde{v}$ to $\mathcal{L}$ which are $\epsilon$-close to a component which is already $\epsilon$-close to an existing $\tilde{w}\in\mathcal{L}$.

	Lastly, for $i^*\in[k]$ for which $\| w_{i^*}\|_2 \le \epsilon/4$, note that any vector $v_{\text{tiny}}$ of norm $\epsilon/4$ is $\epsilon/2$-close to $w_{i^*}$. This completes the proof of Lemma~\ref{lem:learnwithnoise_correct}.
\end{proof}

\begin{lemma}[Running time of \textsc{LearnWithNoise}]
Let
\begin{align*} 
N_1 &= \widetilde{O} ( \epsilon^{-2} \pmin^{-2}  d k^2 \ln(1 / \delta) ) \\
N &=   \pmin^{-4} k \ln ( 1 / \delta ) \cdot \poly \left( \overline{\sigma}, \sqrt{k}/\Delta^2, \ln ( 1 / \pmin ), \ln (1 / \underline{\sigma} ) \right)^{O\left(\sqrt{k} \ln (1 / \pmin)/\Delta^2\right)} \; .
\end{align*}
Then \textsc{LearnWithNoise} (Algorithm~\ref{alg:learnwithnoise}) requires sample complexity $\widetilde{O} (e^{\widetilde{O}(\sqrt{k}/\Delta^2)} (N_1 + N) )$ and runs in time $\widetilde{O} (e^{\widetilde{O}(\sqrt{k}/\Delta^2)} (d N_1 + N) )$
\label{lem:runtime_learnwithnoise}
\end{lemma}

\begin{proof}
	The complexity of the calls to \textsc{CheckOutcome} is dominated by the calls to \textsc{OptimisticDescent}, whose time and sample complexity are essentially identical to those of \textsc{FourierMomentDescent} called with failure probability parameter $\frac{\delta^*}{2W\cdot |\mathcal{S}|}$, except the complexity of the calls to \textsc{CompareMinVariances} now has exponential dependence on $\widetilde{O}(\sqrt{k}\ln(1/\pmin)/\Delta^2)$ rather than on $\widetilde{O}(\sqrt{k}\ln(1/\pmin))$ because we only make $1 - \sqrt{k}/\Delta^2$ multiplicative progress at each step. Note the complexity of \textsc{FourierMomentDescent} depends only logarithmically on the inverse accuracy, so $\ln\left(\frac{\delta^*}{2W\cdot |\mathcal{S}|}\right)$ is absorbed into the $\widetilde{O}(\cdot)$. We conclude by noting that \textsc{OptimisticDescent} is called $|\mathcal{S}|\cdot W$ times, where $|\mathcal{S}| \le \poly(k)\cdot\ln(1/\epsilon)$ and $W = \exp(O(\sqrt{k}/\Delta^2))\cdot\ln(2k/\delta)$.
\end{proof}

We can now complete the proof of Theorem~\ref{thm:learnwithnoise_main}.

\begin{proof}[Proof of Theorem~\ref{alg:learnwithnoise}]
	By Lemma~\ref{lem:learnwithnoise_correct}, \textsc{LearnWithNoise} outputs a list of vectors $\{\tilde{w}_1,...,\tilde{w}_k\}$ for which there exists a permutation $\pi:[k]\to[k]$ such that $\norm{w_i - \tilde{w}_i}_2 \le \epsilon$ for all $i\in[k]$. The runtime and sample complexity bounds follow from Lemma~\ref{lem:runtime_learnwithnoise}.
\end{proof}

\subsection{Tolerating More Regression Noise}
\label{subsec:boost_noisy}

In this subsection we briefly remark that in the case where the mixing weights of $\calD$ are known, we can combine Theorem~\ref{thm:learnwithnoise_main} with the local convergence result of \cite{kwon2019converges} to drive error down to $\epsilon$ even in settings where regression noise greatly exceeds $\epsilon$.

\begin{theorem}[\cite{kwon2019converges}, Theorem 3.2]
	Let $\epsilon > 0$. If $\calD$ is a mixture of linear regressions with known mixing weights, separation $\Delta$, and noise rate $\noise \le O\left(\frac{\Delta}{k^2 \poly\log(k)}\right)$, and $\tilde{w}_1,...,\tilde{w}_k\in\R^d$ is a list of vectors for which there exists a permutation $\pi:[k]\to[k]$ for which $\norm{\tilde{w}_i - w_{\pi(i)}}_2 \le O\left(\frac{\Delta}{k^2}\right)$, then with high probability finite-sample EM on a batch of $\widetilde{O}(d)\cdot \poly(k,\ln(1/\epsilon))$ samples converges at an exponential rate to $\tilde{w}^*_1,...,\tilde{w}^*_k\in\R^d$ for which there exists a permutation $\pi':[k]\to[k]$ such that \begin{equation}
		\max_{i \in [k]} \norm{\tilde{w}^*_i - w_{\pi'(i)}}_2 \le O(\epsilon).
	\end{equation}\label{thm:caramanis}
\end{theorem}

In particular, this implies the following:

\begin{theorem}\label{thm:learnwithnoise_plus}
	Given $\delta,\epsilon>0$ and a mixture of spherical linear regressions $\calD$ with regressors $\{w_1,...,w_k\}$, separation $\Delta$, noise rate $\noise \le O\left(\frac{\Delta}{k^2\poly\log(k)}\right)$, and \emph{known mixing weights}, there is an algorithm which with high probability outputs a list of vectors $\mathcal{L}\triangleq \{\tilde{w}_1,...,\tilde{w}_k\}$ for which there is a permutation $\pi:[k]\to[k]$ for which $\norm{\tilde{w}_i - w_{\pi(i)}}_2 \le \epsilon$ for all $i\in[k]$. Furthermore, \textsc{LearnWithNoise} requires sample complexity 
	\begin{align*}
	N = \widetilde{O}\left(d\ln(1/\delta)\pmin^{-4}\cdot\poly\left(k,1/\Delta,\ln(1/\pmin)\right)^{O(\sqrt{k}\ln(1/\pmin)/\Delta^2)}\right)
	\end{align*}
	and time complexity $Nd\cdot \poly\log(k,d,1/\Delta,1/\pmin)$.
\end{theorem}

\begin{proof}
	Simply run \textsc{LearnWithNoise} to learn to error $\epsilon' = O(\Delta/k^2)$, which is possible because $\noise = O\left(\Delta/(k^2 \poly\log(k))\right) \ll \epsilon'$. Then run finite-sample EM initialized to the output of \textsc{LearnWithNoise} to learn to error $\epsilon$.
\end{proof}



\section{Learning Mixtures of Hyperplanes}
\label{sec:hyperplanes}

In this section we show that our techniques extend to give a sub-exponential time algorithm for learning mixtures of hyperplanes. Formally, we show the following:

\begin{theorem}\label{thm:hyperplanes_main}
	Given $\delta,\epsilon>0$ and a mixture of hyperplanes $\calD$ with directions $\{v_1,...,v_k\}$, separation $\Delta$, with probability at least $1 - \delta$, \textsc{LearnHyperplanes}($\calD,\delta,\epsilon$) (Algorithm~\ref{alg:learnhyperplanes}) returns a list of unit vectors $\mathcal{L}\triangleq \{\tilde{v}_1,...,\tilde{v}_k\}$ for which there is a permutation $\pi:[k]\to[k]$ and signs $\epsilon_1,...,\epsilon_k\in\{\pm 1\}$ for which $\norm{\tilde{v}_i - \epsilon_i v_{\pi(i)}}_2 \le \epsilon$ for all $i\in[k]$. Furthermore, \textsc{LearnHyperplanes} requires sample complexity 
	\begin{align*}
	N = \widetilde{O}\left(d\ln(1/\epsilon)\ln(1/\delta)\pmin^{-4}\Delta^{-2}\cdot\poly\left(k,\ln(1/\pmin),\ln(1/\Delta)\right)^{O(k^{3/5}\ln(1/\pmin))}\right)
	\end{align*}
	and time complexity $Nd \cdot \poly \log(k,d,1/\Delta,1/\pmin,1/\epsilon)$.
\end{theorem}

In Section~\ref{subsec:hyperplanes_moment} we show the key fact that a random step will contract $\min_i \norm{\Pi_i a_t}_2$ by a factor of $1 - \Theta(k^{-3/5})$ with probability at least $\exp(-k^{3/5})$, provided we use a suitable initialization. In Section~\ref{subsec:single_hyperplane} we give the full specification for \textsc{HyperplaneMomentDescent} which can learn a single component in a mixture of hyperplanes. In Section~\ref{subsec:hyperplanes_correct} we prove correctness for \textsc{HyperplaneMomentDescent}. In Section~\ref{subsec:hyperplane_boost} we show that by properly regarding a mixture of hyperplanes as a mixture of well-conditioned, non-spherical MLRs, we can invoke the boosting result of \cite{li2018learning} to amplify a warm start obtained by \textsc{HyperplaneMomentDescent} to an estimate with arbitrarily small error. Finally, in Section~\ref{subsec:learnall_hyperplanes} we combine all of these primitives to obtain \textsc{LearnHyperplanes} and prove Theorem~\ref{thm:hyperplanes_main}.

\subsection{Moment Descent for Hyperplanes}
\label{subsec:hyperplanes_moment}

In this section we give the key technical ingredients for showing that a suitable modification of \textsc{FourierMomentDescent} (Algorithm~\ref{alg:fourier_moment_descent}) can also be used to learn mixtures of hyperplanes.

Similar to the case of mixtures of linear regressions, here the first step is to estimate the span of the directions $\{v_i\}$. Define the matrix 
\begin{equation*}
	\vec{M} \triangleq \Id - \E_{x\sim\calD}[xx^{\top}] = \sum^k_{i=1}p_i v_i v_i^{\top}
\end{equation*} 
and let $\widehat{\vec{M}}^{(N)} \triangleq \Id - \frac{1}{N}\sum^N_{i=1}x_i x_i^{\top}$ for $x_1,...,x_N$ i.i.d. samples from $\calD$. When the context is clear, we will omit the superscript $(N)$.

We will need the following basic concentration inequality, which follows immediately from e.g. Theorem 4.7.1 of~\cite{vershynin2018high}. 

\begin{fact}[Concentration of sample covariance]
	For any $\delsamp,\delta > 0$, we have that 
	\begin{equation*}
		\Pr\left[\norm{\vec{M} - \widehat{\vec{M}}^{(N)}}_2 > \Omega(1/\pmin) \cdot \left(\sqrt{\frac{d + t}{N}} + \frac{d+t}{N}\right)\right] \le 2e^{-t}.
	\end{equation*}
\end{fact}

\begin{corollary}
	For any $\delsamp,\delta > 0$, if $N = \tilde{\Omega}\left(d\cdot \ln(1/\delta)\cdot \pmin^{-4}\cdot\delsamp^{-2}\right)$, then $\norm{\vec{M} - \widehat{\vec{M}}^{(N)}}_2 \le \pmin\cdot\delsamp/2$ with probability at least $1 - \delta$.\label{cor:hyper_pca_conc}
\end{corollary}

Henceforth, let $\vec{W}\in\R^{k\times d}$ be the matrix whose rows are the top $k$ singular vectors of $\vec{M}$, let $\vec{U}\in\R^{k\times d}$ be the matrix whose rows are the first $k$ singular vectors of $\widehat{\vec{M}}^{(N)}$ for $N$ given in Corollary~\ref{cor:hyper_pca_conc}.

We will need the following basic bound which follows straightforwardly from Lemma~\ref{lem:wedin_angle}.



\begin{corollary}
	Let $a_t\in\S^{d-1}$ lie in the row span of $\vec{U}$, let $v_j = x_j/\norm{x_j}_2$, and let $\Pi_j$ be the projector to the orthogonal complement of $v_j$. Suppose $\delsamp < 1/k$ and $\norm{x_j} \le 1$. Then $\norm{\vec{U}\Pi_j a_t}_2 \ge \norm{\Pi_j a_t}_2 - \delsamp^2 - 2\delsamp$.

	In particular, for any constant $c > 0$, if $\norm{\Pi_j a_t}_2 \ge \tau$ for some $\tau = \poly(k)$, then for sufficiently small $\delsamp = 1/\poly(k)$ we can ensure that $\norm{\vec{U}\Pi_j a_t}_2 \ge (1 - k^{-c})\norm{\Pi_j a_t}$. \label{cor:Uv}
\end{corollary}

\begin{proof}
	First note that 
	\begin{equation}
		\langle\vec{U} a_t, \vec{U} v_j\rangle = \langle a_t,v_j\rangle\pm \delsamp \label{eq:innerprodclose}
	\end{equation} by the first part of Lemma~\ref{lem:wedin_angle} and the fact that $|\langle a_t,v_j\rangle| \le 1$.

	We have that \begin{align*}
		\norm{\vec{U} \Pi_j a_t}^2_2 &= \norm{\vec{U} a_t}^2_2 + \langle a_t,v_j\rangle^2 \norm{\vec{U} v_j}^2_2 - 2\langle a_t,v_j\rangle \langle \vec{U} a_t, \vec{U} v_j\rangle \\
		&= 1  + \langle a_t,v_j\rangle^2 \norm{\vec{U} v_j}^2_2 - 2\langle a_t,v_j\rangle \langle \vec{U} a_t, \vec{U} v_j\rangle\\
		&\ge 1 + \langle a_t,v_j\rangle^2\cdot (1 - \delsamp^2) - 2\langle a_t,v_j\rangle\cdot (\langle a_t,v_j\rangle - \delsamp)\\
		&\ge \norm{\Pi_j a_t}^2_2 - \delsamp^2 - 2\delsamp
	\end{align*} 
	where the second step follows from the fact that $\norm{\vec{U}a_t}^2_2 = \norm{a_t}^2_2 = 1$ as $a_t\in\S^{d-1}$ lies in the row span of $\vec{U}$, the third step follows from applying Corollary~\ref{cor:Uv} and the lower bound of \eqref{eq:innerprodclose}, and the last step follows from the fact that $\norm{\Pi_j a_t}^2_2 = 1 - \langle a_t,v_j\rangle^2$ and $|\langle a_t,v_j\rangle| \le 1$.
\end{proof}

With these preliminary tools in hand, we are ready to prove the main result of this section, the mixture of hyperplanes analogue of Lemma~\ref{lem:progress}.

\begin{lemma}\label{lem:hyperplanes_progress}
	Let $0 < c < 1/2$. There are absolute constants $\alpha > 0, \beta > 1, \nu > 0$ such that for any $\delta > \exp(-\nu\cdot k^{1-2c})$, the following holds for $k$ sufficiently large. Let $\sigma_t \triangleq \min_{i\in[k]}\norm{\Pi_i a_t}_2$, and suppose $\delsamp\le \sigma^2_t/9$. Denote the minimizing index $i$ by $i^*$. For $M\triangleq e^{-\nu\cdot k^{1 - 2c}}\ln(2/\delta)$ and $g_1,...,g_M\sim\N(0,\Id_k)$, let $z_j = \frac{g_j\vec{U}}{g_j\norm{\vec{U}}_2}\in\S^{d-1}$ for $j\in[M]$. Let $\sigma^*$ be a number for which $0.9\sigma_t \le \sigma^* \le 1.1\sigma_t$. Let $\eta = k^{-c}\sigma^*$.

	If $|\langle a_t, v_{i^*}\rangle| \ge k^{-c}$, then we have that with probability at least $1 - \delta$, 
	\begin{enumerate}
		\item There exists at least one $j\in[M]$ for which $\norm{\Pi_{i^*}(a_t - \eta\cdot z_j)}^2_2 \le \left(1 - \frac{\alpha}{k^{3c}}\right)\sigma^2_t$. Denote any one of these indices by $j^*$.
		\item For all $j\in[M]$ and $i\in[k]$, 
		\begin{equation*}
		\frac{\norm{\Pi_i(a_t - \eta z_j)}^2_2}{\sigma^2_t} \ge \left(\frac{\norm{\Pi_{i^*}(a_t - \eta z_{j^*})}^2_2}{\sigma^2_t}\right)^{\beta}.
		\end{equation*}
	\end{enumerate}
\end{lemma}

\newcommand{\ceeone}{c}
\newcommand{\ceetwo}{c}

\begin{proof}
	Without loss of generality, suppose $\langle a_t, v_{i^*}\rangle \ge 0$. For $j\in[M]$, let $\omega_j\triangleq -\eta^2 + 2\eta\langle z_j,\Pi_i a_t\rangle$ and $a^{(j)}_{t+1} = a_j - \eta\cdot z_j$. We know that \begin{align}
		\frac{\norm{\Pi_i a^{(j)}_{t+1}}^2_2}{\norm{\Pi_i a_t}^2_2} &= \frac{\norm{\Pi_i (a_t - \eta z_j)}^2_2}{\norm{a_t - \eta z_j}^2_2 \cdot \norm{\Pi_i a_t}^2_2} \nonumber\\
		&= \frac{\norm{\Pi_i a_t}^2 + \eta^2\norm{\Pi_i z_j}^2 - 2\eta\langle z_j,\Pi_i a_t\rangle}{(1 + \eta^2 - 2\eta\langle a_t,z_j\rangle)\cdot\norm{\Pi_i a_t}^2_2} \nonumber\\
		&= \frac{1 + \norm{\Pi_i a_t}^{-2}_2\cdot\left(\eta^2\norm{\Pi_i z_j}^2 - 2\eta\langle z_j,\Pi_i a_t\rangle\right)}{1 + \eta^2 - 2\eta\langle a_t,z_j\rangle} \nonumber\\
		&= \frac{1 + \norm{\Pi_i a_t}^{-2}_2\cdot\left(\eta^2(1 - \langle v_i,z_j\rangle^2) - 2\eta\langle z_j,\Pi_i a_t\rangle\right)}{1 + \eta^2 - 2\eta\langle z_j,\Pi_i a_t\rangle - 2\eta\langle a_t,v_i\rangle\langle z_j,v_i\rangle} \nonumber\\
		&= \frac{1 - \norm{\Pi_i a_t}^{-2}_2\left(\omega_j + \eta^2\langle v_i,z_j\rangle^2\right)}{1 - \omega_j - 2\eta\langle a_t,v_i\rangle\langle z_j,v_i\rangle}.\label{eq:omega}
	\end{align} Define the events 
	\begin{equation}A_j \triangleq \left\{\langle z_j,\Pi_{i^*} a_t\rangle \ge k^{-\ceeone}\norm{\Pi_{i^*} a_t}_2 \ \ \ \text{and} \ \ \ \langle z_j,v_{i^*}\rangle \le -k^{-\ceetwo}\right\}. \label{eq:event_progress}
	\end{equation} and 
	\begin{equation}
		B_j[i]\triangleq \left\{\left|\langle z_j,\Pi_ia_t\rangle\right| \le \xi k^{-\ceeone}\norm{\Pi_i a_t}_2 \ \ \ \text{and} \ \ \ |\langle z_j,v_i\rangle| \le \xi k^{-\ceetwo}\right\}\label{eq:event_not_progress}
	\end{equation}

	for some $\xi > 1$ to be specified later. We show in Claim~\ref{claim:probabilities} below that for any given $j$, there is some absolute constant $\nu > 0$ such that $\Pr[A_j]\ge \exp(-\nu\cdot k^{1-2c})$, and some absolute constant $\xi > 1$ such that $\Pr[B_j[i]]\ge 1 - \exp(-3\nu\cdot k^{1-2c})$.

	We will now argue that the event $A_j$ corresponds to making good progress, while the event $B_j[i]$ corresponds to not making too much progress.

	Suppose $A_j$ held for some $j = j^*$ and $B_j[i]$ held for all $i\in[k], j\in[M]$. If we took $\eta \triangleq k^{-\ceeone}\sigma^*$, we would conclude from the definition of $A_{j^*}$ and the assumption $0.9\sigma_t\le \sigma^* \le 1.1\sigma_t$ that 
	\begin{equation*}
	\omega_{j^*}\ge \norm{\Pi_{i^*} a_t}^2_2 k^{-2\ceeone}\cdot (-0.9^2 + 2\cdot 0.9) = 0.99\norm{\Pi_{i^*} a_t}^2_2 k^{-2\ceeone}.
	\end{equation*}
	Likewise from the definition of $B_j[i]$ we would conclude that 
	\begin{equation*}
		\omega_j\le \norm{\Pi_i a_t}^2_2 k^{-2\ceeone}\cdot (-1.21 + 2.2\xi).
	\end{equation*} 
	for all $j\in[M]$ and $i\in[k]$. So we get that 
	\begin{equation}
		\frac{\norm{\Pi_{i} a^{(j)}_{t+1}}^2_2}{\norm{\Pi_i a_t}^2_2}\ge \frac{1 - (2.2\xi-1.21)k^{-2\ceeone} - \left(1.21\xi^2 k^{-2\ceeone}\right)\cdot \left(k^{- 2\ceetwo}\xi^2\right)}{1 - (2.2\xi-1.21)k^{-2\ceeone}\norm{\Pi_i a_t}^2_2 + 2\cdot\left(0.9\xi k^{-\ceeone}\right)\cdot \left(k^{- \ceetwo}\right)\norm{\Pi_i a_t}_2 \langle a_t, v_i\rangle}\label{eq:ratio_lower}
	\end{equation}
	for every $j\in[M]$ and $i\in[k]$, and for $i = i^*$ and $j = j^*$ we additionally have that \begin{equation}\frac{\norm{\Pi_{i^*} a^{(j^*)}_{t+1}}^2_2}{\norm{\Pi_{i^*} a_t}^2_2} \le \frac{1 - 0.99 k^{-2\ceeone} - \left(0.81 k^{-2\ceeone}\right)\cdot\left(k^{- 2\ceetwo}\right)}{1 -0.99 k^{-2\ceeone}\norm{\Pi_{i^*} a_t}^2_2 + 2\cdot\left(1.1 k^{-\ceeone}\right)\cdot \left(k^{- \ceetwo}\right)\norm{\Pi_{i^*} a_t}_2 \langle a_t,v_{i^*}\rangle}.\label{eq:ratio_upper}
	\end{equation}
	In particular, we get that 
	\begin{align*}
		1 - \frac{\norm{\Pi_i a^{(j)}_{t+1}}^2_2}{\norm{\Pi_i a_t}^2_2}&\le \frac{1}{0.99}\left(1.21\xi^2 k^{-4c} + (2.2\xi - 1.21)k^{-2c}\langle a_t,v_i\rangle^2 + 1.8\xi k^{-2c}\norm{\Pi_i a_t}_2\langle a_t,v_i\rangle\right) \\
		&\triangleq \overline{g}(\langle a_t, v_i\rangle)
	\end{align*} 
	for every $j\in[M]$ and $i\in[k]$, and for $i = i^*, j = j^*$, we additionally have that 
	\begin{align*}
		1 - \frac{\norm{\Pi_{i^*} a^{(j^*)}_{t+1}}^2_2}{\norm{\Pi_{i^*} a_t}^2_2} &\ge \frac{1}{1.01}\left(0.81k^{-4c} + 0.99k^{-2\ceeone}\langle a_t,v_{i^*}\rangle^2 + 2.2 k^{-2c}\norm{\Pi_{i^*} a_t}_2\langle a_t,v_{i^*}\rangle\right) \\
		&\triangleq \underline{g}(\langle a_t,v_{i^*}\rangle)
	\end{align*} 
	where we have used the fact that the denominators of \eqref{eq:ratio_lower} and \eqref{eq:ratio_upper} are in $[0.99,1.01]$ for sufficiently large $k$. Also, we emphasize that these quantities can be expressed as functions $\overline{g}$ and $\underline{g}$ solely in $\langle a_t,v_i\rangle$ because for any $i\in[k]$, $\norm{\Pi_i a_t}^2_2 = 1 - \langle v_i,a_t\rangle^2$.

	To control these quantities, note that the function $\underline{g}(x)$ is increasing over the interval $[0,\tau^*]$ and decreasing over the interval $[\tau^*,1]$ for some constant $\tau^*\in[0.91,0.92]$. When $\langle v_{i^*}, a_t\rangle = 1$, we get that $\underline{g} = \Omega(k^{-2\ceeone})$, because the $0.99k^{-2c}\langle a_t,v_{i^*}\rangle^2$ term in the definition of $\underline{g}$ dominates. And when $\langle v_{i^*}, a_t\rangle = k^{-c}$, we get that $\underline{g} = \Omega(k^{-3c})$, because the $2.2k^{-2c}\norm{\Pi_{i^*}a_t}_2\langle a_t,v_{i^*}\rangle$ term in the definition of $\underline{g}$ dominates. So the first part of the lemma follows.

	On the other hand, there is some absolute constant $\beta' > 1$ such that $\underline{g}(x) \le \overline{g}(x)\le \beta'\underline{g}(x)$ for all $x\in[0,1]$. There is some constant $\beta'' > 1$ for which $\underline{g}(x)/\underline{g}(y) < \beta''$ for all $0\le x \le y \le 1$. The reason is that $\underline{g}(x)$ is increasing over the interval $[0,\tau^*]$ and decreasing over the interval $[\tau^*,1]$, and $\underline{g}(x) = 1 - \Omega(k^{-2c})$ for $x\in[\tau^*,1]$.

	It follows that for any $j\in[M], i\in[k]$ 
	\begin{equation*}
		1 - \frac{\norm{\Pi_i a^{(j)}_{t+1}}^2_2}{\norm{\Pi_i a_t}^2_2} \le \overline{g}(\langle a_t,v_i\rangle) \le \beta'\underline{g}(\langle a_t,v_i\rangle) \le \beta'\beta''\cdot \underline{g}(\langle a_t,v_{i^*}\rangle) \le \beta'\beta''\cdot\left(1 - \frac{\norm{\Pi_{i^*} a^{(j^*)}_{t+1}}^2_2}{\norm{\Pi_{i^*} a_t}^2_2}\right),
	\end{equation*} 
	so by taking $\beta$ in the statement of the lemma to be $\beta'\cdot \beta''$ and invoking the elementary inequality $1 - a\cdot x\le (1 - x)^a$ for $a > 1$, we get the second part of the lemma.

	We conclude that if event \eqref{eq:event_progress} held for some $j\in[M]$ and event \eqref{eq:event_not_progress} held for all $j\in[M]$ and $i\in[k]$, then parts 1 and 2 of Lemma~\ref{lem:hyperplanes_progress} would hold.

	Furthermore, 
	\begin{equation*}
		\Pr\left[\bigwedge_{j\in[M]}\overline{A}_j\right] \le \left(1 - \exp(-\nu\cdot k^{1-2c})\right)^M
	\end{equation*} 
	and 
	\begin{equation*}
		\Pr\left[\bigvee_{i\in[k],j\in[M]}\overline{B}_j[i]\right] \le kM\cdot \exp(-3\nu\cdot k^{1-2c}),
	\end{equation*} 
	so by taking $M = \exp(\nu\cdot k^{1-2c})\ln(2/\delta)$, we get that with probability at least $1 - \delta/2$, the event $A_j$ occurs for some $j\in[M]$. And with probability at least $1 - k\exp(-2\nu\cdot k^{1-2c})\ln(2/\delta) \ge 1 - \delta/2$, the event $B_j[i]$ occurs for all $i\in[k]$ and $j\in[M]$, where we have used the fact that $\delta > \exp(-\nu\cdot k^{1-2c})$.
\end{proof}

It remains to lower bound the probabilities of events $A_j$ and $B_j[i]$.

\begin{claim}
There are absolute constants $\nu > 0, \xi > 1$ such that for any $j\in[M],i\in[k]$, 
\begin{equation*}
\Pr[A_j]\ge \exp(-\nu\cdot k^{1-2c}) \ \ \ \text{and} \ \ \ \Pr[B_j[i]]\ge 1 - \exp(-3\nu\cdot k^{1-2c}).
\end{equation*}\label{claim:probabilities}
\end{claim}

\begin{proof}
	While \eqref{eq:event_progress} is defined with respect to $v_{i^*}$, the argument below holds for general $i$. Henceforth, fix an arbitrary $i\in[k]$ and $j\in[M]$. The key fact we will use is that the two quantities $\langle z,\Pi_i a_t\rangle$ and $\langle z,v_i\rangle$ are approximately independent (if $U$ consisted of the $k$ top singular vectors of $\vec{M}$ itself, these random variables would be exactly independent).

	Note that 
	\begin{align*}
		\rho\triangleq \left\langle \vec{U}\Pi_i a_t, \frac{\vec{U} v_j}{\norm{\vec{U} v_j}_2}\right\rangle &= \frac{1}{\norm{\vec{U} v_j}_2}\cdot \left[\langle \vec{U}a_t,\vec{U}v_j\rangle - \langle a_t,v_j\rangle \norm{\vec{U} v_j}^2_2\right] \\
		&\le \frac{1}{(1-\delsamp^2)^{1/2}}\cdot\left[(\langle a_t,v_j\rangle + \delsamp) - \langle a_t,v_j\rangle\cdot(1-\delsamp^2)\right] \\
		&\le \frac{\delsamp + \delsamp^2}{(1-\delsamp^2)^{1/2}} \le 2\delsamp.
	\end{align*} 
	where the second step follows from the second part of Lemma~\ref{lem:wedin_angle} and \eqref{eq:innerprodclose}. Likewise we have that 
	\begin{equation*}
		\rho \ge (\langle a_t,v_j\rangle - \delsamp) - \langle a_t, v_j\rangle \ge - \delsamp.
	\end{equation*} So we may write 
	\begin{equation}
		\vec{U}\Pi_i a_t = \rho\cdot\frac{\vec{U}v_j}{\norm{\vec{U}v_j}_2} + v^{\perp}\label{eq:decompUPia}
	\end{equation} 
	for $v^{\perp}$ lying in the row span of $\vec{U}$ and orthogonal to $\vec{U}v_j$, and satisfying 
	\begin{equation*}
	\norm{v^{\perp}}^2 \ge \norm{\vec{U}\Pi_i a_t}^2 - \rho^2 \ge \left(\norm{\Pi_i a_t}^2_2 - \delsamp^2 - 2\delsamp\right) - 4\delsamp^2 \ge \frac{1}{2}\norm{\Pi_i a_t}^2_2
	\end{equation*} 
	by Lemma~\ref{lem:hyperplanes_progress} and the assumption that $\norm{\Pi_i a_t} \ge 3\delsamp^{1/2}$. As $\langle g,\vec{U}v_j\rangle$ and $\langle g,v^{\perp}\rangle$ are \emph{independent} Gaussians with variances at least $1 - \delsamp^2\ge 1/2$ and $\frac{1}{2}\norm{\Pi_i a_t}^2_2$ respectively, by the same argument as in Corollary~\ref{cor:joint} we can show that $\left\langle \frac{g}{\norm{g}_2}, \vec{U}v_j\right\rangle \le -2\cdot k^{-c}$ and $\left\langle \frac{g}{\norm{g}_2},v^{\perp}\right\rangle \ge 2\cdot k^{-c}\cdot\norm{\Pi_i a_t}_2$ with probability at least $\exp(-\nu\cdot k^{1-2c})$ for some absolute constant $\nu>0$. By another application of Corollary~\ref{cor:unitcorr}, there is some absolute constant $\xi' > 0$ for which $\left|\left\langle \frac{g}{\norm{g}_2}, \vec{U}v_j\right\rangle\right| \le \xi'\cdot k^{-c}$ and $\left|\left\langle \frac{g}{\norm{g}_2},v^{\perp}\right\rangle\right| \le \xi\cdot k^{-c}\cdot\norm{\Pi_i a_t}_2$ with probability at least $1 -\exp(-3\nu\cdot k^{1-2c})$.

	If this is the case, then by \eqref{eq:decompUPia}, 
	\begin{equation*}
		\left \langle \frac{g}{\norm{g}_2},\vec{U}\Pi_i a_t\right \rangle \ge \frac{-4\delsamp\cdot k^{-c}}{(1 - \delsamp^2)^{1/2}} + 2\cdot k^{-c}\cdot \norm{\Pi_i a_t}_2 \ge \frac{3}{2} k^{-c}\cdot\norm{\Pi_i a_t}_2 \ge k^{-c}\cdot \norm{\Pi_i a_t}_2,
	\end{equation*} 
	where we have used that $\norm{\Pi_i a_t}\ge 3\delsamp^{1/2} \ge 4\cdot \frac{4\delsamp}{(1 - \delsamp^2)^{1/2}}$ for any $\delsamp$ smaller than some absolute constant. We also get that 
	\begin{equation*}
		\left|\left \langle \frac{g}{\norm{g}_2},\vec{U}\Pi_i a_t\right \rangle\right| \le 2\delsamp \cdot \xi'\cdot k^{-c} + \xi'\cdot k^{-c}\cdot \norm{\Pi_i a_t}_2 \le 19\xi' \cdot k^{-c}\cdot\norm{\Pi_i a_t}_2,
	\end{equation*} where we have used that $\norm{\Pi_i a_t}\ge \norm{\Pi_i a_t}^2 \ge 9\delsamp$.

	Noting that $\norm{g\vec{U}}_2 = \norm{g}$ by orthonormality of the columns of $\vec{U}$ so that 
	\begin{equation*}
		\left\langle \frac{g}{\norm{g}_2}, \vec{U}v_j\right\rangle = \left\langle z, v_j\right\rangle \ \ \ \text{and} \ \ \ \left\langle \frac{g}{\norm{g}_2}, \vec{U}\Pi_i a_t\right\rangle = \langle z, \Pi_i a_t\rangle,
	\end{equation*} 
	we conclude that with probability at least $\exp(-\nu\cdot k^{1-2c})$, both events in \eqref{eq:event_progress} hold, and likewise with probability at least $1 - \exp(-3\nu\cdot k^{1-2c})$, both events in \eqref{eq:event_not_progress} hold for $\xi = \xi'/19$.
\end{proof}

\subsection{Algorithm Specification-- Single Component}
\label{subsec:single_hyperplane}

We are now ready to describe our algorithm \textsc{HyperplaneMomentDescent} for learning a single components of $\calD$. The key subroutines are: 
\begin{itemize}
	\item \textsc{HyperplaneMomentDescent} (Algorithm~\ref{alg:optimistic_moment_descent}): the pseudocode for this is very similar to that of \textsc{FourierMomentDescent}, the key differences being 1) the matrix on which we run \textsc{ApproxBlockSVD}, 2) the definition of $\calF_t$, 3) the fact that we maintain that $a_t$ are unit vectors, 4) the parameters which are tuned towards detecting $1 - \Omega(k^{-3/5})$ multiplicative progress instead of $1 - \Omega(k^{-1/2})$, and most importantly, 5) the outer loop over $i\in[S]$ which tries many random initializations, runs a full $T$ rounds of moment descent on each of them, and checks whether the final estimate in any of these runs is close to a component of $\calD$.
	\item \textsc{CheckOutcomeHyperplanes} (Algorithm~\ref{alg:check_outcome_low_noise}): \textsc{CheckOutcome} is used to check whether a given estimate is close to any component of $\calD$.
\end{itemize}

\begin{algorithm}\caption{\textsc{HyperplaneMomentDescent}, Lemma~\ref{lem:hyperplanecorrect}}\label{alg:hyperplane_moment_descent}
\begin{algorithmic}[1]
	\State \textbf{Input}: Sample access to mixture of hyperplanes $\calD$, failure probability $\delta$, error $\epsilon$
	\State \textbf{Output}: $v^*\in\S^{d-1}$ satisfying $\min_{i\in[k]}\norm{\Pi_i v^*}_2 \le\epsilon$, with probability at least $1 - \delta$.
		\State Set $\epsilon' = \epsilon\cdot\pmin/2$.
		\State Set $S = \exp(-\Omega(k^{1-2c}))\cdot \ln(2/\delta)$.
		\State Set $T = \Omega(k^{3/5}\cdot\ln(\mu/\epsilon'))$.
		\State Set $\delsamp = 1/\poly(k)$ sufficiently small.
		\State Set $\delta' = \frac{\delta}{50T}$.
		\State Set $M = e^{-\nu\cdot k^{1 - 2c}}\ln(2/\delta')$.
		\State Set $\delta'' = \frac{\delta}{50MT}$.
		\State Set $\delta^* = \frac{\delta}{2S}$
		\State Set $\underline{\sigma} = 2\cdot(\epsilon'/2)^{\beta}$, where $\beta>1$ is the constant from Lemma~\ref{lem:hyperplanes_progress}.
		\State Set $\overline{\sigma} = 4$.
		\State Set $N_1 \triangleq \Omega\left(d\cdot \ln(1/\delta')\cdot \pmin^{-2}\cdot\delsamp^{-2}\right)$
		\State \multiline{Draw $N_1$ i.i.d. samples $\{x_i\}_{i\in[N_1]}$ from $\calD$ and form the matrix $\widehat{\vec{M}}^{(N_1)}\triangleq \Id - \frac{1}{N_1}\sum^{N_1}_{i=1}x_i x_i^{\top}$.}
		\State Let $\vec{U}\in\R^{d\times k}$ be the output of \textsc{ApproxBlockSVD}$(\widehat{\vec{M}}^{(N_1)}, \delsamp/2, 1/50)$
		\For{$0\le i < S$}
			\State Sample $g\sim\N(0,\Id_k)$ and let $a_0 = \frac{g \vec{U}}{\norm{g \vec{U}}_2}$.
			\For{$0\le t < T$}
				\State \multiline{Let $\calF_t$ be the univariate mixture of Gaussians which can be sampled from by drawing $x\sim\calD$ and computing $\langle a_t,x\rangle$.}
				\State Let $p = 20\ln\left(\frac{3}{2\pmin}\right)$.
				\State Let $\kappa = \Theta(k^{-3/5})$ as in the proof of Lemma~\ref{lem:hyperplanecorrect}.
				\State $\underline{\sigma}^{\text{sharp}}_t \triangleq$~\textsc{EstimateMinVariance}($\calF_t,\overline{\sigma}, \underline{\sigma},p,\delta'$). \Comment{Algorithm~\ref{alg:estimate_min_variance}}
				\For{$j\in[M]$}
					\State Sample $g^{(j)}_t\sim\N(0,\Id_k)$ and define $v^{(j)}_t = \frac{g^{(j)}_t\vec{U}}{\norm{g^{(j)}_t\vec{U}}_2}\in\S^{d-1}$.
					\State Let $a'^{(j)}_t = \frac{a_t - \eta_t v^{(j)}_t}{\norm{a_t - \eta_t v^{(j)}_t}_2}$ for $\eta_t \triangleq k^{-1/5}\cdot\underline{\sigma}^{\text{sharp}}_t$.
					\State \multiline{Let $\calF'^{(j)}_t$ be the univariate mixture of Gaussians which can be sampled from by drawing $x\sim\calD$ and computing $\langle a_t, x\rangle$.}
					\If{\textsc{CompareMinVariances}($\calF_t,\calF'^{(j)}_t,\overline{\sigma},\underline{\sigma},\kappa,\delta''$) = $\mathsf{true}$} \Comment{Algorithm~\ref{alg:compare_min_variances}, Corollary~\ref{cor:compare_min_variances}}
						\State Set $a_{t+1} = a'^{(j)}_t$
						\Break
					\EndIf
				\EndFor
			\EndFor
			\If{\textsc{CheckOutcomeHyperplanes}($\calD,a_t,\epsilon',\delta^*$) = $\mathsf{true}$} \Comment{Algorithm~\ref{alg:check_outcome_hyperplanes}}
				\State Set $v^* = a_t$ and return $v^*$.
			\EndIf
		\EndFor
\end{algorithmic}
\end{algorithm}

\begin{algorithm}\caption{\textsc{CheckOutcomeHyperplanes}$({\cal D},v,\epsilon,\delta)$, Lemma~\ref{lem:check_outcome_hyperplanes}}\label{alg:check_outcome_hyperplanes}
\begin{algorithmic}[1]
	\State \textbf{Input}: Sample access to mixture of hyperplanes $\calD$, direction $v\in\S^{d-1}$, threshold $\epsilon > 0$, failure probability $\delta$
	\State \textbf{Output}: $\mathsf{true}$ if $\min_{i\in[k]}\norm{\Pi_i v}_2 \le \epsilon$, $\mathsf{false}$ if $\norm{\Pi_i v}_2 \ge 2\epsilon/\pmin$, with probability at least $1 - \delta$
			\State Let $\calF$ be the univariate mixture of Gaussians which can be sampled from by drawing $x\sim\calD$ and computing $\langle v,x\rangle$.
			\State Draw $N_2 \triangleq O\left(\ln(1/\delta)\pmin^{-2}\right)$ samples from $\calF$.
			\If{$\ge \frac{4\pmin}{15}\cdot N_2$ samples lie in $[-\epsilon,\epsilon]$}
				\State Return $\mathsf{true}$
			\Else
				\State Return $\mathsf{false}$
			\EndIf
\end{algorithmic}
\end{algorithm}

\subsection{Proof of Correctness}
\label{subsec:hyperplanes_correct}

We first give a proof of correctness for \textsc{CheckOutcomeHyperplanes}.

\begin{lemma}\label{lem:check_outcome_hyperplanes}
	Let $v\in\S^{d-1}$ and $\calD$ be a mixture of hyperplanes, and let $\epsilon > 0$. If there is some component $v_{i^*}$ for which $\norm{v - v_{i^*}}_2 \in [-\epsilon,\epsilon]$, then \textsc{CheckOutcomeHyperplanes}($\calD, v,\epsilon,\delta$) (Algorithm~\ref{alg:check_outcome_hyperplanes}) returns $\mathsf{true}$ with probability at least $1 - \delta$. Otherwise, if $\norm{v - v_i}_2 > 2\epsilon/\pmin$ for all $i\in[k]$, then \textsc{CheckOutcomeHyperplanes} returns $\mathsf{false}$ with probability at least $1 - \delta$.
\end{lemma}

\begin{proof}
	First suppose there is some $i^*\in[k]$ for which $\norm{v - v_{i^*}}_2\le \epsilon$. Then $\calF$ is a mixture of Gaussians with one of its components having variance at most $\epsilon^2$. So for $x\sim\calF$, we get that 
	\begin{equation}
	\Pr[|x|\le\epsilon/2] \ge \pmin\cdot \int^{\epsilon/2}_{-\epsilon/2}e^{-\frac{x^2}{2\epsilon^2}}\, \d x \ge \pmin/3.
	\end{equation} 
	On the other hand if we had that $\norm{v - v_i}_2 > 2\epsilon/\pmin$ for all $i\in[k]$, then $\calF$ is a mixture of Gaussians whose components have variances exceeding $\frac{4\epsilon^2}{\pmin^2}$. So for $x\sim\calF$, we would get that
	\begin{equation*}
		\Pr[|x|\le\epsilon/2] \le \sum^k_{i=1}p_i \cdot \frac{0.8\cdot \epsilon/2}{2\epsilon/\pmin} = \pmin/5,
	\end{equation*} where in the first step we have used the fact that $\int^{\tau}_{-\tau}e^{-\frac{x^2}{2\sigma^2}} \, \d x \le \sqrt{2/\pi}\cdot (\tau/\sigma)\le 0.8\tau/\sigma$ for any $\tau,\sigma>0$.

	We need to take enough samples for our empirical estimate of $\Pr[|x|\le \epsilon]$ to be $\pmin/15$-additively close to the true value with probability at least $1 - \delta$, for which it suffices to take $O(\ln(1/\delta)\pmin^{-2})$ samples.
\end{proof}

We can now prove correctness of \textsc{HyperplaneMomentDescent}.

\begin{lemma}
	Let $\calD$ be a mixture of hyperplanes with mixing weights $\{p_i\}$ and directions $\{v_i\}$. With probability at least $1 - \delta$, \textsc{HyperplaneMomentDescent}($\calD, \delta,\epsilon$) (Algorithm~\ref{alg:hyperplane_moment_descent}) outputs direction $a_T\in\S^{d-1}$ for which $(\epsilon/2)^C \le \min_{i\in[k]}\norm{w_i - a_T}_2 \le \epsilon$ for some absolute constant $C > 0$.\label{lem:hyperplanecorrect}
\end{lemma}

\begin{proof}
	Henceforth, take $c$ in Lemma~\ref{lem:hyperplanes_progress} to be $c = 1/5$. Let $\sigma_t\triangleq \min_{i\in[k]}\norm{w_i - a_t}_2$. Naively we have that $\sigma_t \le 2$.

	By a simple union bound, we first upper bound the probability that the steps of moment descent in the $i$-th iteration of the outer loop all succeed.

	\begin{claim}
		Let $i\in[S]$. With probability at least $9/10$, the randomized components of the inner loop (over $t$) of the $i$-th iteration of the outer loop of \textsc{HyperplaneMomentDescent} all succeed.\label{claim:allsuccess_hyper}
	\end{claim}

	\begin{proof}
		Each $t$-th iteration of the second loop in \textsc{HyperplaneMomentDescent} has the following randomized components: 1) empirically estimating $\vec{M}$, 2) running \textsc{ApproxBlockSVD} on this empirical estimate, 3) running \textsc{EstimateMinVariance}, 4) trying the Gaussian vectors $g$ in the innermost loop over $j\in[M]$, and 5) running \textsc{CompareMinVariances} in this innermost loop.

		Because the failure probability $\delta'$ for 1), 3), 4) were chosen to be $\frac{1}{50T}$, the failure probability $\delta''$ for 5) was chosen to be $\frac{1}{50MT}$, and the failure probability for 2) was chosen to be $1/50$, we can bound by $1/10$ the overall failure probability of these tasks in a single $i$-th iteration of the outer loop of \textsc{HyperplaneMomentDescent}.
	\end{proof}

	Call the event in Claim~\ref{claim:allsuccess_hyper} $\calE_i$. Next, we show that provided $\calE_i$ occurs and the initial point $a_0$ for the $i$-th iteration of the outer loop is close to some $v_j$, then we can bound the extent to which every step of the subsequent inner loop (over $t$) contracts $\sigma^2_t$.

	\begin{claim}
		Let $i\in[S]$ and condition on $\calE_i$. If in the $i$-th iteration of \textsc{HyperplaneMomentDescent}, $|\langle a_0,v_j\rangle|\ge k^{-c}$, then for each $0 \le t < T$: \begin{enumerate}
			\item (Completeness) There exists some $j\in[M]$ for which 
			\begin{align*}
			\textsc{CompareMinVariances}(\calF_t,\calF'^{(j)}_t,\overline{\sigma},\underline{\sigma},\kappa,2\kappa,\delta'')
			\end{align*}
			outputs $\mathsf{true}$ for some $\kappa = \Theta(k^{-3c})$.
			\item (Soundness) For any such $j\in[M]$ for which \textsc{CompareMinVariances} outputs $\mathsf{true}$, 
			\begin{equation}
				\left(1 - \frac{\beta}{k^{3c}}\right)\sigma^2_t \le \sigma^2_{t+1} \le \left(1 - \frac{\underline{\alpha}}{k^{3c}}\right)\sigma^2_t.\label{eq:bounds_iter_hyperplanes}
			\end{equation} 
			for some $\underline{\alpha} < \alpha$, where $\alpha,\beta$ are the constants in Lemma~\ref{lem:hyperplanes_progress}.
		\end{enumerate}\label{claim:hyper_complete_sound}
	\end{claim}

	\begin{proof}
		Suppose inductively that $|\langle a_t, v_j\rangle|\ge k^{-c}$ for some $j\in[k]$. By the first part of Lemma~\ref{lem:hyperplanes_progress}, there exists some $j\in[M]$ for which $\sigma^{(j)}_t \triangleq \min_{i\in[k]}\norm{w_i - a'^{(j)}_t}_2$ satisfies $(\sigma^{(j)}_t)^2 \le \left(1 - \frac{\alpha}{k^{3c}}\right)(\sigma_t)^2$, and because $1 - \frac{\alpha}{k^{3c}} \le \left(\frac{1}{1 + 2\kappa}\right)^2$ for some $\kappa = \Theta(k^{-3c})$, 
		\begin{align*}
		\textsc{CompareMinVariances}(\calF_t,\calF'^{(j)}_t,\overline{\sigma},\underline{\sigma},\kappa,2\kappa,\delta'')
		\end{align*} 
		would return $\mathsf{true}$, completing the proof of completeness.

		For soundness, note that for any such $j$, by Corollary~\ref{cor:compare_min_variances} we know that 
		\begin{equation*}
			(\sigma^{(j)}_t)^2 \le (1 + \kappa)^{-2}\cdot \sigma^2_t \le (1 - \kappa/2)\cdot \sigma^2_t \le \left(1 - \frac{\underline{\alpha}}{k^{3c}}\right)\cdot\sigma^2_t,
		\end{equation*} 
		which gives the upper bound in \eqref{eq:bounds_iter_hyperplanes}. The lower bound follows from the second part of Lemma~\ref{lem:hyperplanes_progress}. This completes the proof of soundness as well as the inductive step, as the upper bound of \eqref{eq:bounds_iter_hyperplanes} implies that $\max_{j\in[k]}|\langle a_{t+1}, v_j\rangle| \ge \max_{j\in[k]}|\langle a_{t}, v_j\rangle| \ge k^{-c}$.
	\end{proof}

	Lastly, we lower bound the probability that in the $i$-th iteration of the outer loop, the randomly chosen initial point $a_0$ is sufficiently close to some $v_j$.

	\begin{claim}
		Let $i\in[S]$. With probability at least $\exp(-\Omega(k^{1-2c}))$, the following holds. In the $i$-th iteration of the outer loop of \textsc{HyperplaneMomentDescent}, $|\langle a_0, v_j\rangle| \ge k^{-c}$ for some $j\in[k]$, where $a_0$ is the initial iterate in the inner loop over $t$.\label{claim:findhyper_init}
	\end{claim}

	\begin{proof}
		We know by Corollary~\ref{cor:unitcorr} that for $g\sim\N(0,\Id_k)$ and $a_0 \triangleq \frac{g\vec{U}}{\norm{g\vec{U}}_2}$, for any $j\in[k]$ we have that $\Pr_g[|\langle a_0,v_j\rangle| \ge k^{-c}] \ge \exp(-\Omega(k^{1-2c}))$.
	\end{proof}

	We are ready to complete the proof of Lemma~\ref{lem:hyperplanecorrect}. If we take $T = k^{3c}\ln(8\epsilon^{-2}\pmin^{-2})/\underline{\alpha}$, then in an iteration $i\in[S]$ for which $\calE_i$ holds and $|\langle a_0, v\rangle|\ge k^{-c}$, by Claim~\ref{claim:hyper_complete_sound} we are guaranteed that 
	\begin{equation}
	2/(\epsilon^2\pmin^2/8)^C\le \sigma^2_T\le \epsilon^2\cdot\pmin^2/4\label{eq:alltoplan_hyper}
	\end{equation} for some absolute constant $C> 0$. We remark that the lower bound on $\sigma^2_T$ ensures that throughout the course of \textsc{HyperplaneMomentDescent}, the parameter $\underline{\sigma}$ passed to \textsc{CompareMinVariances} is a valid lower bound on $\sigma_t$ for all $0\le t\le T$.

	Now for $i\in[S]$, let $A_i$ be the event that the inner loop breaks out with a direction $a$ for which $\norm{\Pi_j a}_2 \le \epsilon\cdot\pmin/2$. Also, let $B_i$ be the event that \textsc{CheckOutcomeHyperplanes} runs successfully. Note that $A_i$ and $B_i$ are independent. By Claim~\ref{claim:allsuccess_hyper}, Claim~\ref{claim:findhyper_init}, and \eqref{eq:alltoplan_hyper}, we know $\Pr[A_i]\ge \frac{9}{10}\exp(-\Omega(k^{1-2c}))\triangleq q$. By Lemma~\ref{lem:check_outcome_hyperplanes}, we know $\Pr[B_i]\ge 1 - \delta^*$.

	The probability that $A_i$ occurs for at least one $i\in[S]$ is at least $1 - (1 - q)^S \ge 1 - e^{-qS}$, while the probability that $B_i$ holds for all $i$ is at least $1 - S\delta^*$. By taking $S = q^{-1}\ln(2/\delta) = \exp(-\Omega(k^{1-2c}))\cdot \ln(2/\delta)$ and $\delta^* = \frac{\delta}{2S}$, we conclude that the output of \textsc{HyperplaneMomentDescent} is some $v^*$ for which $\min_{i\in[k]}\norm{\Pi_i v^*}_2 \le \epsilon$.
\end{proof}

The analysis for the runtime and sample complexity of \textsc{HyperplaneMomentDescent} is essentially the same as that of \textsc{FourierMomentDescent}:

\begin{lemma}[Running time of \textsc{HyperplaneMomentDescent}]
Let
\begin{align*} 
N_1 &= \frac{d k^2 \ln(1 / \delta)}{\epsilon^2 \pmin^2} \\
N &= \pmin^{-4} k \ln ( 1/ \delta ) \cdot \poly \left( k^{3/5}, \ln ( 1 / \pmin ), \ln ( 1 / \epsilon ) \right)^{O\left(k^{3/5} \ln (1 / \pmin)\right)} \\
N_2 &= O\left(\pmin^{-2} k^{3/5}\ln(1/\delta)\right) \\
S &= \exp(\Omega(k^{3/5})\ln(1/\delta)
\; .
\end{align*}
Then \textsc{HyperplaneMomentDescent} (Algorithm~\ref{alg:fourier_moment_descent}) requires sample complexity 
\begin{equation*}
\widetilde{O}\left(N_1 + S\cdot(k^{3/5}e^{k^{3/5}}N + N_2)\right)\end{equation*} 
and runs in time 
\begin{equation*}
\widetilde{O}\left(dN_1 + S\cdot(k^{3/5}e^{k^{3/5}}N + N_2)\right).
\end{equation*}
\label{lem:runtime_hyper}
\end{lemma}

\subsection{Boosting for Mixtures of Hyperplanes}
\label{subsec:hyperplane_boost}

As with \textsc{FourierMomentDescent}, \textsc{HyperplaneMomentDescent} cannot be used on its own to obtain an arbitrarily good estimate for a component of the mixture, as the runtime and sample complexity of the primitives used for estimating minimum variance increase rapidly as the minimum variance of the univariate projections decreases. So at some point we need to switch over to a boosting algorithm.

In this section, we describe how to regard mixtures of hyperplanes as mixtures of non-spherical but fairly well-conditioned linear regressions. With this in place, we can then run either the boosting algorithm of \cite{li2018learning} or the one introduced in this work (see Section~\ref{sec:boosting}), all of which can tolerate the condition numbers of such mixtures.

Let $w\in\S^{d-1}$ be some direction and let $\Pi_w$ denote the projection to the orthogonal complement of $w$. Given $x\sim\calD$, we may regard this as a sample from a mixture of linear regressions as follows. Consider the tuple $(\Pi_w x, \langle x,w\rangle)$. By identifying $\Pi_w$ with $\R^{d-1}$, we may regard $\Pi_w x$ as a vector in $\R^{d-1}$ and $\langle x,w\rangle$ as the response.

Concretely, up to a change of basis we can assume without loss of generality that $w = (0,...,0,1)$, in which case $\Pi_w x$ is simply identified with the first $d - 1$ coordinates of $x$, and the response $\langle x,w\rangle$ is simply the last coordinate of $x$. Then the covariance matrix of the hyperplane orthogonal to $v_j$ is merely the upper $(d-1)\times (d-1)$ submatrix of $\Pi_j$, and because any $x$ sampled from that hyperplane satisfies $\langle v_j,x\rangle = 0$, we have that 
\begin{equation*}
	x_d = \left\langle -\frac{(v_j)_{1:d-1}}{(v_j)_d},x_{1:d-1}\right\rangle,
\end{equation*} where we use the notation of Section~\ref{subsec:notations}. For simplicity, denote $(v_j)_{1:d-1}$ by $v'_j$, $(v_j)_d$ by $a_j$. We may further assume without loss of generality that $a_j = \langle v_j, w\rangle$ is nonnegative for every $j$, as the directions $\{v_j\}$ for a mixture of hyperplanes are only specified up to sign.

Altogether, this yields the following basic claim.

\begin{lemma}
	Given a mixture $\calD$ of hyperplanes with mixing weights $\{p_j\}$ and directions $\{v_j\}$, let $\calD'$ be the mixture of linear regressions with mixing weights $\{p_j\}$, components $\{\N(0,\Id - v'_jv'^{\top}_j)\}$, and regressors $\{-v'_j/a_j\}$. Then $\calD$ and $\calD'$ are identical as distributions over $\R^d$.
\end{lemma}

We will choose $w$ randomly by sampling $g\sim\N(0,\Id_k)$ and defining $w = \frac{g\vec{U}}{\norm{g\vec{U}}_2}$. We need a basic estimate on the condition number of the covariances $\Id - v'_j v'^{\top}_j$ for a typical such $w$, keeping in mind that $v'_j$ is defined with respect to an orthonormal basis under which $v'_j$ is the $d$-th standard basis vector.

\begin{lemma}
	For $g\sim\N(0,\Id_k)$ and $w = \frac{g\vec{U}}{\norm{g\vec{U}}_2}$, the eigenvalues of $\Id - v'_j v'^{\top}_j$ lie in $\left[\Omega(1/k^3),1\right]$ for all $j\in[k]$ with probability at least $4/5$.\label{lem:randomspan_corr}
\end{lemma}

\begin{proof}
	Let $\tilde{v}_j\triangleq \frac{v'_j}{\norm{v'_j}_2}\in\S^{d-2}$. Note that 
	\begin{equation*}
		\Id - v'_j v'^{\top}_j = \Id - \tilde{v}_j \tilde{v}_j^{\top}\cdot \| v'_j \|_2^2 = \Id - \tilde{v}_j\tilde{v}_j^{\top}\cdot(1 - \langle w,v_j\rangle^2),
	\end{equation*} 
	so the eigenvalues of $\Id - v'_j v'^{\top}_j$ are 1 with multiplicity $d - 2$ and $\langle w,v_j\rangle^2$ with multiplicity 1. Fact~\ref{fact:unionbound} below allows us to conclude that with probability at least $4/5$, $\langle w,v_j\rangle^2 \ge \Omega(1/k^3)$ for all $j\in[k]$.
\end{proof}

\begin{fact}
	There is some constant $a_{\mathrm{anti}} > 0$ such that for the random vector $w$ defined in Lemma~\ref{lem:randomspan_corr}, 
	\begin{equation}
		\Pr\left[\langle w,v_j\rangle^2 \ge \frac{a_{\mathrm{anti}}}{k^3} \ \forall \ j\in[k]\right] \ge 4/5. 
	\end{equation}\label{fact:unionbound}
\end{fact}

\begin{proof}
	For any $j\in[k]$, we have that 
	\begin{equation*}
	\langle w,v_j \rangle = \frac{1}{\norm{g\vec{U}}_2} \langle g,\vec{U} v_j\rangle = \frac{1}{\norm{g}_2}\left\langle g,\vec{U}v_j\right\rangle.
	\end{equation*} 
	By the second part of Lemma~\ref{lem:wedin_angle}, $\norm{\vec{U}v_j}_2 \ge (1 - \delsamp^2)^{1/2}\ge 1/2$, and by Fact~\ref{fact:thinshell}, $\norm{g\vec{U}}_2 \le 1.1\sqrt{k}$ with probability at least $1 - e^{-c_{\text{shell}}d/100}$. $\langle g, \vec{U} v_j\rangle$ is distributed as a zero-mean Gaussian with variance at least $1 - \delsamp^2$, so for any $\tau > 0$, with probability at least $1 - \frac{\tau}{(1 - \delsamp^2)^{1/2}}\ge 1 - 2\tau$ we have that $\langle g, \vec{U} v_j\rangle^2 \ge \tau^2$. The proof is completed by taking $\tau = \frac{1}{10k}$ and $a_{\mathrm{anti}} = 1/121$.
\end{proof}

We can now invoke the boosting result of \cite{li2018learning} stated in Theorem~\ref{thm:liliang}.


\begin{corollary}
		Let $\calD$ be a mixture of hyperplanes in $\R^d$ with directions $\{v_j\}$, minimum mixing weight $\pmin$, and separation $\Delta$. There exist constants $a_{\mathrm{sep}}, a_{\mathrm{eig}}>0$ for which the following holds.

		Let $\zeta \triangleq a_{\mathrm{sep}}\Delta\cdot k\cdot\min\left(\frac{a_{\mathrm{eig}}}{k^3},\frac{\pmin}{64}\right)$. There is an algorithm $(\calD,v,\epsilon,\delta)$ which, given any $\epsilon > 0$, $\delta>0$, and $v\in\R^d$ for which there exists $j\in[k]$ with $\norm{w_j - v}_2 \le \frac{\zeta}{a_{\mathrm{eig}}k^3}$, draws $T\cdot M$ samples from $\calD$ for
	\begin{equation*}
		T = O\left( \pmin^{-2} d \ln(\zeta/\epsilon) \right) \ \ \ \text{and} \ \ \ M = \poly\left( \Delta^{-1} , \pmin^{-1} , k, \log T\right)\cdot\ln(1/\delta),
	\end{equation*} 
	runs in time $T\cdot M\cdot d$, and outputs $\tilde{v}\in\R^d$ for which either $\norm{v_j - \tilde{v}}_2 \le \epsilon$ or $\norm{v_j + \tilde{v}}_2 \le \epsilon$ with probability at least $1 - \delta$.
	\label{cor:hyperplaneboost}
\end{corollary}

\begin{proof}
	By the same argument as in Fact~\ref{fact:unionbound}, we know there exists some $a'_{\mathrm{anti}} > 0$ for which $\Pr\left[|\langle v,w\rangle|\ge a'_{\mathrm{anti}}/k^2\right]\ge 1/40$. For every $i\neq j$, we know there exists some $a_{\mathrm{conc}}>0$ for which 
	\begin{align*}
	\Pr\left[|\langle v_i - v_j,w\rangle|\le \frac{a_{\mathrm{conc}}}{\sqrt{k}}\norm{v_i - v_j}_2\right] \ge 1 - \frac{1}{40k^2}.
	\end{align*}

	By a union bound over the former event, the latter event for every $i\neq j$, and the event in Fact~\ref{fact:unionbound}, the probability all of these events happen is at least $3/4$. Condition on these events.

	Given $v\in\S^{d-1}$ satisfying $\norm{v - v_j}_2 \le \delta$, and $w = \frac{g\vec{U}}{\norm{g\vec{U}}_2}$ for $g\sim\N(0,\Id_k)$, note that $u \triangleq -v'/\langle v,w\rangle \in\R^{d-1}$ satisfies \begin{align*}\norm{u - v'_j}_2 &\le \norm{-\frac{v'}{\langle v,w\rangle} + \frac{v'_j}{\langle v,w\rangle}}_2 + \norm{-\frac{v'_j}{\langle v,w\rangle} + \frac{v'_j}{\langle v_j,w\rangle}} \\
	&= \frac{\delta}{|\langle v,w\rangle|} + \norm{v'_j}_2 \cdot\left|\frac{1}{\langle v,w\rangle} - \frac{1}{\langle v_j,w\rangle}\right| \\
	&\le \frac{\delta}{|\langle v,w\rangle|} + \frac{|\langle v - v_j,w\rangle|}{|\langle v,w\rangle| \cdot |\langle v_j,w\rangle|} \\
	&\le \frac{\delta}{|\langle v,w\rangle|} + \frac{\norm{v - v_j}_2 \cdot \norm{w}_2}{|\langle v,w\rangle| \cdot |\langle v_j,w\rangle|} \\
	&\le O(\delta \cdot k^2),
	\end{align*} where in the first step we use the triangle inequality, in the fourth step we use Cauchy-Schwarz, and in the fifth step we use the events we conditioned on. In other words, when $\calD$ is regarded as a mixture of linear regressions $\calD'$ under the direction $w$, $v'$ is a warm start close to $v'_j$.

	Next, we check that this mixture of linear regressions $\calD'$ is well-separated. For any $i\neq j$, let $v''_i,v''_j\in\R^d$ be the vectors $(-v'_i,\langle v_i,w\rangle)$ and $(-v'_j,\langle v_j,w\rangle)$ respectively. Then 
	\begin{align}
	\norm{\frac{v'_i}{\langle v_i, w\rangle} - \frac{v'_j}{\langle v_j,w\rangle}}^2_2 &= \norm{\frac{v''_i}{\langle v_i,w\rangle} - \frac{v''_j}{\langle v_j,w\rangle}}^2_2  \nonumber\\
	&= \frac{\| v''_i \|^2_2}{\langle v_i,w\rangle^2} + \frac{\|v''_j\|^2_2}{\langle v_j,w\rangle^2} - \frac{2\langle v''_i,v''_j\rangle}{\langle v_i,w\rangle \langle v_j,w\rangle} \nonumber\\
	&= \frac{\norm{v_i}^2_2}{\langle v_i,w\rangle^2} + \frac{\norm{v_j}^2_2}{\langle v_j,w\rangle^2} - \frac{2\langle v_i,v_j\rangle}{\langle v_i,w\rangle \langle v_j,w\rangle} \nonumber\\
	&= \frac{1}{\langle v_i,w\rangle^2} + \frac{1}{\langle v_j,w\rangle^2} - \frac{2 - \norm{v_i - v_j}^2_2}{\langle v_i,w\rangle \langle v_j,w\rangle}\nonumber\\ 
	&\ge \left(\frac{1}{\langle v_i,w\rangle} - \frac{1}{\langle v_j,w\rangle}\right)^2 + \frac{\norm{v_i - v_j}^2_2}{\langle v_i,w\rangle\langle v_j,w\rangle},\label{eq:scaledsep}
	\end{align} where in the third step we used the fact that $v_i$ and $v_j$ are the same as $v''_i$ and $v''_j$ up to a change of basis and a change of sign of the entry corresponding to the $w$ direction. Recall that we are assuming without loss of generality that $\langle v_i,w\rangle \ge 0$ for all $i\in[k]$, so \eqref{eq:scaledsep} is at least $\frac{\norm{v_i - v_j}^2_2}{\langle v_i,w\rangle\langle v_j,w\rangle} \ge \Omega(\Delta^2\cdot k^2)$.

	Lastly, by Lemma~\ref{lem:randomspan_corr}, we have that the covariances of the components of this mixture of linear regressions $\calD'$ have eigenvalues all lying in $[\Omega(1/k^3),1]$. So that the scaling is consistent with Theorem~\ref{thm:liliang}, consider the mixture of linear regressions $\tilde{\calD}$ from which one can sample by drawing $(x,y)$ from $\calD'$ and taking $(x\cdot\Theta(k^3),y\cdot\Theta(k^3))$. $\calD'$ has the same regressors as $\calD'$ and thus the same separation $\Omega(\Delta\cdot k)$, but its components' covariances have eigenvalues all lying in $[1,\Theta(k^3)]$. By Theorem~\ref{thm:liliang}, if we take $\zeta = \Omega(\Delta\cdot k)\cdot\min\left(\frac{1}{\Theta(k^3)}, \frac{\pmin}{64}\right)$, then the algorithm of \cite{li2018learning} converges to an $\epsilon$-close estimate for $v'_j$ provided $\norm{v' - v'_j}_2 \le \zeta/\Theta(k^3)$.
\end{proof}

\subsection{Learning All Hyperplanes}
\label{subsec:learnall_hyperplanes}

With \textsc{HyperplaneMomentDescent} and \textsc{HyperplaneBoost} in hand, it is now straightforward to obtain an algorithm that learns all components of a mixture of hyperplanes.

\begin{algorithm}\caption{\textsc{LearnHyperplanes}($\calD,\delta,\epsilon$)}\label{alg:learnhyperplanes}
\begin{algorithmic}[1]
	\State \textbf{Input}: Sample access to mixture of hyperplanes $\calD$ with separation $\Delta$ and directions $\{v_i\}$, failure probability $\delta$, error $\epsilon$
	\State \textbf{Output}: List of vectors $\mathcal{L}\triangleq \{\tilde{v}_1,...,\tilde{v}_k\}$ for which there is a permutation $\pi:[k]\to[k]$ and signs $\epsilon_1,...,\epsilon_k\in\{\pm 1\}$ for which $\norm{\tilde{v}_i - \epsilon_i v_{\pi(i)}}_2 \le \epsilon$ for all $i\in[k]$, with probability at least $1 - \delta$.
			\State Set $\delta' = \delta/2k$
			\State Set $\zeta = a_{\mathrm{sep}}\Delta\cdot k\cdot\min\left(\frac{a_{\mathrm{eig}}}{k^3}, \frac{\pmin}{64}\right)$
			\State Set $\epsilon_{\text{HMD}} = \frac{\zeta}{a_{\mathrm{eig}}k^3}$.
			\State Set $\epsilon_{\text{boost}} = \min\{\epsilon,\poly(\pmin,\Delta,1/k,1/d)^{k^{3/5}\ln(1/\pmin)}\}$.
			\For{$i\in[k]$}
				\State Let $v'_i$ be the output of \textsc{HyperplaneMomentDescent}($\calD, \delta',\epsilon_{\text{FMD}}$)
				\State Let $\tilde{v}_i$ be the output of \textsc{HyperplaneBoost}($\calD, v'_i,\epsilon_{\text{boost}},\delta'$)
				\State Henceforth when sampling from $\calD$, ignore all samples $x\in\R^d$ for which $|\langle \tilde{v}_i, x\rangle| \le \epsilon_{\text{boost}}\cdot\poly(\log d)$.
			\EndFor
\end{algorithmic}
\end{algorithm}

We can complete the proof of Theorem~\ref{thm:hyperplanes_main}.

\begin{proof}[Proof of Theorem~\ref{thm:hyperplanes_main}]
	By Lemma~\ref{lem:hyperplanecorrect}, every $v'_i$ in \textsc{LearnHyperplanes} is $\frac{\zeta}{a_{\mathrm{eig}}k^3}$-close (up to signs) to a direction $v_{i'}$ of $\calD$, and by Corollary~\ref{cor:hyperplaneboost}, \textsc{HyperplaneBoost} improves this to a vector $\tilde{v}_i$ for which $\norm{\tilde{v}_i - v_{i'}}_2 \le\epsilon_{\text{boost}}$, where \begin{equation}\epsilon_{\text{boost}} = \min\{\epsilon,\poly(\pmin,\Delta,1/k,1/d)^{k^{3/5}\ln(1/\pmin)}\}.\end{equation} As a result, only a $\poly(\pmin,\Delta,1/k,1/d)^{k^{3/5}\ln(1/\pmin)}$ fraction of subsequent samples will be removed, and the resulting error can be absorbed into the sampling error that goes into subsequent calls to \textsc{L2Estimate} and subsequent matrices $\widehat{ \vec{M} }^{(N)}_a \in \R^{d \times d}$ that we run \textsc{ApproxBlockSVD} on, in the remainder of \textsc{LearnWithoutNoise}.
\end{proof} 


\section{Boosting Down the Cosine Integral}
\label{sec:boosting}

The main result that we show in this section is the following local convergence guarantee for \textsc{Boost}.

\begin{theorem}
	There are absolute constants $C, C' > 0$ such that the following holds. Let $\epsilon > 0$, and let $\calD$ be any mixture of spherical linear regressions with separation $\Delta$, noise rate $\noise\le C'\cdot(\pmin \cdot \epsilon\cdot  \Delta^4)^{1/5}$. Suppose $\| v- w_{i^*} \|_2 \le \Delta/\gamma$ for $\gamma = C\cdot \pmin^{1/4}$. Then \textsc{Boost}$(\calD,v,\epsilon,\delta)$ (Algorithm~\ref{alg:mainboosting_cos}) returns $v^*$ satisfying $\| v^* - w_{i^*} \|_2 \le \epsilon$. Additionally, it has sample complexity 
	\begin{equation*}
	\tilde{O}\left(d\cdot \poly(1/\epsilon,1/\Delta) \cdot ( \ln(1/\epsilon)\cdot \ln(1/\pmin) )^{O(\ln(1/\pmin))}\right)
	\end{equation*} 
	and runtime 
	\begin{equation*}
	\tilde{O}\left(d^2\cdot \poly(1/\epsilon,1/\Delta)\cdot ( \ln(1/\epsilon)\cdot \ln(1/\pmin) )^{O(\ln(1/\pmin))}\right).
	\end{equation*}\label{thm:mainboosting_cos}
\end{theorem}

\begin{remark}
	Note that our boosting algorithm can tolerate a warm start at distance $O(\Delta\pmin^{-1/4})$, whereas that of \cite{li2018learning} can only tolerate one at distance $O(\Delta\pmin)$ (see Theorem~\ref{thm:liliang}). Our algorithm can also tolerate regression noise as large as $O( \pmin^{1/5} \epsilon^{1/5} \Delta^{4/5} )$. In particular, if $\epsilon = o(\pmin^{1/20}\Delta^{1/5})$, then our algorithm \textsc{Boost} can tolerate noise rate $\noise = \omega(\epsilon)$.
\end{remark}

In Section~\ref{subsec:background_boost} we recall the boosting algorithm of \cite{li2018learning} to motivate the high-level blueprint for our argument. In Section~\ref{subsec:boost_main} we give the full specification of our boosting algorithm and a proof of Theorem~\ref{thm:mainboosting_cos}.

\subsection{Background: Gravitational Allocation}
\label{subsec:background_boost}

In \cite{li2018learning}, Li and Liang boost a warm start to a fine estimate for one of the $w_i$'s by performing stochastic gradient descent on the (regularized) gravitational potential objective 
\begin{align*}
h(v) = \E_{x,y}[\ln(|\langle x,v\rangle - y| + \xi)]
\end{align*} for some $\xi > 0$ which is introduced to ensure smoothness even when $v = w_i $ for some $i\in[k]$. We emphasize that this objective is \emph{concave}.
For any $i^*\in[k]$, the inner product between the expected gradient step and $w_{i^*} - v^{(t)}$, where $v^{(t)}$ is the current iterate, is given by \begin{align}
	\langle -\Delta h(v^{(t)}), w_{i^*} - v^{(t)}\rangle &= -\E_{x,y}\left[\frac{\sign(\langle x,v\rangle - y)\cdot \langle x,w_{i^*} - v^{(t)}\rangle}{|\langle x,v\rangle - y| + \xi}\right]\label{eq:replace}\\
	&= \frac{1}{k}\sum^k_{i=1}p_i\cdot \E_{x\sim \N (0,\Id)}\left[\frac{\sign(\langle x,w_i - v^{(t)}\rangle)\cdot \langle x,w_{i^*} - v^{(t)}\rangle}{|\langle x,w_i - v^{(t)}\rangle| + \xi}\right].\nonumber
\end{align} They argue that provided $\norm{v^{(0)} - w_{i^*}}_2 \le O(\Delta/k)$, the contribution of the $i^*$-th summand dominates that of all other summands, so the correlation of the gradient step with $w_{i^*} - v^{(t)}$ is sufficiently large that each step contracts the distance to $w_{i^*}$ appreciably.

\subsection{Boosting via the Cosine Integral}
\label{subsec:boost_main}

Here we argue that a warm start of $\norm{v^{(0)} - w_{i^*}}_2 \le O(\Delta\cdot\pmin^{1/4})$ is sufficient if we run gradient descent not on the gravitational potential objective, but on the \emph{cosine integral objective}. Concretely, we propose Algorithm~\ref{alg:mainboosting_cos} below for boosting.

\begin{algorithm}\caption{\textsc{Boost}$(\mathcal{D},v,\epsilon,\delta)$, Theorem~\ref{thm:mainboosting_cos}}\label{alg:mainboosting_cos}
\begin{algorithmic}[1]
	\State \textbf{Input}: Mixture of linear regressions $\calD$ with separation $\Delta$ and noise rate $\noise\le O((\pmin\cdot\epsilon\cdot\Delta^4)^{1/5})$, warm start $v$, accuracy $\epsilon$, failure probability $\delta$
			\State Let $\gamma = C\cdot\pmin^{1/4}$.
			\State Set $v^{(0)} \triangleq v$, $T \triangleq O(d\cdot \Delta^8/\epsilon^8 \cdot \ln(\Delta/\gamma\epsilon)$.
			\State Let $\delta' = \frac{\delta}{2T}$.
			\State Let $\overline{\sigma} = 4$.
			\For{$t = 0,...,T-1$}
				\State Let $\xi_t = \text{\textsc{EstimateMinVariance}}(\mathcal{F}_t,\overline{\sigma},\epsilon/10,\Omega(\ln(1/\pmin)),\delta')/1.1$ \Comment{Algorithm~\ref{alg:estimate_min_variance}}
				\If{$\xi_t\cdot (1.1/0.9) \le \epsilon$}
					\Break
				\EndIf
				\State Draw $N = \poly(1/\xi_t,1/\Delta,\ln(\delta'))$ fresh samples from $\mathcal{D}$, call them $(x_1,y_1),...,(x_N,y_N)$.
				\State Form the empirical gradient
				\begin{equation}
					\delta_t = -\frac{1}{N}\sum^{N}_{i=1}\bone{|\langle x_i,v^{(t)}\rangle - y_i|\ge \xi_t}\cdot \frac{\cos( \xi_t^{-1} \pi |\langle x_i,v^{(t)}\rangle - y_i|) }{\langle x_i, v^{(t)}\rangle - y_i}\cdot x_i. \label{eq:step}
				\end{equation}
				\State Set learning rate $\eta_t \triangleq \frac{\xi^5_t}{2d\Delta^4}\cdot \norm{w_{i^*} - v^{(t)}}_2$ and define \begin{equation}
					v^{(t+1)} = v^{(t)} - \eta_t\delta_t.\label{eq:update}
				\end{equation}
			\EndFor
			\State \Return $v^{(T)}$
\end{algorithmic}
\end{algorithm}


\begin{remark}
	An obvious caveat for our result is the exponential dependence on $\ln(1/\pmin)$, which comes from the need to compute the regularization parameter $\xi_t$ at each step. Similar to the $\xi$ in the gravitational potential objective of \cite{li2018learning}, the $\xi_t$ in \textsc{Boost} is to ensure smoothness. In our case, we need $\xi_t$ to be a lower bound for $\norm{w_{i^*} - v^{(t)}}_2$, and the rate of contraction decreases as $\xi_t$ decreases (see Lemma~\ref{lem:contract_cos} below).
\end{remark}

To show Theorem~\ref{thm:mainboosting_cos}, we first show that if $\xi_t$ is chosen to be sufficiently small at each step, $\norm{w_{i^*} - v^{(t)}}_2$ is guaranteed to contract.

\begin{lemma}
	Let $C,C' > 0$ be the constants in Theorem~\ref{thm:mainboosting_cos}.

	For any $t$ and $\delta > 0$, if $\epsilon/10\le \norm{w_{i^*} - v^{(t)}}_2 \le \Delta/\gamma$ for $\gamma = C\cdot \pmin^{1/4}$ and $\xi_t$ satisfies $\xi_t \le \norm{w_{i^*} - v^{(t)}}_2$, then for $N = \poly(1/\xi_t,1/\Delta,\ln(1/\delta))$, we have with probability at least $1 - \delta$ over the $N$ samples used to form the empirical gradient that
	\begin{equation}
		\left(1 - \frac{\xi^4_t}{\sqrt{d}\Delta^4}\right)\norm{w_{i^*} - v^{(t)}}^2_2\le \norm{w_{i^*} - v^{(t+1)}}^2_2 \le\left(1 - \frac{\xi^8_t}{4d\Delta^8}\right) \cdot \norm{w_{i^*} - v^{(t)}}^2_2.
	\end{equation}\label{lem:contract_cos}
\end{lemma}

\begin{proof}
	The key step is to lower bound the correlation between the negative gradient $-\E[\delta_t]$ and the direction $w_{i^*} - v^{(t)}$ in which we would like to move. 
	We have that {\small
	\begin{align}
		& ~ \langle -\E[\delta_t], w_{i^*} - v^{(t)}\rangle \nonumber\\
		= & ~ \E_{x,y} \left[ \bone{|\langle x , v^{(t)} \rangle - y| \ge \xi_t} \cdot \frac{ \cos \left( \xi_t^{-1} \pi | \langle x , v \rangle - y |\right ) \cdot \langle x,w_{i^*} - v^{(t)}\rangle}{\langle x,v^{(t)}\rangle - y_i}\right]\nonumber\\
		= & ~ \sum^k_{i=1}p_i\E_{\substack{x\sim \N(0,\Id) \\ g\sim \N(0,\noise^2)}}\left[\bone{|\langle x,w_i - v^{(t)}\rangle - g|\ge \xi_t}\cdot\frac{-\cos\left( \xi_t^{-1} \pi |\langle x,w_i - v^{(t)}\rangle-g|\right)\cdot \langle x,w_{i^*} - v^{(t)}\rangle}{\langle x,w_i - v^{(t)}\rangle - g}\right].
		\label{eq:sumoveri}
	\end{align}}
	where the last step follows from the fact that $(x,y)$ comes from component $i$ with probability $p_i$, in which case $\langle x,v^{(t)}\rangle - y_i = -\langle x,w_i - v^{(t)}\rangle$.

	For every $i\in[k]$, define 
	\begin{equation*}
		\beta_i \triangleq \left(\| w_i - v^{(t)} \|^2_2 + \noise^2\right)^{1/2} \ \ \ \text{and} \ \ \ \nu_i \triangleq \frac{\norm{w_i - v^{(t)}}^2_2}{\norm{w_i - v^{(t)}}^2_2 + \noise^2}.
	\end{equation*} 
	We have the naive bounds 
	\begin{equation}
		\| w_i - v^{(t)} \|_2 \le \beta_i \le \max \left\{ \| w_i - v^{(t)} \|_2, \noise \right\} \cdot \sqrt{2}\label{eq:naive1}
	\end{equation} and 
	\begin{equation}
		\min \left\{ \frac{1}{2},\frac{\norm{w_i - v^{(t)}}^2_2}{2\noise^2} \right\} \le \nu_i \le 1\label{eq:naive2}
	\end{equation} for all $i$. Then by Lemma~\ref{lem:mainterm} below, we can bound the $i\neq i^*$ and $i = i^*$ terms of \eqref{eq:sumoveri} to get \begin{align}
		\langle-\E[\delta_t],w_{i^*} - v^{(t)}\rangle &\ge p_{i^*}\frac{0.22\xi_t^3 \nu_{i^*}}{\beta_{i^*}^3} - \sum_{i\neq i^*} p_i \frac{\norm{w_{i^*} - v^{(t)}}_2}{\norm{w_i - v^{(t)}}_2}\cdot \frac{0.26\xi^3_t\nu_i}{\beta_i^3} \nonumber\\
		&\ge \xi^3_t\left[p_{i^*}\cdot\frac{0.22\nu_{i^*}}{\beta^3_{i^*}} - \sum_{i\neq i^*}p_i\frac{0.26\norm{w_{i^*} - v^{(t)}}_2}{\norm{w_i - v^{(t)}}^4_2}\right] \nonumber\\
		&\ge \xi^3_t\left[p_{i^*}\cdot\frac{0.22\nu_{i^*}}{\beta^3_{i^*}} - \frac{0.26\norm{w_{i^*} - v^{(t)}}_2}{(\Delta/2)^4}\right]\label{eq:intermedip},
	\end{align} where in the second step we invoked the lower and upper bounds of \eqref{eq:naive1} and \eqref{eq:naive2} respectively, and in the third step we used the fact that for every $i\neq i^*$, 
	\begin{align*}
	\norm{w_i - v^{(t)}}_2 \ge \norm{w_i - w_{i^*}}_2 - \norm{w_{i^*} - v^{(t)}}_2 \ge \Delta/2.
	\end{align*}
	We proceed by casework based on the relation between $\norm{w_{i^*} - v^{(t)}}_2$ and $\noise$.

	\begin{case}
		$\norm{w_{i^*} - v^{(t)}}_2 \ge \noise$.
	\end{case}

	In this case, we know that $\beta_{i^*} \le \sqrt{2}\norm{w_{i^*} - v^{(t)}}_2 \le \sqrt{2}\Delta/\gamma$ and $\nu_{i^*} \ge 1/2$ by \eqref{eq:naive1} and \eqref{eq:naive2}. From \eqref{eq:intermedip} we get that \begin{equation}
		\langle-\E[\delta_t],w_{i^*} - v^{(t)}\rangle \ge \xi^3_t \norm{w_{i^*} - v^{(t)}}_2 \cdot \left[\pmin \cdot \frac{0.11 \cdot (2^{-3/2})}{(\Delta/\gamma)^4} - \frac{0.26}{(\Delta/2)^4}\right]
	\end{equation} 

	So there is an absolute constant $C > 0$ such that for $\gamma = C\cdot \pmin^{1/4}$, we have that \begin{equation}
		\langle-\E[\delta_t],w_{i^*} - v^{(t)}\rangle \ge \xi^3_t \norm{w_{i^*} - v^{(t)}}_2 \cdot \Delta^{-4}\label{eq:contraction}
	\end{equation}

	\begin{case}
		$\norm{w_{i^*} - v^{(t)}}_2 \le \noise$.
	\end{case}

	In this case, we know that $\beta_{i^*} \le \sqrt{2}\noise$ and $\nu_{i^*}\ge \frac{\norm{w_{i^*} - v^{(t)}}^2_2}{2\noise^2}$ by \eqref{eq:naive1} and \eqref{eq:naive2}. From \eqref{eq:intermedip} we get that \begin{align*}
		\langle-\E[\delta_t],w_{i^*} - v^{(t)}\rangle& \ge \xi^3_t\cdot\left[\pmin\cdot \frac{0.22}{\noise^3 \cdot 2^{3/2}}\cdot\frac{\norm{w_{i^*} - v^{(t)}}^2_2}{2\noise^2} - \frac{0.26\norm{w_{i^*} - v^{(t)}}_2}{(\Delta/2)^4}\right] \\
		&\ge \xi^3_t \norm{w_{i^*} - v^{(t)}}_2 \cdot \left[\pmin\cdot\frac{0.22\cdot (2^{-5/2})}{\noise^5}\cdot\norm{w_{i^*} - v^{(t)}}_2 - \frac{0.26}{(\Delta/2)^4}\right] \\
		&\ge \xi^3_t \norm{w_{i^*} - v^{(t)}}_2 \cdot \left[\pmin\cdot\frac{0.22\cdot (2^{-5/2})\cdot(\epsilon/3)}{\noise^5} - \frac{0.26}{(\Delta/2)^4}\right].
	\end{align*} By taking $\gamma = C\cdot\pmin^{1/4}$ as in the previous case, we see that there exists some absolute constant $C' > 0$ such that for $\noise\le C'\cdot(\pmin \cdot \epsilon\cdot  \Delta^4)^{1/5}$, \eqref{eq:contraction} still holds.

	To show moving in the direction opposite the \emph{empirical} gradient suffices, we need concentration. First note that for every sample $(x,y)$, 
	\begin{equation*}
		\norm{\bone{|\langle x,v^{(t)}\rangle - y|\ge \xi_t}\cdot \frac{\cos( \xi_t^{-1} \pi |\langle x,v^{(t)}\rangle - y|) }{\langle x, v^{(t)}\rangle - y}\cdot x}_2 \le \frac{\norm{x}}{\xi_t},
	\end{equation*} 
	and likewise 
	\begin{equation*}
		\left|\left\langle\bone{|\langle x,v^{(t)}\rangle - y|\ge \xi_t}\cdot \frac{\cos( \xi_t^{-1} \pi |\langle x,v^{(t)}\rangle - y|) }{\langle x, v^{(t)}\rangle - y}\cdot x, \frac{w_{i^*} - v^{(t)}}{\norm{w_{i^*} - v^{(t)}}_2}\right\rangle\right| \le \frac{\left|\left\langle x,\frac{w_{i^*} - v^{(t)}}{\norm{w_{i^*} - v^{(t)}}_2}\right\rangle\right|}{\xi_t}.
	\end{equation*} 
	Furthermore, by \eqref{eq:contraction}, the expected gradient satisfies 
	\begin{equation*}
		\left\langle-\E[\delta_t],\frac{w_{i^*} - v^{(t)}}{\norm{w_{i^*} - v^{(t)}}_2}\right\rangle \ge \xi^3_t \Delta^{-4}.
	\end{equation*} 
	By standard Gaussian concentration, for some $N\ge \poly(\xi_t^{-1},\Delta^{-1},\ln(1/\delta))$, we get that with probability at least $1 - \delta/3$, 
	\begin{equation*}
		\norm{\delta_t}\le \frac{2\sqrt{d}}{\xi_t} \ \ \ \text{and} \ \ \ \left\langle -\delta_t, \frac{w_{i^*} - v^{(t)}}{\norm{w_{i^*} - v^{(t)}}_2}\right\rangle \ge \frac{1}{2}\xi^3_t\Delta^{-4}.
	\end{equation*} 
	By \eqref{eq:update}, 
	\begin{equation*}
		\norm{w_{i^*} - v^{(t+1)}}^2_2 = \norm{w_{i^*} - v^{(t)}}^2_2 + \eta_t^2\norm{\delta_t}^2_2 - 2\eta_t\langle -\delta_t, w_{i^*} - v^{(t)}\rangle,
	\end{equation*} 
	so by taking learning rate $\eta_t \triangleq \frac{\xi^5_t}{2d\Delta^4}\cdot \norm{w_{i^*} - v^{(t)}}_2$, we ensure that 
	\begin{equation*}
		\norm{w_{i^*} - v^{(t+1)}}^2_2 \le \left(1 - \frac{\xi^8_t}{4d\Delta^8}\right)\norm{w_{i^*} - v^{(t)}}^2_2.
	\end{equation*} At the same time, from the naive bounds $\norm{\delta_t}^2_2 \ge 0$ and $2\eta_t\langle -\delta_t,w_{i^*} - v^{(t)}\rangle \le \frac{\xi^4_t}{\sqrt{d}\Delta^4}\norm{w_{i^*} - v^{(t)}}^2_2$ which follows by Cauchy-Schwarz, we also have \begin{equation}\norm{w_{i^*} - v^{(t+1)}}^2_2 \ge \left(1 - \frac{\xi^4_t}{\sqrt{d}\Delta^4}\right)\norm{w_{i^*} - v^{(t)}}^2_2.\end{equation}
\end{proof}

	To complete the proof of Lemma~\ref{lem:contract_cos}, it remains to prove the following lemma which was crucial to establishing \eqref{eq:intermedip}.

\begin{lemma}\label{lem:mainterm}
	For any vectors $a,b\in\R^d$ and $\xi\le \norm{b}_2$, we have that 
	\begin{equation}\label{eq:maintermbound}
		\left|\E_{\substack{x\sim \N(0,\Id) \\ g\sim \N(0,\noise^2)}}\left[\bone{|\langle b,x\rangle + g|\ge \xi}\cdot \frac{-\cos\left( \xi^{-1}  \pi |\langle b,x\rangle + g|\right)}{\langle b,x\rangle + g}\cdot \langle a,x\rangle\right]\right| \le \frac{\norm{a}_2}{\norm{b}_2}\cdot\frac{\norm{b}^2_2}{\noise^2 + \norm{b}^2_2}\cdot \frac{0.26\xi^3}{(\noise^2 + \norm{b}^2_2)^{3/2}}.
	\end{equation} 
	Furthermore, we have that for $a = b$, 
	\begin{equation}\label{eq:maintermbound2}
		\E_{\substack{x\sim \N(0,\Id) \\ g\sim \N(0,\noise^2)}}\left[\bone{|\langle b,x\rangle + g|\ge \xi}\cdot \frac{-\cos\left( \xi^{-1}  \pi |\langle b,x\rangle + g|\right)}{\langle b,x\rangle + g}\cdot \langle b,x\rangle\right] = \frac{\norm{b}^2_2}{\noise^2 + \norm{b}^2_2}\cdot \frac{[0.22,0.26]\cdot \xi^3}{(\noise^2 + \norm{b}^2_2)^{3/2}}.
	\end{equation}
\end{lemma}

\begin{proof}
	For notational convenience, given $x\sim \N(0,\Id)$, let $\mathcal{E}_{\xi}$ denote the event that $|\langle b,x\rangle + g| \ge\xi$. We may write 
	\begin{align*}
	a = \frac{\rho\norm{a}_2}{\norm{b}_2} \cdot b + \sqrt{1 - \rho^2} \cdot b^{\perp}
	\end{align*}
	for $\rho = \frac{\langle a,b\rangle}{\norm{a}_2\norm{b}_2}$ and $b^{\perp}\in\S^{d-1}$ orthogonal to $b$. Then the left-hand side of \eqref{eq:maintermbound} can be written as \begin{align}
		& ~ \E_{\substack{x\sim \N(0,\Id) \\ g\sim \N(0,\noise^2)}}\left[\bone{\mathcal{E}_{\xi}} \cdot\frac{-\cos\left( \xi^{-1} \pi | \langle b,x\rangle + g|\right)}{\langle b,x\rangle + g}\cdot \langle a,x\rangle\right] \notag \\
		= & ~ \E_{\substack{x\sim \N(0,\Id) \\ g\sim \N(0,\noise^2)}}\left[\bone{\mathcal{E}_{\xi}} \cdot\frac{-\cos\left( \xi^{-1} \pi |\langle b,x\rangle + g|\right)}{\langle b,x\rangle + g}\cdot \frac{\rho\norm{a}_2}{\norm{b}_2}\cdot\langle b,x\rangle\right] \notag \\  
		= & ~ -\frac{\rho\norm{a}_2}{\norm{b}_2}\E_{\substack{x\sim \N(0,\Id) \\ g\sim \N(0,\noise^2)}}\left[\bone{\mathcal{E}_{\xi}}\frac{\cos\left( \xi^{-1} \pi |\langle b,x\rangle + g|\right)}{\langle b,x\rangle + g}\cdot \langle b,x\rangle\right] \notag \\
		= & ~ -\frac{\rho\norm{a}_2}{\norm{b}_2}\E_{\substack{g\sim\N(0,\noise^2) \\ g'\sim \N(0,\norm{b}^2_2)}}\left[\bone{\mathcal{E}_{\xi}}\frac{\cos\left( \xi^{-1} \pi (g + g')\right)}{g + g'}\cdot g'\right] \notag \\
		= & ~ -\frac{\rho\norm{a}_2}{\norm{b}_2}\cdot\frac{\norm{b}^2_2}{\noise^2 + \norm{b}^2_2}\cdot\E_{\substack{g\sim\N(0,\noise^2) \\ g'\sim \N(0,\norm{b}^2_2)}}\left[\bone{\mathcal{E}_{\xi}}\cos\left( \xi^{-1} \pi (g + g')\right)\right] \notag\\
		= & ~ -\frac{\rho\norm{a}_2}{\norm{b}_2}\cdot\frac{\norm{b}^2_2}{\noise^2 + \norm{b}^2_2}\cdot\E_{g\sim\N(0,\noise^2 + \norm{b}^2_2)}\left[\bone{|g|\ge \xi}\cos\left( \xi^{-1} \pi g\right)\right]
	\end{align}
	where the first step follows from the fact that $\langle b,x\rangle$ and $\langle b^{\perp},x\rangle$ are independent mean-zero random variables, the third step follows by the fact that $\cos(\cdot)$ is even, and the penultimate step follows from the fact that we may decompose $g'$ in terms of $g + g'$ as \begin{equation*}
	 	g' = \frac{\norm{b}^2_2}{\noise^2 + \norm{b}^2_2}(g+g') + h
	\end{equation*} 
	for Gaussian $h$ independent of $g + g'$.

	We conclude the proof of the first half of the lemma by noting that $|\rho|\le 1$ and appealing to the upper bound in Lemma~\ref{lem:technical}, where we take $\beta = (\noise^2 + \norm{b}^2_2)^{1/2}$.

	Next, the upper bound in the second half of the lemma follows immediately from Lemma~\ref{lem:technical}. Finally, for the lower bound in the second half of the lemma, the lower bound in Lemma~\ref{lem:technical} gives
	\begin{equation*}
		\E_{\substack{x\sim \N(0,\Id) \\ g\sim\N(0,\noise^2)}}\left[\bone{|\langle b,x\rangle + g|\ge \xi}\cdot \frac{-\cos\left( \xi^{-1} \pi |\langle b,x\rangle + g|\right)}{\langle b,x\rangle + g}\cdot \langle b,x\rangle\right] \ge \frac{0.23\xi^3}{(\noise^2 + \norm{b}^2_2)^{3/2}} - \exp \left( -\frac{\pi^2(\noise^2 + \norm{b}^2_2)}{2\xi^2} \right),
	\end{equation*} 
	and we conclude by noting that for $\xi \le\beta$, $\exp( - \frac{\beta^2}{2\xi^2}) \le \exp( -\pi^2 / 2 ) / \beta^3 \le 0.01/\beta^3$.
\end{proof}

\begin{lemma}\label{lem:technical}
	For any $\beta,\xi>0$ for which $\xi \le \beta$, we have that
	\begin{equation}\label{eq:technicaleq}
	 \exp( - \frac{\pi^2\beta^2}{2\xi^2} ) - \E_{g\sim \N(0,\beta^2)}\left[\bone{\mathcal{E}_{\xi}} \cdot \cos ( \xi^{-1} \pi |g| ) \right] \in \left[ \frac{0.23\xi^3}{\beta^3} , \frac{0.26\xi^3}{\beta^3} \right].
	\end{equation}
\end{lemma}

\begin{proof}
	We can rewrite the LHS in \eqref{eq:technicaleq} as follows:
	\begin{equation*}
		\textrm{LHS} = \underbrace{ \exp( - \frac{\pi^2\beta^2}{2\xi^2} ) - \frac{1}{\beta\sqrt{2\pi}}\int^{\infty}_{-\infty} \exp \left( - \frac{x^2}{2\beta^2} \right) \cos(\pi x/\xi) \d x}_{\encircle{I}} + \underbrace{\frac{2}{\beta\sqrt{2\pi}}\int^{\xi}_{0} \exp \left( -\frac{x^2}{2\beta^2} \right) \cos(\pi x/\xi) \d x}_{\encircle{II}}
	\end{equation*} 
	Using Claim~\ref{cla:part_1_simplification}, we can show $\encircle{I} = 0$. Using Claim~\ref{cla:part_2_two_side_bounds}, we can upper and lower bound $\encircle{II}$.
\end{proof}

\begin{claim}\label{cla:part_1_simplification}
We have
\begin{align*}
\frac{1}{\beta\sqrt{2\pi}}\int^{\infty}_{-\infty}e^{-\frac{x^2}{2\beta^2}}\cos(\pi x/\xi) \d x = \exp \left( -\frac{\pi^2 \beta^2}{2\xi^2} \right). 
\end{align*}

\end{claim}

\begin{proof}

Let $\upsilon = \frac{ \xi^2 }{ 2\pi^2 \beta^2 }$. Noting that $\cos(x) = \Re(e^{- \i x})$, one can compute $\text{LHS}$ by a standard contour integral. 
	\begin{align*}
		\text{LHS}
		= & ~ \frac{1}{\beta\sqrt{2\pi}}\cdot\frac{\xi}{\pi}\int^{\infty}_{-\infty} \exp \left( -\frac{\xi^2 x^2}{2\pi^2 \beta^2} \right) \cdot \cos(x) \d x \\
		= & ~ \frac{1}{\beta\sqrt{2\pi}}\cdot\frac{\xi}{\pi}\int^{\infty}_{-\infty} \exp(-\upsilon x^2) \cdot \cos(x) \d x \\
		= & ~ \frac{1}{\beta\sqrt{2\pi}}\cdot\frac{\xi}{\pi}\cdot\Re\int^{\infty}_{-\infty} \exp\left(-\upsilon x^2 - \i x\right) \d x \\
		= & ~ \frac{1}{\beta\sqrt{2\pi}}\cdot\frac{\xi}{\pi}\cdot\Re\int^{\infty}_{-\infty} \exp (-\upsilon (x + \i / (2\upsilon) )^2 - 1 / (4\upsilon) ) \d x \\
		= & ~ \frac{\exp(- 1 / (4\upsilon))}{\beta\sqrt{2\pi}}\cdot\frac{\xi}{\pi}\cdot\Re\int^{\infty}_{-\infty} \exp (-\upsilon (x + \i / (2\upsilon) )^2  ) \d x \\
		= & ~ \frac{ \exp ( -1 / ( 4 \upsilon ) ) }{\beta\sqrt{2\pi}}\cdot\frac{\xi}{\pi}\cdot\Re\int^{\infty + \i / (2 \upsilon ) }_{-\infty + \i / ( 2 \upsilon ) } \exp( -\upsilon x^2 ) \d x\\
		= & ~ \frac{ \exp ( -1 / ( 4 \upsilon ) ) } {\beta\sqrt{2\pi}}\cdot\frac{\xi}{\pi}\cdot\frac{\sqrt{\pi}}{\sqrt{\upsilon}} \\
		= & ~ \exp \left( - \frac{ \pi^2 \beta^2 }{ 2 \xi^2 } \right),
	\end{align*}
where the second step follows from definition of $v$, the third step follows from $\cos(x) = \Re ( \exp(-\i x) )$, the fourth step follows from $-v x^2 - \i x =  - v (x^2 + \i x/v - 1/(4v^2) ) - 1/4v= -\upsilon (x + \i / (2\upsilon) )^2 - 1 / (4\upsilon) $, the fifth step follows from pulling the term $\exp(-1/(4v))$ out of integral, the sixth step follows from shifting the integral range, the seventh step follows from Cauchy's theorem\footnote{By Cauchy's theorem, the integral around the box in the complex plane with vertices $-R$, $R$, $-R + \i/(2\upsilon)$, and $R + \i/(2\upsilon)$ is zero. The sum of the contributions of the edges between $-R$ and $-R + \i/(2\upsilon)$ and between $R$ and $R + \i/(2\upsilon)$is imaginary and thus contributes 0 to the real part. If we take $R\to\infty$, we see that the integral we want to compute is the same as the one where you ignore the $\i/(2\upsilon)$ terms, which is a standard Gaussian integral.}, and the last step follows from definition of $v$.

Thus, we complete the proof.
\end{proof}

\begin{claim}\label{cla:part_2_two_side_bounds}
Let $\xi \leq \beta$. We have
\begin{align*}
\frac{2}{\beta\sqrt{2\pi}}\int^{\xi}_{0} \exp \left( - \frac{x^2}{2\beta^2} \right) \cos(\pi x/\xi) \d x ~ \in ~ [ 0.23 \xi^3/\beta^3, 0.26 \xi^3/\beta^3 ].
\end{align*}
\end{claim}

\begin{proof}

We will use the bound 
\begin{align*}
1 - \frac{x^2}{2\beta^2} \le \exp \left( -\frac{x^2}{2\beta^2} \right) \le 1 - \frac{x^2}{2\beta^2} + \frac{x^4}{8\beta^4}
\end{align*}
to obtain upper and lower bounds. 

Noting that $\cos(\pi x/\xi) \ge 0$ for $x\in[0,\xi/2]$ and $\cos(\pi x/\xi) \le 0$ for $x\in [\xi/2,\xi]$, we get that
	\begin{align*}
		\text{LHS}
		\ge & ~ \frac{2}{\beta\sqrt{2\pi}}\int^{\xi/2}_0 \cos(\pi x/\xi)\cdot\left(1 - \frac{x^2}{2\beta^2}\right) \d x + \frac{2}{\beta\sqrt{2\pi}}\int^{\xi}_{\xi/2} \cos(\pi x/\xi)\cdot\left(1 - \frac{x^2}{2\beta^2} + \frac{x^4}{8\beta^4}\right) \d x \\
		\ge & ~ \frac{2}{\beta\sqrt{2\pi}}\left(\frac{\xi^3}{\pi \beta^2} - 0.021 \cdot \frac{\xi^5}{\beta^4}\right) \ge \frac{0.23\xi^3}{\beta^3}
	\end{align*} 
	and \begin{align*}
		\text{LHS}
		\le & ~ \frac{2}{\beta\sqrt{2\pi}}\int^{\xi/2}_0 \cos(\pi x/\xi)\cdot\left(1 - \frac{x^2}{2\beta^2} + \frac{x^4}{8\beta^4}\right) \d x + \frac{2}{\beta\sqrt{2\pi}}\int^{\xi}_{\xi/2} \cos(\pi x/\xi)\cdot\left(1 - \frac{x^2}{2\beta^2}\right) \d x \\
		\le & ~ \frac{2}{\beta\sqrt{2\pi}}\left(\frac{\xi^3}{\pi \beta^2} + 0.0002 \cdot \frac{\xi^5}{\beta^4}\right) \le  \frac{0.26\xi^3}{\beta^3},
	\end{align*} 
	where in the last steps we used the fact that $\xi \le \beta$.
\end{proof}

We can now complete the proof of Theorem~\ref{thm:mainboosting_cos}.

\begin{proof}[Proof of Theorem~\ref{thm:mainboosting_cos}]
	The only probabilistic components of \textsc{Boost} is the invocation of \textsc{EstimateMinVariance} and the event of Lemma~\ref{lem:contract_cos} holding at each step. For a given $t$, with probability $1 - 2\delta' = 1 - \delta/T$ these two events both hold, so by a union bound over all $T$ iterations, the failure probability of \textsc{Boost} is at most $\delta$ as desired.

	We now proceed to show correctness of \textsc{Boost}. Conditioned on making progress in every step of \textsc{Boost}, note that $\max_i \norm{w_i - v^{(t)}}_2 \le \Delta/\gamma + \norm{w_i - w_{i^*}} \le \Delta/\gamma + 2 \le 4$, so $\overline{\sigma} = 4$ is always a valid upper bound for the maximum variance of any component of a univariate mixture of Gaussians $\calF_t$ encountered over the course of \textsc{Boost}. So we conclude by Lemma~\ref{lem:estimate_min_variance} that $\xi_t \le \norm{w_{i^*} - v^{(t)}}_2$. Then because of the lower bound of Lemma~\ref{lem:contract_cos}, the inequality $\xi_t\cdot(1.1/0.9) \ge \norm{w_i - v^{(t)}}_2$, and the fact that \textsc{Boost} breaks out of its main loop if $\xi_t \cdot(1.1/0.9)$, we know that at all times in main loop of \textsc{Boost}, $\norm{w_{i^*} - v^{(t)}}_2 \ge \epsilon/10$.

	So by the upper bound of Lemma~\ref{lem:contract_cos} and the fact that $\xi_t = \Omega(\epsilon)$, we conclude that after $T \triangleq O\left(d\cdot \Delta^8/\epsilon^8 \cdot \ln(\Delta/\gamma\epsilon)\right)$ iterations, $\norm{w_{i^*} - v^{(T)}}_2 \le \epsilon$.

	For the time and sample complexity, at every time step $t$ we must draw 
	\begin{align*}
	N = \poly(1/\xi_t,1/\Delta,\ln(1/\delta')) \le \poly(1/\epsilon,1/\Delta,\ln(T),\ln(1/\delta))
	\end{align*} 
	samples to form the empirical gradient in time $d\cdot N$. We also know that each invocation of \textsc{EstimateMinVariance}, by Lemma~\ref{lem:estimate_min_variance}, requires time and sample complexity 
	\begin{align*}
	N' \triangleq \tilde{O}\left( ( \mu_0\cdot\ln(1/\epsilon)\ln(1/\pmin) )^{O(\ln(1/\pmin))}\cdot \ln(2T/\delta)\right).
	\end{align*}
	So \textsc{Boost} requires 
	\begin{equation*}
	T\cdot (N + N') = \tilde{O}\left(d\cdot \poly(1/\epsilon,1/\Delta) \cdot ( \ln(1/\epsilon)\cdot \mu_0\cdot \ln(1/\pmin) )^{O(\ln(1/\pmin))}\right)
	\end{equation*} 
	samples and 
	\begin{equation*}
	T(d\cdot N + N') = \tilde{O}\left(d^2\cdot\poly(1/\epsilon,1/\Delta) \cdot ( \ln(1/\epsilon)\cdot \mu_0\cdot \ln(1/\pmin) )^{O(\ln(1/\pmin))}\right)
	\end{equation*} time.
\end{proof} 

\section{Acknowledgments}

The first and second authors would like to thank Sam Hopkins and Tselil Schramm for answering questions regarding low-degree hypothesis testing, and Ankur Moitra for helpful discussions at an early stage of this project.

\bibliographystyle{alpha}
\bibliography{ref}

\newpage
\appendix
\section*{Appendix}

\section{Failure of Low-Degree Identifiability}
\label{subsec:failure}

In this section, we exhibit a pair of mixtures of spherical linear regressions which are far in parameter distance but which agree on all degree-$\Omega(k)$ moments.
This demonstrates that any method which hopes to achieve sample complexity which is subexponential in $k$ cannot rely solely on low order moments of the MLR.

First, we exhibit a pair of non-identical univariate mixtures of zero-mean Gaussians whose moments match up to degree $2k-1$ and whose variances and mixing weights satisfy reasonable bounds. We remark that the proof, in particular the application of Borsak-Ulam, is largely inspired by that of Lemma 2.9 in \cite{hardt2015tight}.

\begin{lemma}
	There exist $\sigma_1,...,\sigma_k,\sigma'_1,...,\sigma'_k\ge 0$ such that the following holds. Let $D_1$ (resp. $D_2$) be the uniform mixture of univariate Gaussians with components $\N(0,\sigma^2_1),...,\N(0,\sigma^2_k)$ (resp. $\N(0,\sigma'^2_1),...,\N(0,\sigma'^2_k)$). Then \\
	1) there is some $i\in[k]$ for which $|\sigma_i - \sigma'_j| > \Omega(1/\sqrt{k})$ for all $j\in[k]$,\\
	2) $|\sigma_i - \sigma_j|, |\sigma'_i - \sigma'_j| > 1/2$ for all $i\neq j$, \\
	3) $\sigma_i,\sigma'_i\in[1/2,k+1]$ for all $i\in[k]$, and \\
	4) $D_1$ and $D_2$ match on all moments of degree at most $2k - 1$.\label{lem:univariate}
\end{lemma}

\begin{proof}
	For each $i\in[k]$, define $\sigma_i(z) = i + \alpha z$ for $\alpha = 1/4$ and consider the map $M:\S^{k-1}\to\R^{k-1}$ given by \begin{equation}
		M(z)_{\ell} = \sum^k_{i=1}\sigma_i(z)^{2\ell} \ \ \ \ell = 1,...,k-1.
	\end{equation} $M$ is clearly continuous, so by Borsak-Ulam, there exists $z\in\S^{k-1}$ for which $M(z) = M(-z)$. For each $i\in[k]$, define $\sigma_i \triangleq \sigma_i(z)$ and $\sigma'_i \triangleq \sigma_i(-z)$. Then because $\alpha = 1/4$ and $\norm{z}_{\infty}\le 1$, $\sigma_i,\sigma'_i\in[i-1/4,i+1/4]$, 2) and 3) are immediately satisfied. Furthermore, this implies that for any $i\in[k]$, $|\sigma_i - \sigma'_j| > 1/2$ for all $j\neq i$. For $j = i$, $|\sigma_i - \sigma'_i| = 2|z_i|$, and because $\norm{z}_2 = 1$, there must exist some $i$ for which $|z_i|\ge\frac{1}{\sqrt{k}}$, from which 1) follows.

	To see that 4) is satisfied, first note that $D_1$ and $D_2$ are mixtures of zero-mean Gaussians and thus both have odd-degree moments equal to zero. Then for any $1\le \ell \le k-1$, note that the the $2\ell$-th moment of $D_1$ is 
	\begin{align*}
	\frac{1}{k}\sum^k_{i=1}\sigma^{2\ell}_i\cdot (2\ell - 1)!! = \frac{1}{k}(2\ell - 1)!!\cdot M(z)_{\ell}.
	\end{align*}
	Likewise, the $2\ell$-th moment of $D_2$ is 
	\begin{align*}
	\frac{1}{k}(2\ell - 1)!!\cdot M(-z)_{\ell}.
	\end{align*}
	So because $M(z)_{\ell} = M(-z)_{\ell}$ for all $\ell = 1,...,k-1$, we conclude that 4) is satisfied.
\end{proof}

We can now exhibit a moment-matching example for mixtures of linear regressions. Let the parameters $\sigma_1,...,\sigma_k,\sigma'_1,...,\sigma'_k$ be as in Lemma~\ref{lem:univariate}. 

\begin{lemma}
	Take any mixture of spherical linear regressions $\calD$ in $\R^d$ with mixing weights $p_1,...,p_k$ and $\Omega(1)$-separated regressors $v_1,...,v_k\in\R^d$ satisfying $\norm{v_i}_2\le\poly(k)$ for all $i\in[k]$. Take any additional direction $v\in\S^{d-1}$, and any $\lambda\ge 0$. Let $\calD_1$ (resp. $\calD_2$) be the mixture of $3k$ linear regressions with regressors $v_1,...,v_k,\pm\sigma_1v,...,\pm\sigma_kv$ (resp. regressors $v_1,...,v_k,\pm\sigma'_1v,...,\pm\sigma'_kv$) and mixing weights $\frac{p_1}{Z},...,\frac{p_k}{Z},\frac{\lambda/2k}{Z},...,\frac{\lambda/2k}{Z}$, where $Z = \lambda + 1$.

	Then $\calD_1,\calD_2$ satisfy the following: 

	\begin{enumerate}
		\item Both are mixtures of $\Omega(1)$-separated linear regressions whose regressors are $\poly(k)$-bounded in $L_2$ norm
		\item They match on all moments of degree at most $2k - 1$
		\item $\tvd(\calD_1,\calD_2) = \frac{\lambda}{\lambda + 1}$
		\item There exists a regressor $w$ of $\calD_1$ such that for any regressor $w'$ of $\calD_2$, $\norm{w - w'}_2 = \Omega(\sqrt{k})$.
	\end{enumerate}

\end{lemma}

\begin{proof}
	1) follows by 2) and 3) from Lemma~\ref{lem:univariate}. 4) follows by 1) from Lemma~\ref{lem:univariate}. For 3), note that the components of $\calD_1,\calD_2$ in direction $v$ all have disjoint support, so $\tvd(\calD_1,\calD_2) = \frac{\lambda}{\lambda + 1}$.

	It remains to check that $\calD_1,\calD_2$ match on moments of degree at most $2k - 1$. As $\calD_1,\calD_2$ are identical on the components they share with $\calD$, it suffices to show this for the mixtures $\calD'_1,\calD'_2$ obtained by conditioning out the components appearing in $\calD$, that is, the two mixtures of $2k$ spherical linear regressions with uniform mixing weights and directions $\pm\sigma_1 v,...,\pm\sigma_k v$ and directions $\pm\sigma'_1 v,...,\pm\sigma'_k v$ respectively.

	Equivalently, we must show that for any direction $(\vec{x},y)\in\R^{d+1}$, where $\vec{x}\in\R^d$ and $y\in\R$, the univariate Gaussian mixtures $D_1,D_2$ obtained from projecting $\calD'_1,\calD'_2$ in the direction $(\vec{x},y)$ have identical degree-$s$ moment for any $s \le 2k - 1$. These moments will be zero for odd $s$. For $s = 2\ell$, noting that for any $\sigma \ge 0$, \begin{equation}
		(\vec{x},y)^{\top}\Sigma(\sigma v)(\vec{x},y) = \norm{x}^2_2 + \sigma^2 y^2 + 2\sigma y\langle v,x\rangle.
	\end{equation} Without loss of generality, assume $\norm{x}_2 = 1$, and let $\gamma\triangleq \langle v,x\rangle$. We see that the projection $D_1$ has $2\ell$-th moment \begin{equation}
		\frac{1}{k}(2\ell - 1)!!\cdot\left[\sum^k_{i=1}(\sigma^2_i y^2 + 1 + 2\sigma_i\gamma y)^{\ell} + (\sigma^2_i y^2 + 1 - 2\sigma_i\gamma y)^{\ell}\right].\label{eq:d1moments}
	\end{equation} Note that there is some degree-$\ell$ polynomial $p$ for which the $i$-th summand in \eqref{eq:d1moments} is $p(\sigma^2_i)$. In the same way, we can see that the projection $D_2$ has $2\ell$-th moment $\frac{1}{k}(2\ell - 1)!! \cdot \sum^k_{i=1}p(\sigma'^2_i)$. As the univariate mixtures in Lemma~\ref{lem:univariate} match on all $2\ell$-th moments for $\ell \le k -1$, we know that $\sum^k_{i=1}p(\sigma^2_i) = \sum^k_{i=1}p(\sigma'^2_i)$ for all polynomials $p$ of degree at most $k - 1$, so the projections $D_1,D_2$ indeed match on all moments up to degree $2k - 1$.
\end{proof}



\section{Integrating Against Fourier Transforms of Piecewise Polynomials}
\label{sec:hermite}

In this section we prove Lemma~\ref{lem:piecewise-fourier-moment}.
Note that there are indeed explicit expressions for the Fourier moments of piecewise polynomials in terms of hypergeometric functions, but we avoid explicitly describing these for simplicity.

We will show Lemma~\ref{lem:piecewise-fourier-moment} in a couple of steps.
First, we show:
\begin{lemma}
\label{lem:trig-ints}
Let $r$ be a nonnegative integer.
Then, we have that 
\[
\int_0^1 x^r \cos (a x) \d x = \frac{A_r (a, \sin (a), \cos (a))}{a^r} \; , \qquad \int_0^1 x^r \sin (ax) \d x = \frac{B_r (a, \sin (a), \cos (a))}{a^r}\; ,
\]
where $A_r, B_r$ are degree-$r$ polynomials over $\R^3$ whose coefficients can be computed in time $O(r^2)$.
\end{lemma}
\begin{proof}
We proceed by induction on $r$.
The base case is trivial: if $r = 0$, then 
\[
 \int_0^1 \cos (ax) \d x = \frac{\sin (a)}{a} \; ,
\]
and
\[
 \int_0^1 \cos (ax) \d x = \frac{1 - \cos (a)}{a} \; ,
\]
Now assume $r > 0$, and that the claim holds for $r - 1$.
Then by integration by parts, 
\begin{align*}
\int_0^1 x^r \cos (ax) \d x 
= & ~ \frac{\sin (a)}{a} - \frac{r}{a}  \int \alpha_r (x) \sin (ax) \d x \\
= & ~ \frac{1}{a^r} \left( a^{r - 1} \sin (a) - r B_{r - 1} (a, \sin (a), \cos (a)) \right) \; ,
\end{align*}
and similarly
\begin{align*}
\int_0^1 x^r \sin (ax) \d x 
= & ~ \frac{1 - \cos(a)}{a} + \frac{r}{a}  \int \alpha_r (x) \cos (ax) \d x \\
= & ~ \frac{1}{a^r} \left( a^{r- 1} (1 - \cos (a)) - r A_{r - 1} (a, \sin (a), \cos (a)) \right) \; .
\end{align*}
This establishes that these are of the desired form.
Moreover, this recurrence demonstrates that given the coefficients to $A_{r - 1}, B_{r - 1}$, one can obviously compute the coefficients to $A_r, B_r$ using at most $O(r)$ additional time.
This completes the proof.
\end{proof}
\noindent
Note that we must have $\frac{A_r (a, \sin (a), \cos (a))}{a^r}$ and $\frac{B_r (\sin (a), \cos (a))}{a^r}$ converge to a finite value as $a \to 0$, as they must both converge to $\int_0^1 x^r \d x = r - 1$.
In particular, they are both analytic functions over the entire real line, if we take the convention that these functions evaluate to $r - 1$ at $0$, which we will.
We now show:
\begin{lemma}
\label{lem:monomial}
Let $\tau > 0$, and let $r, \ell$ be non-negative integers.
Let $\alpha_r (x) = x^r (x) \cdot \mathbf{1}_{[0, 1]} (x)$.
There is an algorithm that runs in time $O(r^2)$ and outputs
\[
\int_{\tau}^\tau \widehat{\alpha_r} [\omega] \cdot \omega^\ell \d \omega \; .
\]
\end{lemma}
\begin{proof}
By Lemma~\ref{lem:trig-ints}, we know that there exist $A_{r - 1}, B_{r - 1}$ which are degree $r$ polynomials whose coefficients we can compute in $O(r^2)$ time so that
\[
\widehat{\alpha_r} [w] = \frac{A_r (2 \pi \omega, \sin (2 \pi \omega), \cos (2 \pi \omega ))}{(2 \pi \omega)^r} + \i \frac{B_r (2 \pi \omega, \sin (2 \pi \omega), \cos (2 \pi \omega))}{(2 \pi \omega)^r} \; .
\]
Therefore we may evaluate the integral 
\[
\int_{\tau}^\tau \frac{A_r (2 \pi \omega, \sin (2 \pi \omega), \cos (2 \pi \omega ))}{(2 \pi \omega)^r} \cdot \omega^{\ell} \d \omega 
\]
in additional $O(\ell r^2)$ time by first applying integration by parts $r - \ell$ times to remove the denominator, and then solving the trigonometric integral. 
and similarly we can evaluate
\[
\int_{\tau}^\tau \frac{B_r (2 \pi \omega, \sin (2 \pi \omega), \cos (2 \pi \omega ))}{(2 \pi \omega)^r} \cdot \omega^{\ell} \d \omega \; .
\]
in time $O(r^2)$.
This completes the proof.
\end{proof}
\noindent
We now have all the tools necessary to prove Lemma~\ref{lem:piecewise-fourier-moment}.
\begin{proof}[Proof of Lemma~\ref{lem:piecewise-fourier-moment}]
Note that a piecewise polynomial can be written as $\sum_{i = 1}^s p_i (x) \mathbf{1}_{I_i} (x)$, where $p_i$ is a degree $d$ polynomial and $ \mathbf{1}_{I_i}$ are indicator variables for intervals.
By linearity of the Fourier transform, it suffices to compute the Fourier moment of $\alpha_i (x) = p_i (x)  \mathbf{1}_{I_i} (x)$ for each $i = 1, \ldots, s$, and to do so, it suffices to compute $x^j  \mathbf{1}_{I_j} (x)$ for every monomial $j = 0, \ldots, d$.
By a change of variables, Lemma~\ref{lem:monomial} gives an algorithm that runs in time $O(d^2)$ to compute the $\ell$-th Fourier moment of $x^j  \mathbf{1}_{I_j} (x)$.
Thus we can compute the $\ell$-th Fourier moment of $\alpha_i (x)$ in time $O(d^3)$, and hence of the entire piecewise polynomial in time $O(s d^3)$, as claimed.
\end{proof} 


\section{Deferred Proofs}
\label{sec:defer}

 \subsection{Proof of Lemma~\ref{lem:moment_conc}}
 \label{app:moment_conc}

We first require the following inequality.

\begin{fact}[Rosenthal Bound, see e.g. Theorems 6.1 and 6.2 of \cite{pinelis1994optimum}]\label{fact:rosenthal}
 	Let $X_1,...,X_n$ be independent random variables for which $\E[X_i] = 0$ and $\E[|X_i|^t] < \infty$ for some $t \ge 2$. If we define $X = \frac{1}{n}\sum^n_{i=1}X_i$, then 
 	\begin{equation*}
 		\E[|X|^t] \le \frac{1}{n^t} \cdot \left[ C_1(t) \cdot \left( \sum_{i=1}^n \E[|X_i|^t] \right) + C_2(t) \cdot \left( \sum_{i=1}^n \E[X^2_i] \right)^{t/2} \right],
 	\end{equation*} 
 	where $C_1(t) = (c\gamma)^t$ and $C_2(t) = (c\sqrt{\gamma}e^{t/\gamma})^t$ for any $\gamma\in[1,t]$ and universal constant $c > 0$. In particular, we can take $\gamma$ for which $\gamma\ln\gamma = 2t$ to get $C_1(t) = C_2(t) = (c't/\ln(t))^t$ for some other universal constant $c'>0$.
 \end{fact}

We can apply this to get a moment bound on the deviation of the empirical $p$-th moment from the true $p$-th moment. Define the random variable $X = \frac{1}{N}\sum^N_{i=1}Z^p_i - \E_{Z\sim\calF}[Z^p]$, where recall that $\calF$ is a mixture of $k$ univariate Gaussians, and $Z_1,...,Z_N$ are $N$ draws from $\calF$.

 \begin{lemma}[Moment bound for empirical deviation of $p$-th moment]\label{lem:moment_bound}
 	There is an absolute constant $c'>0$ for which the following holds for any $p$. For all $t\in\N$ we have that 
 	\begin{equation*}
 		\E[X^t] \le \left(\frac{c'}{\sqrt{N}}\cdot\maxvar{\calF}^p\cdot p^{p/2}\cdot t^{p/2+1}\right)^t.
 	\end{equation*}
 	 Then for any $r,\gamma > 0$, we have that for $N = \left(\frac{\alpha}{\gamma r}\right)^2$, 
 	 \begin{equation*}
 	 	\Pr_{Z_1, \cdots, Z_N} \left[\left|\frac{1}{N}\sum^N_{i=1}Z_i^p - \E_{Z\sim\calF}[Z^p]\right| > r\right] \le \gamma^t.\label{eq:conc2}
 	 \end{equation*}
 \end{lemma}
\begin{proof}
 	For simplicity, we define
 	\begin{align*}
 	\sigma_{\max} = \sigma_{\max} ( \calF ) .
 	\end{align*}
 	For every $i\in[N]$, define the random variable $X_i \triangleq Z^p_i - \E[Z^p_i]$. To apply Fact~\ref{fact:rosenthal}, we must compute moments of $X_i$. First note that for any even $d \in \mathbb{N}$ and $Z \sim \calF$, 
 	\begin{equation*}
 	\E[Z^d] = \sum^k_{j=1} p_j \cdot \M_d( \N( 0,\sigma^2_j ) ) \le \sum^k_{j=1}p_j \cdot \sigma^d_j\cdot d^{d/2}.
 	\end{equation*} 
 	In particular, we have that $\E[Z^d] \le \smax^d \cdot d^{d/2}$. So for all $i\in[N]$ and even $t \in \mathbb{N}$, we get that \begin{align*}
 		\E[X^t_i] &= \E[(Z^p_i - \E[Z^p_i])^t] \\
 		&= \sum^t_{\ell = 0}(-1)^{\ell} \binom{t}{\ell} \cdot \E[Z^{p\ell}_i]\cdot \E[Z^p_i]^{t - \ell} \\
 		&\le \sum_{\ell\in[t] \ \text{even}}\binom{t}{\ell} \cdot (p\ell)^{p\ell/2}\smax^{p\ell} \cdot p^{(p/2)(t-\ell)} \smax^{p(t-\ell)} \\
 		&= p^{pt/2} \smax^{pt}\sum_{\ell\in[t] \ \text{even}} \binom{t}{\ell} \ell^{p\ell/2} \\
 		&\le (2t^{p/2} \cdot p^{p/2} \cdot \smax^p)^t.
 	\end{align*} 
 	where the third step follows from the fact that for any degree $d$, $\E[Z^d_i] = 0$ if $d$ is odd, and $\E[Z^d_i] \le \E_{g\sim\N(0,\smax)}[g^d] \le d^{d/2}\cdot\smax^d$ if $d$ is even; and the last step follows by naively upper bounding the terms $\ell^{p\ell/2}$ by $t^{pt/2}$.

 	In particular, 
 	\begin{equation*}
 		\frac{1}{N^t}\sum^N_{i=1}\E[X_i^t] \le \frac{1}{ N^{t-1} } \cdot \left( 2t^{p/2} \cdot p^{p/2} \cdot \smax^p \right)^t.
 	\end{equation*} 
 	For $\E[X_i^2]$, note that 
 	\begin{equation*}
 		\E[X_i^2] = \E[Z_i^{2p}] - \E[Z_i^p]^2 \le \E[Z_i^{2p}] \le \smax^{2p}\cdot (2p)^p,
 	\end{equation*} 
 	so 
 	\begin{equation*}
 		\frac{1}{N^t}\left(\sum^N_{i=1}\E[X^2_i]\right)^{t/2} = \frac{1}{N^{t/2}}(\smax^{p}\cdot (2p)^{p/2})^t.
 	\end{equation*} 
 	We conclude by Fact~\ref{fact:rosenthal} that the random variable $X$ satisfies the following moment bound:
 	\begin{align*}
 		\E[X^t] 
 		\le & ~ \frac{(c't/\ln(t))^t}{N^{t-1}} \cdot \left(2t^{p/2} \cdot p^{p/2} \cdot \smax^p \right)^t + \frac{(c't/\ln(t))^t}{N^{t/2}} \cdot \left( \smax^p \cdot (2p)^{p/2} \right)^t \\
 		= & ~ \smax^{pt}\cdot p^{pt/2}\cdot\left(\frac{(c't/\ln(t))^t}{N^{t-1}}\cdot 2^t t^{pt/2} + \frac{(c't/\ln(t))^t}{N^{t/2}} \cdot 2^{pt/2}\right) \\
 		\le & ~ \smax^{pt} \cdot p^{pt/2} \cdot \left(\frac{(t/\ln t)^t}{N^{t/2}} \cdot t^{pt/2}\right) \cdot (c')^t \\
 		\le & ~ \left( {c'\cdot \smax^{p}\cdot p^{p/2}\cdot t^{p/2 + 1}} / {\sqrt{N}}\right)^t
 	\end{align*} 
 	as claimed, where the third step follows from choosing $c'>1$ to be sufficiently large constant.
 \end{proof}

Finally, we need some standard facts about Orlicz norms.

\begin{definition}[Orlicz norms]
 	Let $\Psi:\R_{>0}\to\R_{>0}$ be a convex, increasing function satisfying $\Psi(0) = 0$ and $\Psi(x)\to\infty$. We call such a function $\Psi$ a \emph{Young function}. Let $X$ be a random variable over $\R_{>0}$. The \emph{Orlicz norm of $X$ with respect to $\Psi$} is defined by 
 	\begin{equation}\label{eq:orliczdef}
 		\norm{X}_{\Psi}\triangleq \inf\{c > 0: \E[\Psi(X/c)] \le 1\}.
 	\end{equation}
 \end{definition}

\begin{fact}[Sufficient condition for bound]\label{fact:orliczsuffice}
 	If $\alpha,\beta>0$ satisfies $\E[\Psi(X/\alpha)]\le \beta$, then $\norm{X}_{\Psi}\le \alpha\cdot \beta$.
 \end{fact}
 \begin{fact}[Tail bound given Orlicz norm bound]\label{fact:orlicztail}
 	Let $X$ be a random variable over $\R_{>0}$. If $\sigma = \norm{X}_{\Psi}<\infty$, then 
 	\begin{equation*}
 		\Pr[X\ge \beta\norm{X}_{\Psi}] \le 1/\Psi(\beta)
 	\end{equation*}
 \end{fact}

\begin{fact}[Approximation of $e^{x^{\alpha}}$ by Young function]
 	For any $0<\alpha<1$, the function $\Psi_{\alpha}$ given by 
 	\begin{equation*}
 		\Psi_{\alpha}(x)\triangleq \begin{cases}(\alpha e)^{1/\alpha}\cdot x & x < (1/\alpha)^{1/\alpha} \\
 		e^{x^{\alpha}} & x \ge (1/\alpha)^{1/\alpha}
 		\end{cases}
 	\end{equation*} 
 	is a Young function satisfying 
 	\begin{equation}\label{eq:youngexpbound}
 	\Psi_{\alpha}(x) \le e^{x^{\alpha}}.
 	\end{equation}
 \end{fact}

We can now complete the proof of Lemma~\ref{lem:moment_conc}.

\begin{proof}[Proof of Lemma~\ref{lem:moment_conc}]
 	Take $\alpha = \frac{1}{p+2}$. Note that 
 	\begin{align*}
 		\E[\Psi(X/c)] &\le \E[e^{(X/c)^{\alpha}}] \\
 		&= \sum^{\infty}_{t=0}\frac{1}{t!}\E[(X/c)^{\alpha t}] \\
 		&\le \sum^{\infty}_{t=0}\frac{c^{-\alpha t}}{t!}\E[X^t]^{\alpha} \\
 		&\le \sum^{\infty}_{t=0}\frac{1}{t!}\cdot\left( c^{-1} \cdot c'\cdot \maxvar{\calF}^{p}\cdot p^{p/2}\cdot t^{p/2 + 1} / \sqrt{N}\right)^{\alpha t} \\
 		&= \sum^{\infty}_{t=0}\frac{1}{t!} \cdot \left( c^{-1} \cdot c'\cdot \smax^{p}\cdot p^{p/2} / \sqrt{N} \right)^{\alpha t}t^{t/2},
 	\end{align*} 
 	where the third step follows by Jensen's and concavity of $x\mapsto x^{\alpha}$ when $0<\alpha < 1$, the fourth step follows by Lemma~\ref{lem:moment_bound}, and the last step follows by the fact that $\alpha(p/2 + 1) = 1/2$.

 	So if we take $c = c' \cdot \smax^p \cdot p^{p/2} / \sqrt{N} $, then $\E[\Psi(X/c)] = O(1)$, so we conclude by Fact~\ref{fact:orliczsuffice} that 
 	\begin{equation*}
 	\norm{X}_{\Psi} = c'' \cdot \smax^p \cdot p^{p/2} / \sqrt{N} 
 	\end{equation*} 
 	for some absolute constant $c'' > 0$. We can now apply Fact~\ref{fact:orlicztail} to get that 
 	\begin{equation*}
 		\Pr\left[X \ge \beta \cdot c'' \cdot \smax^p\cdot p^{p/2} / \sqrt{N} \right] \le 1/\Psi(\beta).
 	\end{equation*} 
 	If we take $\beta =  \gamma\sqrt{N} / ( c'' \cdot \smax^p\cdot p^{p/2} )$, then provided 
 	\begin{align*}
 	N \ge (c'')^2 \cdot \gamma^{-2} \cdot ( \max\{ p+2 , \ln(1/\delta) \} ) ^{2p+4} \cdot p^p \cdot \smax^{2p},
 	\end{align*}
 	we have that $\beta\ge (p+2)^{p+2}$ so that $1/\Psi(\beta)\le \exp(\beta^{1/(p+2)})$, and $\beta \ge (\ln (1/\delta) )^{p+2}$ so that $\exp( - \beta^{1/(p+2)} ) \le \delta$, so we get that 
 	\begin{equation*}
 		\Pr_{Z_1,...,Z_N}\left[ \Big| \frac{1}{N} \sum^N_{i=1} Z^p_i - \E_{Z\sim \calF}[Z^p] \Big| \le \gamma \cdot \smax^p \cdot p^{p/2} \right] \le \delta.
 	\end{equation*}

Finally, we would like to relate the deviation term $\gamma \cdot \smax^p \cdot p^{p/2}$ to $\beta \cdot \E_{Z \sim \calF}[Z^p]$. By \eqref{eq:expectedzp}, if we take $\gamma = \beta \cdot\pmin$, the lemma follows.
 \end{proof}

\subsection{Proof of Fact~\ref{fact:monotone}}

\begin{proof}
	Define the function $F(\sigma) \triangleq \int_{[-\tau,\tau]^c}\N(0,\sigma^2;x)\cdot x^p\, \d x$. It suffices to show that $F'(\sigma) > 0$ for all $\sigma\in(0,\sigma^*]$. We have that 
	\begin{align*}
		\deriv{\sigma}F(\sigma) 
		&= \int_{[-\tau,\tau]^c} \frac{e^{-x^2/(2\sigma^2)}(x^2 - \sigma^2)}{\sqrt{2\pi}\sigma^4}\cdot x^p \, \d x \\
		&= \frac{1}{\sigma^3} \cdot \left( \int_{[-\tau,\tau]^c}\N(0,\sigma^2 ; x)\cdot (x^{p+2} - \sigma^2 x^p)\, \d x \right) \\
		&= \left((p+1)!! - (p-1)!!\right)\sigma^p - \left(\int_{[-\tau,\tau]} \N(0,\sigma^2 ; x)\cdot (x^{p+2} - \sigma^2 x^p)\, \d x \right).
	\end{align*} 
	As the above expression tends to zero as $\tau\to\infty$, and because $\N(0,\sigma^2 ; x)\cdot(x^{p+2} - \sigma^2 x^p)$ is even, it suffices to show that the function $G(t)\triangleq \int^{\tau}_0 \N(0,\sigma^2 ; x)\cdot(x^{p+2} - \sigma^2 x^p) \, \d x $ is increasing in $\tau$. By the fundamental theorem of calculus, 
	\begin{equation*}
		G'(\tau) = \N(0,\sigma^2,\tau)\cdot (\tau^{p+2} - \sigma^2\tau^p) > 0,
	\end{equation*} 
	by the assumption that $\sigma < \sigma^* < \tau$.
\end{proof}

\subsection{Proof of Corollary~\ref{cor:unitcorr}}
\begin{proof}
	Let $c = 1/2 - \gamma$. Note that $\langle g,w\rangle \sim \N(0,1)$, so by Fact~\ref{fact:gaussian_tails} we have that for sufficiently large $d$,
	\begin{equation*}
	\Pr[\langle g,w\rangle \ge 1.1\underline{\alpha} d^c] \ge \frac{1}{\sqrt{2\pi}}\cdot\frac{1}{2.2\underline{\alpha} d^c}\cdot e^{-1.21\underline{\alpha}^2 d^{2c}/2}\ge 2e^{-\underline{\beta}\cdot d^{2c}}
	\end{equation*} for $\underline{\beta} = 1.21\underline{\alpha}^2$, and 
	\begin{equation*}
	\Pr[\langle g,w\rangle \le 0.9\overline{\alpha}d^{c}] \ge 1 - \frac{1}{\sqrt{2\pi}}\cdot\frac{1}{0.9\overline{\alpha}d^{c}}\cdot e^{-0.81\overline{\alpha}^2d^{2c}/2}\ge 1 - \frac{1}{2}\cdot e^{-\overline{\beta}\cdot d^{2c}}
	\end{equation*} for $\overline{\beta} = 0.81\overline{\alpha}^2/2$.
	On the other hand, by Fact~\ref{fact:thinshell}, $\Pr[\norm{g}_2 \in [0.9,1.1]\cdot\sqrt{d}]\ge 1 - 2e^{-c_{\text{shell}}d/100}$. By a union bound, we conclude that 
	\begin{equation*}
		\Pr\left[\langle v,w\rangle\ge d^{c-1/2}\right]\ge 2e^{-\underline{\beta}\cdot d^{2c}} - 2e^{-c_{\text{shell}}d/100} \ge e^{-\underline{\beta}\cdot d^{2c}},
	\end{equation*} 
	and similarly 
	\begin{equation*}
		\Pr\left[\langle v,w\rangle\le \overline{\beta}\cdot d^{c-1/2}\right]\ge 1 - 2e^{-c_{\text{shell}}d/100} - \frac{1}{2}e^{-\overline{\beta}\cdot d^{2c}} \ge 1 - e^{-\overline{\beta}\cdot d^{2c}}.
	\end{equation*}
\end{proof}

\subsection{Proof of Corollary~\ref{cor:joint}}
\begin{proof}
	Note that $\langle g,w_1\rangle$ and $\langle g,w_2\rangle$ are independent and distributed as $\N(0,1)$.

	Decompose $g \in \R^d$ as 
	\begin{equation}
		g = \langle g,w_1\rangle w_1 + \langle g,w_2\rangle w_2 + g^{\perp},
	\end{equation} 
	where $g^{\perp} \in \R^d$ is a standard Gaussian vector in the subspace orthogonal to $w_1,w_2 \in \R^d$. 

	We will first lower bound the probability of the event on the left-hand side of \eqref{eq:wantprobrelation}.

	By Fact~\ref{fact:thinshell}, we have that for some absolute constant $t > 0$, 
	\begin{equation*}
		\Pr\left[\norm{g^{\perp}}^2_2 = d \pm t\right] \ge 1/2.
	\end{equation*} 
	Call this event $\calE$. Let $\calE'$ be the event that $\langle g,w_2\rangle \le \sqrt{d+1}\cdot \alpha_2 d^{-1/2} = O(1)$. We know that $\Pr[\calE\wedge \calE'] \ge \Omega(1)$.

	Conditioning on $\calE$ and $\calE'$, first note that $\langle v,w_2\rangle \le \frac{1}{\sqrt{d+1}}\le \alpha_2 \cdot d^{-1/4}$. We also have that
	\begin{equation}
		\langle v,w_1\rangle \ge \frac{\langle g,w_1\rangle}{\sqrt{\langle g,w_1\rangle^2 + 1 + d + t}},
	\end{equation} 
	so if we take $\alpha' > \alpha_1$ to be the solution to 
	\begin{equation}
		\frac{\alpha' d^{1/4}}{\sqrt{\alpha'^2\sqrt{d} + 1 + d + t}} = \alpha_1 \cdot d^{-1/4},\label{eq:alphaprimedef}
	\end{equation} 
	we conclude that 
	\begin{equation*}
		\Pr\left[\left(\langle v, w_1\rangle \ge \alpha_1\cdot d^{-1/4}\right) \wedge \left(\langle v, w_2\rangle \le \alpha_2\cdot d^{-1/4}\right)\right] \ge \Omega(1) \cdot \Pr_{h\sim\N(0,1)}[h\ge \alpha'd^{1/4}] .
	\end{equation*} 
	Furthermore, squaring both sides of \eqref{eq:alphaprimedef} and rearranging, we see that 
	\begin{equation*}
		\alpha'^2_1 - \alpha^2_1 = \frac{\alpha^2_1 \alpha'^2}{\sqrt{d}} + \frac{\alpha^2_1}{d(1+t)} = O(1/\sqrt{d}),
	\end{equation*} 
	so in particular $\Pr[h\ge\alpha' d^{1/4}] \ge \frac{1}{\poly(d)} \cdot \Pr[h\ge \alpha_1 d^{1/4}]$.

	We next upper bound the probability of the event on the right-hand side of \eqref{eq:wantprobrelation}. 

	Write $g$ as $g = \langle g, w_1\rangle w_1 + g'^{\perp}$ for $g'^{\perp}$ a standard Gaussian vector orthogonal to $w_1$. Then the event on the right-hand side of \eqref{eq:wantprobrelation} is the event that $\langle g,w_1\rangle \ge \alpha_1 d^{-1/4} \norm{g}_2$, or equivalently, that 
	\begin{equation*}
		 \langle g,w_1\rangle\ge \frac{\alpha_1 d^{-1/4}}{\sqrt{1 - \alpha_1^2 d^{-1/2}}}\norm{g'^{\perp}}_2.
	\end{equation*} 
	Let $\alpha'' > \alpha_1$ be the solution to 
	\begin{equation}
		\frac{\alpha_1 d^{-1/4}}{\sqrt{1 - \alpha^2_1 d^{-1/2}}} = \alpha''\cdot d^{-1/4}.\label{eq:alphaprimeprimedef}
	\end{equation} 
	Then the above event has probability given by the integral 
	\begin{equation}
		\int^{\infty}_{0}\Pr_{h\sim\N(0,1)}[h\ge \alpha'' d^{-1/4}\cdot \beta^{1/2}]\cdot \mu(\beta) \, \d\beta,\label{eq:integral}
	\end{equation} 
	where $\mu(\beta)$ is the density of the random variable $\norm{g'^{\perp}}^2_2$. By Fact~\ref{fact:gaussian_tails}, 
	\begin{equation*}
		\Pr_h[h \ge \alpha'' d^{-1/4}\cdot\beta^{1/2}] \le \frac{d^{1/4}}{\alpha''\beta^{1/2}}\cdot e^{-\frac{1}{2}\alpha''^2\beta/\sqrt{d}}.
	\end{equation*} 
	For $\beta\in[0.9d,1.1d]$, this quantity is at most $1/\poly(d)\cdot e^{-\frac{1}{2}\alpha''^2\beta/\sqrt{d}}$. So we may write \eqref{eq:integral} as 
	\begin{align*}
		& ~\int_{[0.9d,1.1d]}\Pr_{h\sim\N(0,1)}[h\ge \alpha'' d^{-1/4}\cdot \beta^{1/2}]\cdot \mu(\beta) \, \d\beta \\
		& ~ + \int_{[0.9d,1.1d]^c}\Pr_{h\sim\N(0,1)}[h\ge \alpha'' d^{-1/4}\cdot \beta^{1/2}]\cdot \mu(\beta) \, \d\beta \\
		\le & ~ \frac{1}{\poly(d)}\int^{\infty}_0 e^{-\frac{1}{2}\alpha''^2\beta/\sqrt{d}}\cdot \mu(\beta) \, \d\beta \\
		& ~ + \exp(-\Omega(d)) \\
		= & ~ \frac{1}{\poly(d)}\cdot\frac{1}{(2\pi)^{(d-1)/2}}\int e^{-\frac{g^2_1 + \cdots + g^2_{d-1}}{2}\cdot \left(1 + \alpha''^2/\sqrt{d}\right)}\d g_1 \cdots \d g_{d-1} \\
		= & ~ \frac{1}{\poly(d)}\cdot (1 + \alpha''^2/\sqrt{d})^{-(d-1)/2}.
	\end{align*} 
	Finally, we observe that 
	\begin{align*}
		(1 + \alpha''^2/\sqrt{d})^{-(d-1)/2} 
		= & ~ \left((1 + \alpha''^2/\sqrt{d})^{\sqrt{d}/\alpha''^2}\right)^{-\frac{d-1}{\sqrt{d}}\alpha''^2/2} \\
		\le & ~ \left(e - O(1/\sqrt{d})\right)^{-\frac{d-1}{\sqrt{d}}\alpha''^2/2} \\
		\le & ~ O(e^{-\sqrt{d}\alpha''^2/2}).
	\end{align*} 
	We are done by \eqref{fact:gaussian_tails} if we can show that $\Pr[h\ge \alpha'' d^{1/4}]\le\poly(d)\Pr[h\ge\alpha d^{1/4}]$. But by squaring both sides of \eqref{eq:alphaprimeprimedef} and rearranging, we see that 
	\begin{equation*}
		\alpha''^2 - \alpha^2_1 = \frac{\alpha''^2 \alpha_1^2}{\sqrt{d}} = O(1/\sqrt{d}),
	\end{equation*} 
	so in particular $\Pr[h\ge \alpha'' d^{1/4}] \le \poly(d) \cdot \Pr[h\ge\alpha d^{1/4}]$ as desired.
\end{proof}

\subsection{Proof of Lemma~\ref{lem:findspan}}
\label{app:findspan}

\begin{proof}
Suppose $\calD$ has parameters $(\{p_i\}),\{w_i\})$. For the first part of the lemma, first assume that the noise rate $\noise = 0$ so that with probability $p_i$, $y = \langle w_i,x\rangle$. For entry $(j,j')\in[d]^2$, we may write
\begin{align*}
	2\E[\vec{M}^{x,y}_a]_{j,j'} = & ~ \sum^k_{i=1}p_i \E_{x\sim\N(0,\Id)}\left[\langle w_i - a,x\rangle^2\cdot x_jx_{j'} - \bone{j=j'}\cdot\langle w_i - a,x\rangle^2\right] \\
	= & ~ \sum^k_{i=1}p_i \E_{x\sim\N(0,\Id)}\left[\langle w_i - a,x\rangle^2\cdot x_jx_{j'}\right] - \bone{j=j'} \cdot \sum^k_{i=1}p_i\norm{w_i - a}^2_2 \\
	= & ~ 
	\begin{cases}
	\overset{k}{ \underset{i=1}{\sum} } p_i \left((w_i - a)^2_j\E[x_j^4] + \underset{ \ell \neq j }{ \sum } (w_i - a)^2_{\ell'}\E[x_j^2x_{\ell}^2]\right) - \overset{k}{ \underset{i=1}{\sum} } p_i\norm{w_i - a}^2_2 & \mathrm{~if~} j = j'\\
	\overset{k}{ \underset{i=1}{\sum} } p_i\left(2(w_i - a)_j(w_i - a)_{j'}\E[x^2_jx^2_{j'}]\right) & \mathrm{~if~} j\neq j'
	\end{cases} \\
	= & ~  \begin{cases}\overset{k}{ \underset{i=1}{\sum} } p_i \left(3(w_i - a)^2_j + \underset{\ell\neq j}{\sum} (w_i - a)^2_{\ell'}\right) - \overset{k}{ \underset{i=1}{\sum} } p_i\norm{w_i - a}^2_2 & \mathrm{~if~} j = j'\\
	\overset{k}{ \underset{i=1}{\sum} } p_i\left(2(w_i - a)_j(w_i - a)_{j'}\right) & \mathrm{~if~} j\neq j'\end{cases} \\
	= & ~ 2\sum^k_{i=1}p_i (w_i - a)_j(w_i - a)_{j'},
\end{align*} 
as claimed. Now if the noise rate $\noise$ is nonzero so that with probability $p_i$, let $y'$ be the random variable which equals $\langle w_i,x\rangle$ with probability $p_i$, so that $y = y' + g$ for $g\sim\N(0,\noise^2)$, then \begin{align*}
	2\E[\vec{M}^{x,y}_a] &= \E\left[(y' - \langle a,x\rangle + g)^2 xx^{\top} - (y' + g)^2\cdot\Id\right] \\
	&= 2\sum^k_{i=1}p_i (w_i - a)(w_i - a)^{\top} + \E[g^2\cdot xx^{\top}] - \E[g^2]\cdot\Id  \\
	&= 2\sum^k_{i=1}p_i (w_i - a)(w_i - a)^{\top},
\end{align*} where the second step follow by the fact that $g$ is independent of the random variables $x,y'$.

The second part of the lemma follows from the following fact, which quantifies the extent to which the matrix $\vec{M}^{x,y}_a$ concentrates in spectral norm.
This is already proven in the noiseless case, see e.g. Eq.~(34) in \cite{yi2016solving}, and the noisy version follows from a straightforward modification of that proof using Theorem 4.7.1 in~\cite{vershynin2018high}.

\begin{fact}[Concentration of Empirical Moments]
	\begin{equation*}
		\Pr\left[\norm{\frac{1}{N}\sum^N_{i=1}\vec{M}^{x_i,y_i}_a - \E_{(x,y)\sim\calD}\left[\vec{M}^{x,y}_a\right]}_2 \ge \Omega\left( \max_{i \in [k]} \norm{w_i - a}^2_2\cdot\frac{\ln\left(\pmin N\right)}{\sqrt{\pmin N}}\cdot \sqrt{d\ln(k/\delta)}\right)\right] \le \delta.
	\end{equation*}
\end{fact}
\end{proof}

\subsection{Proof of Lemma~\ref{lem:runtime}}
\label{app:runtime}

\begin{proof}
We bound the sample complexity and runtime of each of the $O(MT)$ iterations.
In each iteration, we first sample $N_1$ points, and perform an approximate $k$-SVD on a $N_1 \times d$ matrix, where $N_1$ is defined as in Line~\ref{line:n1}.
By Corollary~\ref{cor:constant_min_approx}, $\underline{\sigma}^{\text{sharp}}_t$ is at most a constant factor smaller than $\underline{\sigma} = \Omega(\epsilon)$.
Therefore the sample complexity of this step is at most 
\[
N_1 = \widetilde{O} ( \epsilon^{-2} \pmin^{-2}  d k^2 \ln(1 / \delta) )\; .
\]
and the runtime is at most $\widetilde{O} (N_1 k d)$ by Lemma~\ref{lem:correlation}.
The other contribution to the sample complexity and runtime of each iteration (at least in most regimes) is from \textsc{CompareMinVariances}.
By our choice of parameters and Corollary~\ref{cor:compare_min_variances}, the sample complexity of $\textsc{CompareMinVariances}$ is 
\[
N =   \pmin^{-4} k \ln ( 1 / \delta ) \cdot \poly \left(\sqrt{k}, \ln ( 1 / \pmin ), \ln (1 / \epsilon ) \right)^{O\left(\sqrt{k} \ln (1 / \pmin)\right)} \; .
\]
and the runtime is bounded by $\widetilde{O} (N)$.
Since we run for $MT = \widetilde{O} \left( \sqrt{k} e^{\sqrt{k}} \ln(1 / \epsilon) \right)$ iterations, this completes the proof.

\end{proof}

\subsection{Proof of Lemma~\ref{lem:runtime_hyper}}

\begin{proof}
Prior to the outer loop, we first sample $N_1$ points, and perform an approximate $k$-SVD on a $N_1 \times d$ matrix, where $N_1$ is defined as in Line~\ref{line:n1}.
 
Therefore the sample complexity of this step is at most 
\[
N_1 = \widetilde{O} \left( d \cdot \poly(k)/\pmin^2\right) \; .
\]
and the runtime is at most $\widetilde{O} (N_1 k d)$ by Lemma~\ref{lem:correlation}.

The bulk of the contribution to the sample complexity and runtime comes from the $S$ iterations of the outer loop, each of which consists of $MT$ iterations of the inner loop (over $t$) and a call to \textsc{CheckOutcomeHyperplanes}. The complexity of these $MT$ iterations is dominated by an invocation of \textsc{CompareMinVariances}.
By our choice of parameters and Corollary~\ref{cor:compare_min_variances}, the sample complexity of one run of $\textsc{CompareMinVariances}$ is 
\[
N =  \pmin^{-4} k \cdot \poly \left( k^{3/5}, \ln ( 1 / \pmin ), \ln (1 / \epsilon ) \right)^{O\left(k^{3/5} \ln (1 / \pmin)\right)} 
\]
and the runtime is bounded by $\widetilde{O} (N)$.
Each iteration $i\in[S]$ involves $M\cdot T = \widetilde{O} \left( k^{3/5} e^{k^{3/5}} \ln (1 / \epsilon) \right)$ such iterations. Additionally, the $i$-th iteration runs \textsc{CheckOutcomeHyperplanes}, a run of which has time and sample complexity
\[
N_2 = O\left(\pmin^{-2}\cdot \ln(2S/\delta)\right) = O\left(\pmin^{-2}\cdot k^{3/5}\ln(2/\delta)\right).
\] We conclude that \textsc{HyperplaneMomentDescent} requires sample complexity 
\begin{align*}
\widetilde{O}\left(N_1 + S\cdot\left(k^{3/5}e^{k^{3/5}}N + N_2\right)\right)
\end{align*}
and runs in time 
\begin{align*}
\widetilde{O}\left(dN_1 + S\cdot\left(k^{3/5}e^{k^{3/5}}N + N_2\right)\right).
\end{align*}
\end{proof}


\end{document}